\documentclass[mnsc,nonblindrev]{optonline} 
\OneAndAHalfSpacedXI 

\newcount\Comments
\newcommand{\kibitz}[2]{\ifnum\Comments=1{\textcolor{#1}{\textsf{\footnotesize #2}}}\fi}

\usepackage{color, bm, enumitem, mathtools}
\usepackage{graphicx, transparent}
\usepackage{comment}
\usepackage{soul}
\usepackage{physics}
\usepackage{dsfont}

\graphicspath{{Figures/}}
\usepackage[dvipsnames]{xcolor}

\usepackage{natbib}
\bibpunct[, ]{(}{)}{,}{a}{}{,}%
\def\bibfont{\small}%

\TheoremsNumberedThrough     
\ECRepeatTheorems
\EquationsNumberedThrough

\DeclareMathOperator{\minority}{\mathcal{I}^{\text{min}}}
\DeclareMathOperator{\majority}{\mathcal{I}^{\text{maj}}}
\DeclareMathOperator{\all}{\mathcal{I}}

\renewcommand{\qed}{\hfill \ensuremath{\Box}}

\usepackage{algpseudocode}
\usepackage{multirow}
\usepackage{array}
\usepackage{subcaption}
\usepackage{caption}
\usepackage{titlesec}
\usepackage{xurl}
\usepackage{lmodern}
\usepackage{fix-cm}

\usepackage{hyperref}
\usepackage{cleveref}
\hypersetup{
colorlinks,
breaklinks,
linkcolor={[RGB]{174,83,139}},
citecolor={[RGB]{97,62,173}},
urlcolor={[RGB]{115,135,195}}
}

\TITLE{Equitable Auction Design with Provable Regret Guarantees\footnote{An earlier iteration of this paper was disseminated under a different title (`Equitable Auction Design: With and Without Distributions'). The current manuscript substantially extends and supersedes that preliminary version.}}
\RUNTITLE{Equitable Auction Design with Provable Regret Guarantees}
\RUNAUTHOR{Wang, Ko{\c{c}}yi{\u{g}}it, and Rujeerapaiboon}
\ARTICLEAUTHORS{
\AUTHOR{Ruiqin Wang$^1$, {\c{C}}a{\u{g}}{\i}l Ko{\c{c}}yi{\u{g}}it$^2$, Napat Rujeerapaiboon$^3$}

\AFF{
$^1$\textit{School of Information Management and Engineering, Shanghai University of Finance and Economics, China}\\
$^2$\textit{Luxembourg Centre for Logistics and Supply Chain Management, University of Luxembourg, Luxembourg}\\
$^3$\textit{Department of Industrial Systems Engineering and Management, National University of Singapore, Singapore}\\
\EMAIL{wangruiqin@mail.shufe.edu.cn}, \EMAIL{cagil.kocyigit@uni.lu}, \EMAIL{napat.rujeerapaiboon@nus.edu.sg}
}
}

\begin{document}

\ABSTRACT{\indent We study a mechanism design problem where a seller aims to allocate a good to multiple bidders, each with
a private value. The seller supports or favors a specific group, referred to as the minority group. Specifically,
the seller requires that allocations to the minority group are at least a predetermined fraction (equity level)
of those made to the rest of the bidders. Such constraints arise in various settings, including government
procurement and corporate supply chain policies that prioritize small businesses, environmentally responsible suppliers, or enterprises owned by historically disadvantaged individuals. We include the equity requirements as constraints in the design of a mechanism, and our goal is to characterize equitable mechanisms with provable performance guarantees. Motivated by settings where reliable information about the bidders' values is unavailable, our main focus is a regret-based mechanism design problem, which makes no assumptions about the distribution of the bidders' values and aims to minimize the maximum (worst-case) ex-post regret in view of all possible bidders' values. Here, regret is defined as the difference between the highest revenue achievable in hindsight and the revenue generated by a mechanism. For this problem, we propose a closed-form mechanism and prove that its ex-post regret is at most a constant multiple (dependent on the equity level) of the optimal worst-case regret. We find that, when bidders' value supports are identical, this factor is less than 1.31 across all possible equity levels. We also show that the proposed mechanism is asymptotically optimal as the equity level decreases and increases. Additionally, we show that this mechanism is also equipped with an additional layer of pricing protection for the minority group. As a benchmark based on distributional information, we also study a stochastic mechanism design problem, which assumes that the bidders' values are random variables following a known distribution and aims to maximize the expected revenue. We characterize the optimal mechanism in this benchmark problem under standard independence and regularity assumptions. Both mechanisms can be interpreted as set-asides, a common policy tool that reserves a fraction of goods for minority groups. Numerical results demonstrate that the stochastic mechanism performs well when the bidders' value distribution is accurately estimated, while the regret-based mechanism exhibits greater robustness under estimation errors.}

\maketitle

\section{Introduction}
Auctions are widely used in practice to sell a diverse range of items, including housing, financial instruments, and commodities. Key reasons for employing auctions include the limited availability of the items being offered and limited knowledge of demand, which necessitate a strategic allocation mechanism to achieve objectives such as revenue maximization, welfare maximization or regret minimization. In addition to these objectives, fairness and equity in allocations are important considerations in certain contexts, particularly when the goal is to support or favor a specific group, referred to here as the \textit{minority} group. In fact, governments, in both procurement and allocation contexts, often have objectives to favor certain groups such as small businesses or enterprises owned by historically disadvantaged individuals. For instance, in the ``Executive Order on Further Advancing Racial Equity and Support for Underserved Communities Through the Federal Government," issued by the former president of the United States, Joseph R. Biden, the importance of equitable procurement is emphasized. Section~7 of this order set a government-wide goal for federal procurement dollars, aiming for~15 percent to be awarded to small businesses owned by socially and economically disadvantaged individuals by Fiscal Year 2025 \citep{federalregister2023}. Another example is corporate procurement policies, where firms may set targets to source from environmentally responsible suppliers, small businesses, or suppliers owned by historically disadvantaged individuals to promote sustainability, innovation, and competition.

Motivated by the importance of incorporating fairness and equity into allocation mechanisms, in this paper, we study a mechanism design problem with equity constraints. In this problem, a seller aims to sell an item to multiple bidders, each with a privately known willingness-to-pay (value) for the item. The seller seeks to allocate a certain proportion of the item to the minority group, introducing equity constraints on allocation decisions. We formulate these equity constraints to ensure that they hold ex-post, i.e., after the bids have been revealed. 
We consider two variants of this mechanism design problem, each with a different objective: 
\emph{regret-based mechanism design}, which makes no assumptions about the distribution of the bidders' values and aims to minimize the maximum (worst-case) ex-post regret in view of all possible bidders' values, where regret is defined as the difference between the highest revenue achievable in hindsight and the revenue generated by a mechanism,
and \emph{stochastic mechanism design}, which assumes that the bidders' values are random variables following a known distribution and aims to maximize the expected revenue (when the distribution is known, this objective coincides with minimizing the expected regret). 
While the stochastic formulation is appropriate in settings where reliable distributional information is available, obtaining accurate estimates of bidders’ value distributions can be challenging in many real-world applications. For example, government auctions often occur infrequently, resulting in limited available data. This motivates the study of the regret-based formulation, which does not rely on distributional assumptions and remains applicable even when little information about bidders’ values is available.
% {color{blue}
% While the second problem is relevant in scenarios involving large amounts of data or information, accurately estimating the distribution of bidders' values can be challenging in many real-world settings. For instance, government auctions occur infrequently, resulting in limited available data. In such cases, the first problem becomes more practically relevant, as it does not rely on knowledge of the bidders' value distributions.
% \noterq{A minor comment: this phrasing may give the impression of ``why do we even need the stochastic case?''. Let me think about how to revise it (the contribution part below has a similar issue).}
% }

The contributions of this paper are summarized as follows:
\begin{itemize}[noitemsep]
    \item %Recognizing that in many practical situations, the seller is unlikely to know the underlying distribution of bidders' values or estimate it accurately, 
    Motivated by settings in which reliable distributional information is unavailable or difficult to estimate, 
    we formulate the regret-based mechanism design problem with the objective of minimizing the worst-case regret. This approach neither relies on distributional assumptions nor requires knowledge of the bidders' value distributions. 
    To address this problem, we propose a mechanism and prove that its ex-post regret is at most a factor--depending on the equity level--of the optimal worst-case regret. %This multiple depends on the level of equity targeted. 
    We numerically find that, when bidders’ value supports are identical, this factor is less than~$1.31$ across all possible equity levels. %
    We also show that the proposed mechanism is asymptotically optimal as the equity level decreases and increases.
    A further finding is that the proposed mechanism provides the minority group with an additional layer of pricing protection, in that its ratio of payment to allocation is no worse than that of its majority counterpart.
    % While the fairness guarantee is formally defined at the allocation level, the proposed mechanism also ensures equitable pricing for the minority group. 
    %\notenr{I want to avoid pairing fairness with pricing (i.e. I think it is better to link fairness exclusively to allocation) to avoid discussion, so here's my suggestion: While allocation fairness is the primary concern and motivation of this study, we also find, somewhat serendipitously, that the minority group benefits from an additional layer of pricing protection under the proposed mechanism. Specifically, the ratio of payment to allocation for the minority group is no worse than that of its majority counterpart.}
    From a practical perspective, the proposed mechanism can be interpreted as a set-aside, a method commonly used in practice where a certain fraction of the item is reserved for allocation to the minority group; see, e.g.,~\cite{athey2013set}. The fraction set aside for the minority group increases with and is uniquely determined by the equity level. 
    \item %For the stochastic mechanism design problem, we
    As a benchmark based on distributional information, we also study stochastic mechanism design
    and characterize the optimal mechanism in closed form for any equity level desired under a standard assumption on the bidders' values—specifically, that they are mutually independent and their distributions are regular. 
    Similar to the proposed regret-based mechanism, there is also a set-aside implementation corresponding to this optimal stochastic mechanism.  
    \item We compare the two proposed mechanisms with additional benchmarks inspired by the literature, which can only be computed numerically, in a stress test experiment. 
    Specifically, we derive the regret-based mechanism, as well as the optimal stochastic mechanism based on the in-sample distribution, and then assess the distributions of the revenues and regrets they generate under the out-of-sample distribution.
    Our results suggest that the optimal stochastic mechanism should be preferred when the two distributions are similar, while the regret-based mechanism performs better in other cases. 
\end{itemize}
Consequently, we provide a closed-form, easy-to-implement mechanism with a provable performance guarantee, regardless of the trustworthiness of the bidders' value distribution.

\textbf{Related Literature.} Our work contributes to the mechanism design literature on fairness. The literature studies various notions of fairness, including envy-freeness, where agents prefer their own allocations over any allocation given to other agents (see, e.g., \cite{barman2018finding}); 
bidder-level nondiscrimination, where the bidders with the same bids are treated identically (see, e.g., \citealt{chen2025optimal});
and objectives that balance seller's revenue and bidders' surplus to promote long-term engagement (see, e.g., \citealt{babaioff2026paretoefficientmultibuyermechanismscharacterization}). Research on fair allocations in the context of auction design with consideration of group fairness (i.e., favoring a minority group), which is the focus of our paper, remains relatively limited. The closest works to ours in this line of research are \cite{pai2012auction} and \cite{fallah2024fair}. 

\citet{pai2012auction} consider the mechanism design problem of a seller offering a single item to multiple buyers, subject to an equity constraint that requires buyers from a target group to win the item with an ex-ante probability of at least a given threshold, while maximizing efficiency. \citet{fallah2024fair} study a similar problem in a dynamic setting where the seller interacts with two groups of buyers over multiple rounds. Their objective is to maximize discounted revenue while ensuring an equity constraint that requires the discounted average of items allocated to a target group up to any particular period, combined with the expected discounted allocation in future rounds, to exceed a given threshold. In both of these works, the underlying distribution of the buyers' values is assumed to be known to the seller, and the equity constraints are enforced in expectation based on this distribution. In contrast, our paper formulates the equity constraints to hold ex-post, i.e., after the bids are revealed. Thus, while our notion of equity is similar in spirit to those studied in \cite{pai2012auction} and \cite{fallah2024fair}, it does not rely on distributional assumptions or information and is guaranteed to hold irrespective of the underlying distribution. Additionally, we study the regret-based mechanism design problem, which does not require knowledge of the bidders' value distributions, even in the objective. We therefore adopt a robust approach compared to these papers, both in terms of the objective (in the case of regret-based mechanism design) and the equity requirements (in both stochastic and regret-based mechanism design). Our robust approach not only reduces reliance on distributional knowledge but also simplifies the theoretical analysis, enabling us to characterize (approximately) optimal mechanisms in closed form.

Furthermore, our work contributes to the (distributionally) robust mechanism design literature (see, e.g., \cite{anunrojwong2023robust, chen2024screening, wang2024power, giannakopoulos2023robust, rujeerapaiboon2023target, koccyiugit2024regret} and references therein), which assumes that only limited or no information about the bidders' value distributions is available and evaluates mechanisms based on their performance under the most adverse distribution consistent with the available information. To the best of our knowledge, we are the first to study the robust auction design setting with equity constraints.

In addition, our work is related to other fair resource allocation problems, such as 
scheduling for scarce resources (see, e.g., \cite{vardi2024price}) and
policy design for resource allocation based on contextual information (see, e.g., \cite{jo2023fairness, tang2023learning, freund2023group, bansak2024outcome} and references therein). In particular, the equity notion we study in our paper resembles the concepts of minority prioritization and statistical parity (a.k.a., group fairness) in allocations, which should either be similar across different groups or favor minority groups, as studied in the respective literature.

\textbf{Organization.} The paper is structured as follows. Section~\ref{sec: Problem Formulation} introduces the problem formulation and preliminaries. Sections~\ref{sec: Regret-Based Mechanism Design} and \ref{sec: stochastic} study the regret-based and stochastic mechanism design problems, respectively. Section~\ref{sec: Numerical Experiment} presents the numerical experiments. Omitted proofs can be found in Appendix~\ref{appendix: proofs}, and additional experimental results are provided in Appendix~\ref{appendix: experiment}.

\textbf{Notation.} For any vector $\bm{v} \in \mathbb{R}^I$, we use $v_i$ to denote its $i^{\text{th}}$ component and $\bm v_{-i}$ to represent the subvector of $\bm v$ obtained by excluding $v_i$. Random vectors are indicated with a tilde (e.g., $\tilde{\bm{v}}$), while their realizations are denoted by the same symbols without the tilde (e.g., $\bm{v}$). The collection of all bounded Borel-measurable functions from a Borel set $\mathcal{A}$ to another Borel set $\mathcal{C}$ is denoted by $\mathcal{L}(\mathcal{A}, \mathcal{C})$. 
The natural logarithm is written as $\log(\cdot)$, and $\mathbb{I}_{\mathcal{E}}$ denotes the indicator function of the event~$\mathcal{E}$.
% Throughout, 
Finally, we use \textit{increasing} and \textit{non-decreasing} to mean strictly increasing and weakly increasing, respectively.

\section{Problem Formulation}\label{sec: Problem Formulation}
We consider the auction design problem of a seller who wishes to sell a single item to $I \geq 2$ bidders. Let $\mathcal{I} = \{1, \dots, I\}$ represent the set of bidders. The bidders are categorized into two groups: minority and majority. 
Minority bidders, referring to a target group such as historically disadvantaged groups, may outnumber majority bidders. These classifications may be based on protected and observable features such as race, gender, age, etc. We use $\mathcal{I}^{{\operatorname{min}}} = \{1, \dots, I^{\operatorname{min}}\}$ and $\mathcal{I}^{{\operatorname{maj}}} = \{I^{\operatorname{min}} + 1, \dots, I^{\operatorname{min}} + I^{\operatorname{maj}}\}$ to represent the sets of minority and majority bidders, respectively. 
The overall set $\mathcal{I}$ of bidders can be expressed as $\mathcal{I} = \mathcal{I}^{{\operatorname{min}}} \cup \mathcal{I}^{{\operatorname{maj}}}$. Each bidder $i \in \mathcal{I}$ assigns a value $v_i \in \mathcal V_i = [0,\overline{v}_i]$, $\overline{v}_i \in \mathbb R_{++}$, to the item, which is unknown to the seller and other bidders.  
We let $\bm v = (v_1, \dots, v_I)$ be the vector of all bidders' values and denote by $\mathcal V = \times_{i \in \mathcal I} \mathcal V_i$ the set of all bidders' potential values. We do not make any distributional assumptions about the bidders' values and assume that the only information available to the seller is the set $\mathcal V$. For notational convenience, we define
${\overline{v}^{{\operatorname{min}}}} = \max_{i \in \mathcal{I}^{\operatorname{min}}} \overline{v}_i$ and ${\overline{v}^{{\operatorname{maj}}}} = \max_{i \in \mathcal{I}^{\operatorname{maj}}} \overline{v}_i$, which denote the maximum possible value within the minority group and the majority group, respectively.

By invoking the Revelation Principle, we focus on truthful direct mechanisms under which bidders choose to bid their true values \citep{myerson1981optimal}. Formally, we define an auction as a mechanism~$(\bm q, \bm m)$, comprising an allocation rule $\bm q \in \mathcal L(\mathcal V, [0,1]^I)$ and a payment rule $\bm m \in \mathcal L(\mathcal V, \mathbb R^I)$, that satisfies the following incentive compatibility \eqref{eq:IC}, individual rationality \eqref{eq:IR} and allocation feasibility~\eqref{eq:AF} constraints. 
\begin{align}
        &q_i(v_i,\bm{v}_{-i}) v_i  - m_i(v_i, \bm{v}_{-i})  \geq  q_i(w_i,\bm{v}_{-i}) v_i - m_i(w_i, \bm{v}_{-i}) \;\;\; \forall i \in \mathcal I, \forall \bm v \in \mathcal V, \forall w_i \in \mathcal V_i \tag{IC}\label{eq:IC}\\
&q_i(v_i,\bm{v}_{-i}) v_i  - m_i(v_i, \bm{v}_{-i}) \geq 0 \;\;\; \forall i \in \mathcal I, \forall \bm v \in \mathcal V \tag{IR}\label{eq:IR}\\
&\sum_{i \in \mathcal I} q_i(\bm v) \leq 1 \;\;\; \forall \bm v \in \mathcal V \tag{AF}\label{eq:AF}
\end{align}
Given that the bidders report their values as $\bm v \in \mathcal V$, $q_i(\bm v)$ represents 
the proportion of the item allocated to bidder $i \in \mathcal I$\footnote{In the case of an indivisible item, $q_i(\bm v)$ can represent the probability of allocating the item to bidder $i \in \mathcal I$.}, and $m_i(\bm v)$ represents the payment made by the same bidder.
The incentive compatibility \eqref{eq:IC} constraint ensures that every bidder has no incentive to misreport their value, the individual rationality \eqref{eq:IR} constraint guarantees that every bidder's utility cannot be negative when bidding truthfully, and the allocation feasibility \eqref{eq:AF} constraint assures that the allocated amount does not exceed $1$.

We will require that the allocation rule $\bm q$ satisfies the following equity \eqref{eq:Eq} constraint, where~$\gamma > 0$ quantifies the level of equity. 
\begin{align}
    &\sum_{i \in \mathcal I^{{\operatorname{min}}}} q_i(\bm v) \geq \gamma \sum_{i \in \mathcal I^{{\operatorname{maj}}}} q_i(\bm v) \;\;\;\forall \bm v \in \mathcal V \tag{Eq}\label{eq:Eq} 
\end{align}
This constraint ensures that the proportion allocated to the minority group is at least as high as $\gamma$ times the one for the majority group, thereby ensuring equitable opportunities for the minority group. This principle is inspired by the approach of minority prioritization and statistical parity in allocations, which are employed, for example, in policy design for resource allocation based on contextual information \citep{tang2023learning, jo2023fairness}. We exclude the case where $\gamma = 0$, for which~\eqref{eq:Eq} is automatically satisfied and there is no need to distinguish between majority and minority bidders.

\begin{remark} 
    In his ``Executive Order on Further Advancing Racial Equity and Support for Undeserved Communities Through The Federal Government," the former president of the United States, Joseph R. Biden,  emphasized the importance of equitable procurement. Section~7 of this order set a government-wide goal for federal procurement dollars, aiming for 15 percent to be awarded to small businesses owned by socially and economically disadvantaged individuals in Fiscal Year 2025. 
    Our formulation of equity \eqref{eq:Eq} can articulate targets of this nature. In fact, consider the alternative equity constraint
\begin{equation*}
\sum_{i \in \mathcal I^{{\operatorname{min}}}} q_i(\bm v) \geq \delta \sum_{i \in \mathcal I} q_i(\bm v),
\end{equation*}
where $\delta \in (0, 1)$. This constraint requires that $\delta$ proportion of allocations be directed to bidders from the minority group.  Note that this constraint can be written in the form of \eqref{eq:Eq} as follows:
    \begin{equation*}
        \begin{aligned}
            \sum_{i \in \mathcal I^{{\operatorname{min}}}} q_i(\bm v) \geq \delta \sum_{i \in \majority} q_i(\bm v) + \delta \sum_{i \in \minority} q_i(\bm v)
            \iff \sum_{i \in \mathcal I^{{\operatorname{min}}}} q_i(\bm v) \geq \frac{\delta}{1-\delta} \sum_{i \in \majority} q_i(\bm v). 
        \end{aligned}
    \end{equation*}\qed
\end{remark}
\begin{remark}
Our formulation of equity \eqref{eq:Eq} also covers equity defined in terms of average allocations. Consider the alternative equity constraint
\begin{equation*}
    \frac{1}
    {I^{\operatorname{min}}}\sum_{i \in \mathcal I^{\operatorname{min}}} q_i(\bm v)
    \ge
    \frac{\gamma'}
    {I^{\operatorname{maj}}} \sum_{i \in \mathcal I^{\operatorname{maj}}} q_i(\bm v),
\end{equation*}
where $\gamma' > 0 $. This constraint requires that the average allocation received by the minority group be at least $\gamma'$ times that received by the majority group. Note that this constraint is equivalent to \eqref{eq:Eq} with
$\gamma=\frac{\gamma' I^{\operatorname{min}}}{I^{\operatorname{maj}}}$. \qed
\end{remark}

We denote by $\mathcal M$ the set of mechanisms that satisfy \eqref{eq:IC}, \eqref{eq:IR}, \eqref{eq:AF} and~\eqref{eq:Eq}:
\begin{equation*}\label{def:M}
\mathcal{M} = 
\left\{(\bm q ,\bm m) \in \mathcal L(\mathcal V, [0,1]^I) \times \mathcal L(\mathcal V, \mathbb R^I)
\;\left|\;
\eqref{eq:IC},\,\eqref{eq:IR},\,\eqref{eq:AF},\,\eqref{eq:Eq}
\right.\right\}.
\end{equation*}

The seller's objective is to minimize the worst-case (maximum) ex-post regret\footnote{For brevity, we use the term ``regret'' to refer to the ex-post regret throughout the remainder of the paper.} across all possible realizations of the bidders' values.  The regret of a mechanism is defined as the difference between the revenue that could have been achieved with full knowledge of the bidders’ values and the actual revenue generated by the mechanism. If the seller knew the bidders’ true values $\bm{v}$, they could allocate a $\frac{1}{1+\gamma}$ proportion of the item to the highest bidder overall and a $\frac{\gamma}{1+\gamma}$ proportion to the highest bidder within the minority group.
Each of these bidders would then be charged an amount equal to the product of their allocated proportion and their true value.
This allocation ensures that the equity constraint~\eqref{eq:Eq} is satisfied while also generating the highest possible revenue for the seller. The revenue generated by this allocation is given by
\begin{equation}\label{eq: hindsight revenue}
\frac{1}{1+\gamma} \max_{i \in \mathcal I} v_i + \frac{\gamma}{1+\gamma} \max_{i \in \mathcal I^{\operatorname{min}}} v_i,
\end{equation}
which serves as a benchmark for what can be achieved with complete knowledge of the bidders' values. Note that if the highest bidder within the minority group is also the highest bidder overall, the entire item is allocated to this bidder.

Noting that the worst-case regret is defined as the maximum regret over all possible bidders' values in $\mathcal V$, we can now formulate the regret-based mechanism design problem with the objective of minimizing the worst-case regret as follows.
\begin{equation}\label{Reg-MDP}\tag{R-MDP}
    \begin{aligned}
        &\min_{(\bm q,\bm m)\in\mathcal{M}} \;\max_{\bm{v}\in\mathcal V}&& \frac{1}{1+\gamma} \max_{i \in \mathcal I} v_i + \frac{\gamma}{1+\gamma} \max_{i \in \mathcal I^{\operatorname{min}}} v_i - \sum_{i \in \mathcal I} m_i(\bm v)
    \end{aligned}
\end{equation}
We denote by $\operatorname{REG}^\star$ the optimal value of~\eqref{Reg-MDP}. In Section \ref{sec: Regret-Based Mechanism Design}, we will characterize a feasible mechanism that achieves a constant-factor approximation guarantee for problem \eqref{Reg-MDP}. 

\subsection*{Preliminaries}
We close this section with some preliminaries that will be used in the following sections. The next lemma\footnote{As the result of \Cref{lem: monotonicity of allocation} is well known, we omit the proof.} presents a well-known result available from the literature. Specifically, this lemma prescribes necessary and sufficient conditions for \eqref{eq:IC} to hold.
\begin{lemma}[See, e.g., \cite{krishna2009auction}]\label{lem: monotonicity of allocation}
    A mechanism $(\bm q, \bm m)$ satisfies \eqref{eq:IC} if and only if, for all~$i \in \mathcal I$,
    \begin{itemize}
        \item[(i)] allocation rule $q_i$ is non-decreasing in bidder $i$'s value, i.e.,  $q_i(v_i, \bm{v}_{-i}) \geq q_i(w_i, \bm{v}_{-i})$ for all $\bm v \in \mathcal V$ and  $w_i \in \mathcal V_i : v_i \geq w_i$,
        \item[(ii)] payment rule $m_i$ satisfies $m_i(\bm v) = q_i(\bm v)v_i - \int_{0}^{{v}_i} q_i(x, \bm v_{-i})\text{d}x + m_i(0, \bm v_{-i})$ for all $\bm v \in \mathcal V$.
    \end{itemize}
\end{lemma}
From~\eqref{eq:IR}, it holds that $m_i(0, \bm v_{-i}) \leq 0$ for all $i \in \mathcal I$, and at optimality of \eqref{Reg-MDP}, it must hold that $m_i(0, \bm v_{-i}) = 0$ because the payment $m_i(\bm v)$ is increasing as $m_i(0, \bm v_{-i})$ increases.

\section{Worst-Case Regret Minimization}\label{sec: Regret-Based Mechanism Design}
In this section, we propose a mechanism that is approximately optimal in \eqref{Reg-MDP}.
In particular, we demonstrate that this mechanism generates a regret that is at most a constant multiple (dependent on the equity level) of the optimal value of \eqref{Reg-MDP}. 
We will establish this by deriving a lower bound on the optimal value of \eqref{Reg-MDP} and an upper bound on the regret
of the proposed mechanism. For ease of exposition, from now on we denote by
${\operatorname{REG}}_{(\bm {q}, \bm {m})}(\bm{v})$ the regret of a mechanism $(\bm {q}, \bm {m})$ that is incurred for a given value $\bm{v}\in\mathcal{V}$, which is defined as
\begin{equation*}
    \begin{aligned}
        {\operatorname{REG}}_{(\bm {q}, \bm {m})}(\bm{v}) = & \frac{1}{1+\gamma} \max_{i \in \mathcal I} v_i + \frac{\gamma}{1+\gamma} \max_{i \in \mathcal I^{\operatorname{min}}} v_i - \sum_{i \in \mathcal I} m_i(\bm v).
    \end{aligned}
\end{equation*}
Thus, the objective of \eqref{Reg-MDP} can be expressed as $\max_{\bm{v} \in \mathcal{V}} {\operatorname{REG}}_{(\bm {q}, \bm {m})}(\bm{v})$, and its optimal value $\operatorname{REG}^\star = \min_{(\bm q,\bm m)\in\mathcal{M}} \max_{\bm{v} \in\mathcal V} {\text{REG}}_{(\bm {q}, \bm {m})}(\bm{v})$.

\subsection{Lower Bound on the Optimal Regret}

First, we establish a lower bound on the optimal value of \eqref{Reg-MDP}.
\begin{theorem}[Lower Bound] \label{thm:regret_lower}
It holds that
$\operatorname{REG}^\star \geq \frac{1}{e}\max\left\{\overline{v}^{{\operatorname{min}}},\frac{\overline{v}^{{\operatorname{maj}}}}{1+\gamma}\right\}.$
\end{theorem}

The lower bound from \Cref{thm:regret_lower} is tight when $\gamma \to 0$ or $\gamma \to \infty$. Indeed, the equity~constraint \eqref{eq:Eq} becomes inactive as $\gamma \to 0$, meaning that no preferential treatment is applied to minority bidders and all bidders are treated equally. In this case, \Cref{thm:regret_lower} yields a lower bound on $\operatorname{REG}^\star$ of
$
\frac{1}{e}\max\left\{\overline{v}^{\operatorname{min}}, \overline{v}^{\operatorname{maj}}\right\},
$
which coincides with the known optimal worst-case regret in auction design without equity considerations from \cite{koccyiugit2024regret}. Similarly, as $\gamma \to \infty$, the equity constraint effectively forbids any allocation to majority bidders. Consequently, problem~\eqref{Reg-MDP} reduces to the same minimax regret auction design problem restricted to minority bidders only. In this case, \Cref{thm:regret_lower} yields a lower bound of
$
\frac{1}{e}\overline{v}^{\operatorname{min}},
$
which again coincides with the known optimal value reported in~\cite{koccyiugit2024regret}.

In the next section, we propose a mechanism and derive its upper bound on the worst-case regret. This upper bound matches the lower bound from \Cref{thm:regret_lower} in the two equity limits, implying the optimality of the proposed mechanism in these scenarios. For other values of $\gamma$, the lower and upper bounds allow us to quantify the suboptimality gap of the proposed mechanism.

%We will later show that the lower bound of Theorem~\ref{thm:regret_lower} is tight when $\gamma \downarrow 0$ and as $\gamma \uparrow \infty$, and that it matches the upper bound on the regret derived for the mechanism we propose in the next section. This also implies that the proposed mechanism is optimal in these asymptotic regimes. For other values of $\gamma$, the lower and upper bounds allow us to quantify the suboptimality gap of the proposed mechanism.

\subsection{Proposed Mechanism and Its Regret}
Next, we propose a mechanism $(\bm{\hat{q}}, \bm{\hat{m}})$, and then, in Theorem~\ref{thm:regret_upper}, we characterize an upper bound on its regret ${\text{REG}}_{(\bm{\hat{q}}, \bm{\hat{m}})}(\bm{v})$ across all $\bm{v} \in \mathcal{V}$, which thereby also bounds the optimal value of~\eqref{Reg-MDP}. To define $(\bm{\hat{q}}, \bm{\hat{m}})$, from now on, 
we denote by ${i^{\operatorname{min}}(\bm{v})}$ (respectively, ${i^{\operatorname{maj}}(\bm{v})}$) the bidder who has the largest value in the minority (respectively, majority) group. We use the lexicographical rule to break ties when there are multiple such bidders, that is,
$${i^{\operatorname{min}}(\bm{v})} = \min \arg \max_{i \in \mathcal{I}^{\operatorname{min}}} v_i  \quad
\text{and} \quad {i^{\operatorname{maj}}(\bm{v})} = \min \arg \max_{i \in \mathcal{I}^{\operatorname{maj}}} v_i.$$ 

The definition of $(\bm {\hat{q}}, \bm {\hat{m}})$ will rely on the following quantities:
\begin{equation*}
\begin{aligned}
    &\ell(\bm{v}) = \frac{1}{1+\gamma}v_{{i^{\operatorname{maj}}(\bm{v})}} + \frac{\gamma}{1+\gamma}v_{{i^{\operatorname{min}}(\bm{v})}} \;\;\;\forall \bm v \in \mathcal V,\\
&\overline{\ell} = \frac{1}{1+\gamma}{\max\{\overline{v}^{\operatorname{min}},\overline{v}^{\operatorname{maj}}\}} + \frac{\gamma}{1+\gamma}\overline{v}^{\operatorname{min}}>0.
\end{aligned}
\end{equation*}
Consider any $\bm{v} \in \mathcal{V}$ and suppose for the sake of the discussion here that the seller knows this value. 
In this case, we can interpret~$\ell(\bm{v})$ as the revenue obtained by allocating a $\frac{1}{1+\gamma}$ proportion to the highest majority bidder and a $\frac{\gamma}{1+\gamma}$ proportion to the highest minority bidder. Note also that when $v_{i^{\operatorname{maj}}(\bm{v})} \geq v_{i^{\operatorname{min}}(\bm{v})}$, $\ell(\bm{v})$ coincides with the highest possible revenue~\eqref{eq: hindsight revenue} that the seller could equitably achieve in hindsight. In addition, $\overline{\ell}$ is equal to~\eqref{eq: hindsight revenue} 
%$\overline{\ell} = \ell(\bm{v})$ 
when $\bm v = (\overline{v}_1, \dots, \overline{v}_I)$, i.e., when all bidders' values are equal to their respective upper bounds.

We now define the mechanism $(\bm {\hat{q}}, \bm {\hat{m}})$ as follows. For all minority bidders $i\in\mathcal{I}^{\operatorname{min}}$,
\begin{equation*}
    {\hat{q}}_i(\bm{v}) = \begin{cases}
        \dfrac{\gamma}{1+\gamma}\left(  1 + \log \left(\dfrac{\ell(\bm{v})}{\overline{\ell}}\right)  \right)^+ & \text{ if } i = {i^{\operatorname{min}}(\bm{v})} \text{ and } v_{{i^{\operatorname{maj}}(\bm{v})}} > v_{{i^{\operatorname{min}}(\bm{v})}} \\
        \left(  1 + \log \left(\dfrac{v_i}{\overline{\ell}}\right) \right)^+ & \text{ if } i = {i^{\operatorname{min}}(\bm{v})} \text{ and } v_{{i^{\operatorname{maj}}(\bm{v})}} \leq v_{{i^{\operatorname{min}}(\bm{v})}} \\
         0& \text{ otherwise,}
    \end{cases}
\end{equation*}
and, for all majority bidders $i\in\mathcal{I}^{\operatorname{maj}}$,
\begin{equation*}
    {\hat{q}}_i(\bm{v}) = \begin{cases}
        \dfrac{1}{1+\gamma}\left(  1 + \log \left(\dfrac{\ell(\bm{v})}{\overline{\ell}}\right) \right)^+ & \text{ if } i = {i^{\operatorname{maj}}(\bm{v})} \text{ and } v_{{i^{\operatorname{maj}}(\bm{v})}} > v_{{i^{\operatorname{min}}(\bm{v})}} \\
        0 & \text{ otherwise,}
    \end{cases}
\end{equation*}
Finally, in accordance with Lemma~\ref{lem: monotonicity of allocation}(ii), to ensure incentive compatibility, we let
$${\hat{m}}_i(\bm{v}) = v_i{\hat{q}}_i(\bm{v})-\int_{0}^{v_i}{\hat{q}}_i(x,\bm{v}_{-i}) {\rm{d}} x \quad \forall i \in \mathcal{I}.$$
Under this mechanism, only the highest bidder from the minority group and the highest bidder from the majority group can be allocated a positive proportion of the item. Specifically, if the highest minority bid exceeds the highest majority bid, no positive proportion of the item can be allocated to any majority bidder. Conversely, if the highest bid comes from the majority group, the proportion of the item allocated to the highest majority bidder is limited to $\frac{1}{\gamma}$ times the proportion allocated to the highest minority bidder, ensuring fairness.

The following remark discusses that $(\hat{\bm q}, \hat{\bm m})$ can be implemented as a set-aside mechanism.
\begin{remark}[Set-aside Interpretation]\label{rem: regret set-aside}
Set-asides are commonly used in the procurement and sales of resources to secure increased opportunities for target groups \citep{athey2013set}. These mechanisms set aside a fraction of the item(s) for the targeted (minority) group. 
The mechanism $(\hat{\bm q}, \hat{\bm m})$  employs a set-aside implementation. The allocation rule $\hat{\bm q}$ can be equivalently expressed as:
\begin{equation*}\label{eq: set-aside interpretation of hat q hat m}
\begin{aligned}
\hat{q}_i(\bm v) = \frac{1}{1 + \gamma} \hat{q}_i^{\operatorname{all}}(\bm v) + \frac{\gamma}{1 + \gamma} \hat{q}_i^{\operatorname{min}}(\bm v)\;\;\;\forall i \in \mathcal I,
\end{aligned}
\end{equation*}
where
\begin{equation*}
\begin{aligned}
&\hat{q}_i^{\operatorname{all}}(\bm v) = \begin{cases}
\left(  1 + \log \left(\dfrac{\ell'(\bm{v})}{\overline{\ell}}\right) \right)^+ &\text{if } i = \min \arg \max_{j \in \mathcal I} v_j\\
0 &\text{otherwise},
\end{cases}
\end{aligned}
\end{equation*}
\begin{equation*}
\begin{aligned}
&\hat{q}_i^{\operatorname{min}}(\bm v) = \begin{cases}
\left(  1 + \log \left(\dfrac{\ell'(\bm{v})}{\overline{\ell}}\right) \right)^+ &\text{if } i = \min \arg \max_{j \in \mathcal{I}^{\operatorname{min}}} v_j\\
0 &\text{otherwise,}
\end{cases}
\end{aligned}
\end{equation*}
where
%\begin{equation*}
%    \ell'(\bm{v}) = \frac{1}{1+\gamma}\max_{i \in \mathcal{I}} v_i+\frac{\gamma}{1+\gamma}\max_{i \in \mathcal{I}^{\operatorname{min}}} v_i.
%\end{equation*}
%Here, 
$ \ell'(\bm{v})$ represents the maximum revenue achievable in hindsight, that is, it is equal to \eqref{eq: hindsight revenue}. We note that $ \ell'(\bm{v}) = {\ell}(\bm{v})$ if $v_{{i^{\operatorname{maj}}(\bm{v})}} > v_{{i^{\operatorname{min}}(\bm{v})}}$ and $ \ell'(\bm{v}) = v_{{i^{\operatorname{min}}(\bm{v})}}$ if $v_{{i^{\operatorname{maj}}(\bm{v})}} \leq v_{{i^{\operatorname{min}}(\bm{v})}}$.
Note that $\hat{\bm q}^{\operatorname{min}}$ allocates exclusively to the minority group. Therefore, the allocation $\hat{\bm q}$ reserves a fraction of the item, specifically $\frac{\gamma}{1 + \gamma}$, for allocation to the minority group. The payment rule $\hat{\bm m}$ can similarly be expressed as $\frac{1}{1 + \gamma} \hat{\bm m}^{\operatorname{all}}(\bm v) + \frac{\gamma}{1 + \gamma} \hat{\bm m}^{\operatorname{min}}(\bm v)$, where $\hat{\bm m}^{\operatorname{all}}$ and $\hat{\bm m}^{\operatorname{min}}$ are defined according to Lemma~\ref{lem: monotonicity of allocation}(ii) as functions of $\hat{\bm q}^{\operatorname{all}}$ and $\hat{\bm q}^{\operatorname{min}}$, respectively.  \qed 
\end{remark}

This mechanism is feasible in \eqref{Reg-MDP} as we formalize in the following proposition.
\begin{proposition}\label{prop:feasibility of robust case}
    The mechanism $(\bm {\hat{q}}, \bm {\hat{m}})$ belongs to $\mathcal M$ and is therefore feasible in \eqref{Reg-MDP}.
\end{proposition}

Next, we establish an upper bound on the regret of the feasible mechanism $(\bm {\hat{q}}, \bm {\hat{m}})$.
The bound depends on two auxiliary constants: 
\begin{equation}\label{eq: definition of beta star}
\beta^\star=\max \left\{ \frac{\overline{v}^{\operatorname{maj}}}{(1+\gamma)\overline{\ell}},\frac{1}{e} \right\},
\end{equation}
\begin{equation}\label{eq: definition of theta star}
    \theta^\star = \max_{u\in \mathcal U} \left\{
        \left( \frac{{\overline{v}^{\operatorname{maj}}}+\gamma u}{(1+\gamma){\overline{\ell}}} \right)
        \log \left( \frac{{\overline{v}^{\operatorname{maj}}}+\gamma u}{(1+\gamma){\overline{\ell}}} \right) - \max \left\{ \dfrac{u}{\overline{\ell}}, \frac{1}{e} \right\}
        \log \max \left\{ \dfrac{u}{\overline{\ell}}, \frac{1}{e} \right\} \right\},
\end{equation}
where $\mathcal U = \left[\left(\frac{{\overline{\ell}(1+\gamma)}}{e\gamma}-\frac{{\overline{v}^{\operatorname{maj}}}}{\gamma}\right)^+,{\overline{v}^{\operatorname{min}}}\right]$.
Based on the above definitions, the upper bound is formally presented in the following theorem.

% {\color{orange}
% $$\mathcal U = \left[\left(\frac{{\overline{\ell}(1+\gamma)}}{e\gamma}-\frac{{\overline{v}^{\operatorname{maj}}}}{\gamma}\right)^+,{\overline{v}^{\operatorname{min}}}\right] = \left[\left(\frac{{\overline{\ell}}}{e}+\frac{{\overline{\ell}}}{e\gamma}-\frac{{\overline{v}^{\operatorname{maj}}}}{\gamma}\right)^+,{\overline{v}^{\operatorname{min}}}\right]$$

% When $\gamma\downarrow 0$ and $\frac{{\overline{\ell}}}{e}={\overline{v}^{\operatorname{maj}}}$, then
% $\frac{{\overline{\ell}}}{e}-{\overline{v}^{\operatorname{maj}}}$ is exactly $0$ but $\gamma$ is a positive value close to $0$, so $\lim_{\gamma\rightarrow0} (\frac{{\overline{\ell}}}{e\gamma}-\frac{{\overline{v}^{\operatorname{maj}}}}{\gamma}) = \lim_\gamma 0 = 0$. Then

% $$\lim_{\gamma\downarrow 0 }\mathcal U =\left[\left(\frac{{\overline{\ell}}}{e}\right)^+,{\overline{v}^{\operatorname{min}}}\right]=\left[\left({\overline{v}^{\operatorname{maj}}}\right)^+,{\overline{v}^{\operatorname{min}}}\right]=\left[{\overline{v}^{\operatorname{maj}}},{\overline{v}^{\operatorname{min}}}\right]$$
% }

%{\color{red}[Quick link for \Cref{prop:extreme_cases} ($\gamma\downarrow 0 $ and $\gamma\uparrow \infty$).}

\begin{theorem}[Upper Bound] \label{thm:regret_upper}
The regret of the mechanism $(\hat{\bm{q}},\hat{\bm{m}})$ satisfies
\begin{equation*}
\begin{aligned}
    {\operatorname{REG}}_{(\bm {\hat{q}}, \bm {\hat{m}})}(\bm{v}) \leq 
    {\overline{\ell}}\max\left\{\frac{1}{e}, \theta^\star - \beta^\star \log \beta^\star \right\} \;\;\;\forall \bm{v} \in \mathcal{V}.
\end{aligned}
\end{equation*}
\end{theorem}

The lower bound from \Cref{thm:regret_lower} and the upper bound from \Cref{thm:regret_upper} do not depend on $\bm{v}$.
This independence allows us to compare the two bounds and derive the approximation ratio, discussed later in \Cref{sec: approximation ratio}.
Computing the upper bound of Theorem~\ref{thm:regret_upper} requires maximizing a continuous objective function over a bounded interval $\mathcal U$. This can be done to arbitrary accuracy by grid search, so this is computationally easy. 

\subsection{Approximation Ratio and Asymptotic Optimality}\label{sec: approximation ratio}

By Theorems~\ref{thm:regret_lower} and~\ref{thm:regret_upper}, the mechanism $(\bm {\hat{q}}, \bm {\hat{m}})$ achieves a performance guarantee in terms of regret, as stated in the following corollary.

\begin{corollary}[Approximation Factor]\label{cor: constant approximation factor}
The regret of the mechanism $(\hat{\bm{q}},\hat{\bm{m}})$ is at most $\lambda$ times the optimal regret of \eqref{Reg-MDP},
where 
\begin{equation*}
\begin{aligned}
    \lambda =~\ &  {\overline{\ell}}\max\Bigl\{1, e\Bigl(\theta^\star - \beta^\star \log \beta^\star\Bigr) \Bigr\}  /  \max\Bigl\{\overline{v}^{{\operatorname{min}}},\frac{\overline{v}^{{\operatorname{maj}}}}{1+\gamma}\Bigr\},
\end{aligned}
\end{equation*}
and $\beta^\star$ and $\theta^\star$ are defined in \eqref{eq: definition of beta star} and \eqref{eq: definition of theta star}, respectively. 
\end{corollary}
\Cref{cor: constant approximation factor}
follows from Theorems~\ref{thm:regret_lower} and~\ref{thm:regret_upper}.
The ratio $\lambda$ is uniquely determined by the equity parameter $\gamma$ and the upper bounds $\overline{v}^{{\operatorname{min}}}$ and $\overline{v}^{{\operatorname{maj}}}$ of the minority and majority values, respectively. In the symmetric case where these bounds are equal—which should typically be assumed unless there is prior information suggesting otherwise—we can compute $\lambda$ numerically and obtain a worst-case approximation ratio of $1.31$, attained when $\gamma = 0.91$.

When $\overline{v}^{{\operatorname{min}}}$ and $\overline{v}^{{\operatorname{maj}}}$ diverge, the approximation ratio for each fixed $\gamma > 0$ is computed and illustrated in Figure~\ref{fig: lambda}. We observe that the bound improves when $\overline{v}^{{\operatorname{min}}}>\overline{v}^{{\operatorname{maj}}}$ and deteriorates otherwise. This highlights the challenge of achieving more equitable outcomes when minority bidders are less inclined to obtain the item compared to majority bidders. In practice, the choice of $\gamma$ reflects the mechanism designer's fairness preferences; smaller values of $\gamma$ may be preferable as they allow the mechanism designer to balance fairness guarantees against revenues. For instance, the former equitable procurement targets of the United States mandated that 15 percent of procurement dollars be allocated to minority-owned businesses (see Remark~1). 

Furthermore, Figure~\ref{fig: lambda} shows that the approximation factor $\lambda$ converges to one in the two asymptotic equity regimes as $\gamma$ approaches~zero and $\gamma$ approaches $\infty$. This means that our approximate mechanism $(\bm{\hat{q}}, \bm{\hat{m}})$ becomes optimal in these two regimes. We formally establish this result in \Cref{prop:extreme_cases} below.

\begin{figure}[!htbp]
    \centering
    \includegraphics[scale=0.45]{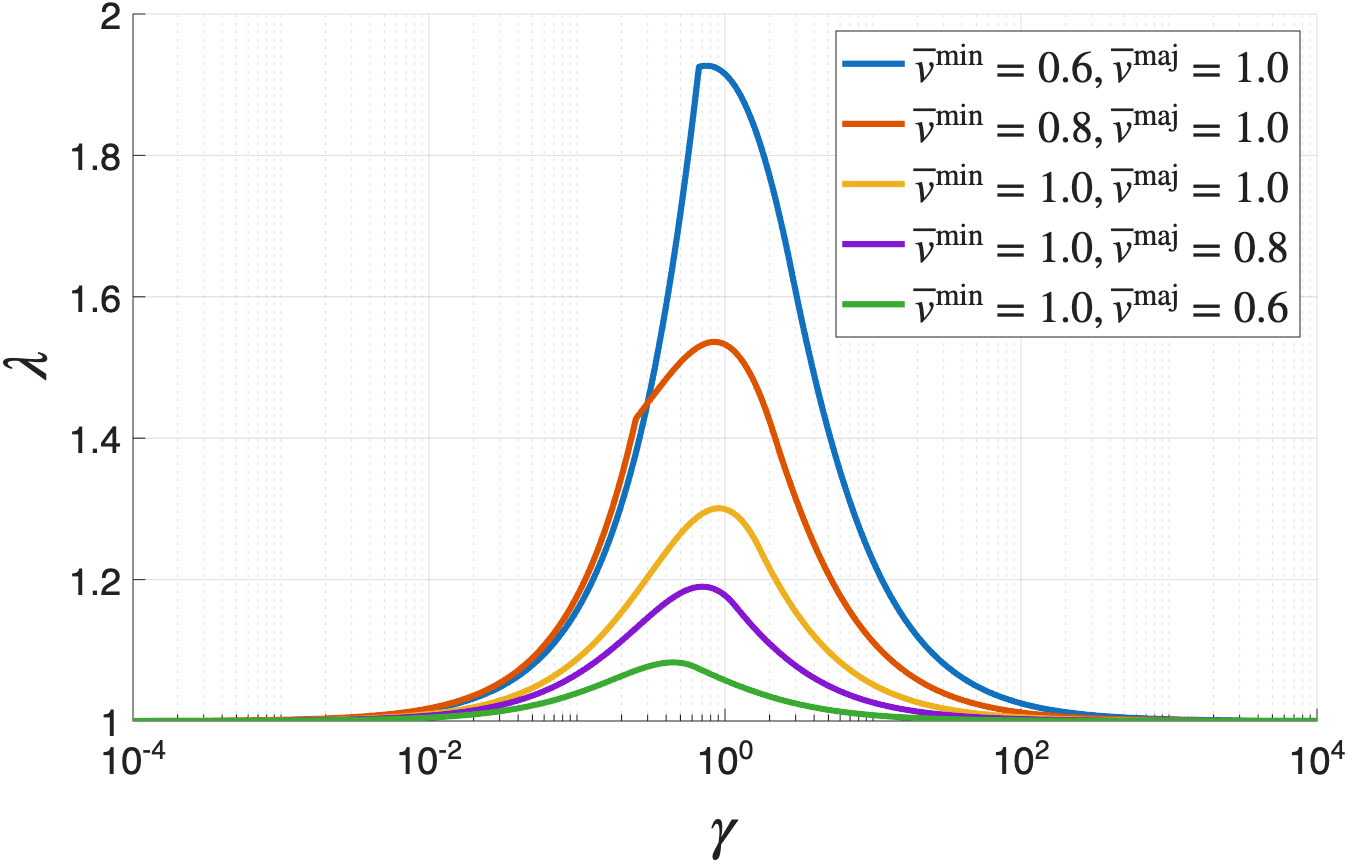}
    \caption{The approximation factor $\lambda$ as a function of the equity parameter $\gamma$ for different pairs of maximum value $(\overline{v}^{\operatorname{min}}, \overline{v}^{\operatorname{maj}})$: (0.6, 1), (0.8, 1), (1, 1), (1, 0.8), and (1, 0.6).}
    \label{fig: lambda}
\end{figure}

\begin{proposition}\label{prop:extreme_cases}
The approximation factor $\lambda$ tends to one as $\gamma$ tends to zero or to infinity.
\end{proposition}
In the two equity-limit scenarios $\gamma \downarrow 0$ and as $\gamma \uparrow \infty$, the factor $\lambda$ tends to~$1$, and
the lower bound from \Cref{thm:regret_lower} and the upper bound from \Cref{thm:regret_upper} become tight. 
Consequently, our proposed mechanism $(\bm{\hat{q}}, \bm{\hat{m}})$ is asymptotically optimal in these two regimes.

While allocation fairness is the primary concern and motivation of this study, we also find, somewhat serendipitously, that the minority group benefits from an additional layer of pricing protection under the proposed mechanism. Specifically, the ratio of payment to allocation for the minority group is no worse than that of its majority counterpart.
% Beyond allocation-based fairness, our proposed mechanism also ensures fairness of payments. 
We formally establish this property in Proposition~\ref{prop:fairness of payments} where we extend it to all set-aside mechanisms $(\bm q^{\operatorname{sa}}, \bm m^{\operatorname{sa}})$ that are individually rational and incentive compatible and of the following form:
\begin{equation*}
\begin{aligned}
&q^{\operatorname{sa}}_i(\bm v) = t q_i^{\operatorname{all}}(\bm v) + (1-t) 
q_i^{\operatorname{min}}(\bm v)\;\;\;&\forall i \in \mathcal I,\\
% &m^{\operatorname{sa}}_i(\bm v) = t m_i^{\operatorname{all}}(\bm v) + (1-t) m_i^{\operatorname{min}}(\bm v)\;\;\;&\forall i \in \mathcal I,
\end{aligned}
\end{equation*}
where $t\in[0,1]$ and $(q^{\rm{all}},q^{\rm{min}})$ satisfy
\begin{equation*}
\begin{aligned}
    &q_i^{\operatorname{all}}(\bm v) = 0 \quad \text{ whenever } i \neq \min \arg \max_{j \in \mathcal I} v_j,\\
    &q_i^{\operatorname{min}}(\bm v) = 0 \quad \text{ whenever } i \neq \min \arg \max_{j \in \mathcal I^{\operatorname{min}}} v_j.
    \end{aligned}
\end{equation*}
%The allocation rules $\bm{q}^{\operatorname{all}}$ and $\bm{q}^{\operatorname{min}}$ allocate the item only to the highest bidder and the highest minority bidder, respectively.
By Remark~\ref{rem: regret set-aside}, $(\bm{\hat{q}}, \bm{\hat{m}})$ is a set-aside mechanism. %\noterq{show feasibility}

\begin{proposition}\label{prop:fairness of payments}
For all $ \bm{v} \in \mathcal{V}$, it holds that
$q^{\operatorname{sa}}_{i^{\operatorname{min}}(\bm{v})}(\bm{v}) m^{\operatorname{sa}}_{i^{\operatorname{maj}}(\bm{v})}(\bm{v}) \geq q^{\operatorname{sa}}_{i^{\operatorname{maj}}(\bm{v})}(\bm{v}) m^{\operatorname{sa}}_{i^{\operatorname{min}}(\bm{v})}(\bm{v}).$
% $$\sum_{i\in \minority}q^{\operatorname{sa}}_i(\bm{v})  
% \sum_{i\in \majority}m^{\operatorname{sa}}_i(\bm{v})  \geq \sum_{i\in \majority}q^{\operatorname{sa}}_i(\bm{v})  
% \sum_{i\in \minority}m^{\operatorname{sa}}_i(\bm{v}) .$$
\end{proposition}
By \Cref{prop:fairness of payments}, whenever the two bidders' allocations are non-zero, the highest minority bidder enjoys a \textit{unit price} (the ratio of payment to allocation) no greater than that of the highest majority bidder. By the definition of $(\bm {q}^{\operatorname{sa}}, \bm {m}^{\operatorname{sa}})$, the allocation and the payment for a group's highest bidder are exactly the total allocation and total payments for the whole group. %\noteck{I commented out some part here, fyi}

\section{Expected Revenue Maximization}\label{sec: stochastic}
% While the primary focus of this paper is on the distribution-free, regret-based approach, the mechanism design problem under consideration has not been previously addressed even in the classical stochastic setting. Therefore, to establish a benchmark for comparison and a deeper understanding of the problem structure, this section is dedicated to studying the problem under the assumption that the seller has knowledge of the bidders' value distribution.

While the primary focus of the paper is to propose a robust mechanism for more equitable allocation outcomes when the distribution of bidders’ values $\bm v$ is unknown or difficult to estimate, alternative approaches become relevant when the distribution can be accurately estimated. In particular, under perfect distributional information, the standard objective in mechanism design is (expected) revenue maximization (this coincides with expected regret minimization when the distribution is known). Without equity considerations, the optimal mechanism for this objective admits a closed-form solution (see, e.g., \citep{myerson1981optimal}) under the independence and regularity conditions discussed below. We extend this analytical result by incorporating an equity constraint that ensures the total allocation to minority bidders is not disproportionately small compared to that of the majority group. Later in \Cref{sec: Numerical Experiment}, we use this mechanism in our numerical study. While it performs best in-sample, its revenue drops significantly out-of-sample, suggesting that its optimality is fragile—especially in comparison to the assumption-free and distribution-free mechanism we propose in Section 3.

We now introduce the stochastic mechanism design problem with the revenue maximization objective. We model each bidder $i$'s value as a continuous random variable $\tilde{v}_i$ governed by a cumulative distribution function~$\mathbb{F}_i$ with support $\mathcal{V}_i$. The joint cumulative distribution function of the bidders' values is denoted by $\mathbb{F}$. We assume that $\mathbb{F}$ is known to the seller and solve
\begin{equation}\label{Rev-MDP}\tag{S-MDP}
\begin{aligned}
&\max_{(\bm q,\bm m)\in\mathcal{M}} &&\mathbb{E}_{\mathbb F}\left[\sum_{i \in \mathcal I} m_i(\tilde{\bm v})\right].
\end{aligned}
\end{equation}

The goal of this section is to characterize the optimal mechanism $(\bm q^\star, \bm m^\star)$ for problem \eqref{Rev-MDP}. This characterization relies on the following assumption on the distribution of the bidders' values, which we assume to hold throughout this section.

\begin{assumption}[Independence \& Regularity]\label{assum:regularity}
Bidders' values $\tilde v_i$ are mutually independent under~$\mathbb F$.
Furthermore, distribution $\mathbb F$ is regular, i.e., for all $i \in \mathcal I$, the marginal density $f_i$ of $\tilde v_i$ exists and is strictly positive on $\mathcal V_i$, and the virtual value function defined as
\begin{equation*}
    \begin{aligned}
        \psi_i(v_i) = v_i - \frac{1 - \mathbb F_i(v_i)}{f_i(v_i)}
    \end{aligned}
\end{equation*}
is non-decreasing in $v_i$ on $\mathcal V_i$. 
\end{assumption}
Independence and regularity are standard assumptions and widely used in mechanism design literature. The class of regular distributions is very large and contains, for example, uniform, normal, logistic and exponential distributions and their truncations; see e.g., \citep{ewerhart2013regular}.

% We present the proof of Lemma~\ref{lem: reformulation in terms of virtual values} in Appendix~\ref{appendix: proofs}.

%The proof largely follows the arguments in \citep[Chapter~5]{krishna2009auction}, with the addition of the equity constraint \eqref{eq:Eq}.

% \notenr{Not sure if it makes too much sense to have just one proof in the appendix given that the paper is not that long.}

For all minority bidders $i \in \mathcal I^\text{min}$, we claim that the optimal allocation rule is
\begin{equation*}
    \begin{aligned}
        q_i^\star(\bm v) = \begin{cases}
            1 &\text{if } i = \min \arg \max_{j \in  \mathcal I} \psi_j(v_{j}), \,\psi_i(v_i) \geq 0\\
            \dfrac{\gamma}{1 + \gamma} &\text{if } i = \min \arg \max_{j \in  \mathcal I^\text{min}} \psi_j(v_{j}),\, \psi_i(v_{i}) < \max_{j \in  \mathcal I^\text{maj}} \psi_j(v_{j}),\\
&\quad\quad\quad\quad\quad\quad\quad\quad\quad\;\;\text{and } \max_{j \in \mathcal I^{\operatorname{maj}}} \psi_j(v_j) + \gamma \psi_i(v_i) \geq 0,\\
            0 &\text{otherwise,}
        \end{cases}
    \end{aligned}
\end{equation*}
and for all majority bidders $i \in \mathcal I^\text{maj}$, it is
\begin{equation*}
    \begin{aligned}
        q_i^\star(\bm v) = \begin{cases}
            \dfrac{1}{1 + \gamma} &\text{if } i = \min \arg \max_{j \in  \mathcal I} \psi_j(v_{j}),\, \psi_i(v_i) + \gamma \max_{j \in \mathcal I^{\operatorname{min}}} \psi_j(v_j) \geq 0\\
            0 &\text{otherwise.}
        \end{cases}
    \end{aligned}
\end{equation*}
Here, we again use a lexicographic tie-breaker to break ties, but it does not play a role in our analysis.
%By Lemma~\ref{lem: monotonicity of allocation}~(ii), for any $i \in \all$, the optimal payment rule is given by
Following Lemma~\ref{lem: monotonicity of allocation}~(ii), the optimal payment rule for each $i\in\mathcal{I}$ is then given by
\begin{equation*}
    \begin{aligned}
        m_i^\star(\bm v) = q^\star_i(\bm v)v_i - \int_{0}^{{v}_i} q^\star_i(x, \bm v_{-i})\text{d}x.
    \end{aligned}
\end{equation*}

The revenue-based mechanism also adopts a set-aside interpretation.

\begin{remark}[Set-aside Interpretation]
The optimal mechanism $(\bm q^\star, \bm m^\star)$ incorporates a set-aside approach. Indeed, the optimal allocation rule~$\bm q^\star$ can be equivalently expressed as:
\begin{equation*}
\begin{aligned}
{q}^\star_i(\bm v) = \frac{1}{1 + \gamma} {q}_i^{\text{all}}(\bm v) + \frac{\gamma}{1 + \gamma} {q}_i^{\operatorname{min}}(\bm v)\;\;\;\forall i \in \mathcal I,
\end{aligned}
\end{equation*}
where
\begin{equation*}
\begin{aligned}
&q_i^{\text{all}}(\bm v) = \begin{cases}
1 &\text{if } i = \min \arg \max_{j \in \mathcal I} \psi_{j}(v_j),  \psi_{i}(v_i) + \gamma \max_{j \in \minority} \psi_{j}(v_j) \geq 0\\
0 &\text{otherwise,}
\end{cases}
\end{aligned}
\end{equation*}
\begin{equation*}
\begin{aligned}
&q_i^{\operatorname{min}}(\bm v) = \begin{cases}
1 &\text{if } i = \min \arg \max_{j \in \minority} \psi_{j}(v_j), \max_{j \in \all} \psi_{j}(v_j) + \gamma \psi_{i}(v_i) \geq 0\\
0 &\text{otherwise.}
\end{cases}
\end{aligned}
\end{equation*}
Note that $\bm q^{\operatorname{min}}$ allocates exclusively to the minority group. Therefore, the optimal allocation $\bm q^\star$ reserves a fraction of the item, specifically $\frac{\gamma}{1 + \gamma}$, for allocation to the minority group. The payment rule $\bm m^\star$ can similarly be expressed as $\frac{1}{1 + \gamma} \bm m^{\text{all}}(\bm v) + \frac{\gamma}{1 + \gamma} \bm m^{\operatorname{min}}(\bm v)$, where $\bm m^{\text{all}}$ and $\bm m^{\operatorname{min}}$ are defined according to Lemma \ref{lem: monotonicity of allocation}(ii) as functions of $\bm q^{\text{all}}$ and $\bm q^{\operatorname{min}}$, respectively. \qed 
\end{remark}

We next establish the feasibility and the optimality of $(\bm q^\star, \bm m^\star)$.

\begin{proposition}\label{prop:feasibility of stochastic case}
    The mechanism $(\bm q^\star, \bm m^\star)$ belongs to $\mathcal M$ and is therefore feasible in \eqref{Rev-MDP}.
\end{proposition}

% We present the proof of \Cref{prop:feasibility of stochastic case} in Appendix~\ref{appendix: proofs}.

\begin{theorem}\label{theorem:optimal mechanism REV}
    The mechanism $(\bm q^\star, \bm m^\star)$ is optimal in \eqref{Rev-MDP}.
\end{theorem}

\section{Numerical Illustration}\label{sec: Numerical Experiment}

We now numerically assess the performance of the regret-based robust mechanism $(\hat{\bm{q}},\hat{\bm{m}})$ from Section~\ref{sec: Regret-Based Mechanism Design} and the revenue-based mechanism $(\bm{q}^\star,\bm{m}^\star)$ from Section~\ref{sec: stochastic}.
Throughout this numerical illustration, we assume there are $I = 2$ bidders: one minority bidder (bidder $1$) with the value support $[0,\overline{v}^{\operatorname{min}}]$ and one majority bidder (bidder $2$) with the value support $[0,\overline{v}^{\operatorname{maj}}]$. We set the equity level $\gamma = \frac{1}{4}$. 
The robust mechanism $(\hat{\bm{q}},\hat{\bm{m}})$ can be obtained without knowing the bidders' value distribution. 
The revenue-based mechanism $(\bm{q}^\star,\bm{m}^\star)$ is derived from a crisp (but possibly misspecified) value distribution $\mathbb{F}$. We assume that bidders' values are mutually independent under $\mathbb{F}$ and that its marginals are scaled beta distributions, each with shape parameters equal to $2$. That is, $\frac{\tilde v_1}{\overline{v}^{\operatorname{min}}},\frac{\tilde v_2}{\overline{v}^{\operatorname{maj}}} \sim \operatorname{Beta}(2,2)$ and are mutually independent.
The density of $\operatorname{Beta}(2,2)$ 
is bell-shaped and symmetric and thus resembles the normal distribution.

In practice, the estimated (in-sample) distribution is likely to be subject to estimation errors. In other words, the true (out-of-sample) distribution of the bidders' values $\tilde{\bm{v}}$ may deviate from the estimated distribution~$\mathbb{F}$. To capture this, we assume that the true distribution is representable as $\mathbb{F}^{\epsilon,\rho} = (1-\epsilon)\mathbb{F} + \epsilon \mathbb{B}^{\rho}$, where $\epsilon \in [0,1]$ represents the contamination level, and $\mathbb{B}^{\rho}$, $\rho \in [-1,+1]$, denotes an extremal distribution to be defined below. The distribution $\mathbb{F}^{\epsilon,\rho}$ is a contaminated version of $\mathbb{F}$, where smaller values of $\epsilon$ indicate greater fidelity in the estimated distribution $\mathbb{F}$, while larger values indicate lower fidelity.

We define the extremal distribution $\mathbb{B}^{\rho}$ as the distribution supported on the extreme points of $\mathcal{V} = [0,\overline{v}^{\operatorname{min}}]\times[0,\overline{v}^{\operatorname{maj}}]$, i.e., $\{0,\overline{v}^{\operatorname{min}}\}\times\{0,\overline{v}^{\operatorname{maj}}\}$, under which
\begin{equation*}
    (\tilde{v}_1,\tilde{v}_2) = \begin{cases}
        (0,0) & \text{ with probability } \frac{1+\rho}{4} \\
        (0,\overline{v}^{\operatorname{maj}}) & \text{ with probability } \frac{1-\rho}{4} \\
        (\overline{v}^{\operatorname{min}},0) & \text{ with probability } \frac{1-\rho}{4} \\    (\overline{v}^{\operatorname{min}},\overline{v}^{\operatorname{maj}}) & \text{ with probability } \frac{1+\rho}{4}.
    \end{cases}
\end{equation*}
Note that $\tilde{v}_1$ and $\tilde{v}_2$ may not be independent under $\mathbb{B}^{\rho}$, depending on the value of $\rho$, which can be interpreted as a correlation parameter.
Additionally, the distribution $\mathbb{B}^{\rho}$ is discrete. Thus, unlike $\mathbb{F}$, $\mathbb{B}^{\rho}$ is not regular. As a result, the distribution $\mathbb{F}^{\epsilon,\rho}$ is also irregular, except when $\epsilon = 0$.

We evaluate the performance of the mechanisms $(\hat{\bm{q}},\hat{\bm{m}})$ and $(\bm{q}^\star,\bm{m}^\star)$ primarily in terms of the revenue they generate and the regret they incur. This evaluation follows a stress-test approach, conducted on the distribution $\mathbb{F}^{\epsilon,\rho}$, which may differ from $\mathbb{F}$. For comprehensiveness, we introduce two other benchmark mechanisms:
\begin{itemize}[noitemsep]
    \item Inspired by~\cite{pai2012auction} and~\cite{fallah2024fair}, we consider a variant of~\eqref{Rev-MDP} that enforces equity only in expectation, i.e., the original~\eqref{eq:Eq} is replaced by
    \begin{align*}
    &\mathbb{E}_{\mathbb{F}} \left[ \sum_{i \in \mathcal I^{\operatorname{min}}} q_i(\tilde{\bm v}) \right] \geq \gamma \mathbb{E}_{\mathbb{F}} \left[ \sum_{i \in \mathcal I^{\operatorname{maj}}} q_i(\tilde{\bm v}) \right].
    \end{align*}
    The resulting mechanism design problem is a relaxation of~\eqref{Rev-MDP}, and we denote the optimal mechanism of this relaxed problem as $(\overline{\bm{q}},\overline{\bm{m}})$. 
    \item We also consider the optimal mechanism $({\bm{q}}^{\epsilon,\rho},{\bm{m}}^{\epsilon,\rho})$ in~\eqref{Rev-MDP}, which is tailored to the true distribution $\mathbb{F}^{\epsilon,\rho}$.
    We emphasize that this mechanism cannot be realistically computed because, in practice, the distribution of bidders' values is not known and estimation errors are inevitable.
\end{itemize}
We remark that the mechanism $(\overline{\bm{q}}, \overline{\bm{m}})$ may lie outside $\mathcal{M}$ and fail to be fair ex-post. Additionally, to the best of our knowledge, no closed-form expression exists for $(\overline{\bm{q}}, \overline{\bm{m}})$. Likewise, our analytical characterization of $({\bm{q}}^{\star}, {\bm{m}}^{\star})$ does not straightforwardly extend to $({\bm{q}}^{\epsilon,\rho}, {\bm{m}}^{\epsilon,\rho})$ due to the irregularity of $\mathbb{F}^{\epsilon,\rho}$. Hence, unlike $(\bm{q}^\star, \bm{m}^\star)$ and $(\hat{\bm{q}}, \hat{\bm{m}})$, we must resort to a numerical approach to find $(\overline{\bm{q}}, \overline{\bm{m}})$ and $({\bm{q}}^{\epsilon,\rho}, {\bm{m}}^{\epsilon,\rho})$. Since the respective mechanism design problems are infinite-dimensional linear programs, we solve a discrete approximation. 
Specifically, we approximate $\mathcal{V}$
by the grid
$\big([0,\overline{v}^{\operatorname{min}}]\cap \delta \mathbb{Z}\big)\times\big([0,\overline{v}^{\operatorname{maj}}]\cap \delta \mathbb{Z}\big)$,
where $\delta = 0.01$, and assign each grid point the probability mass of its lower-left grid cell under $\mathbb{F}$.
The discrete approximation of distribution $\mathbb{F}^{\epsilon,\rho}$ is obtained by replacing $\mathbb{F}$ with its discrete approximation in the construction of $\mathbb{F}^{\epsilon,\rho}$.
The resulting finite-dimensional linear programs are then solved in MATLAB R2025a via the YALMIP interface \citep{Lofberg2004} and the MOSEK solver \citep{mosek}. We remark that this discretization scheme for computing $(\overline{\bm{q}}, \overline{\bm{m}})$ and $({\bm{q}}^{\epsilon,\rho}, {\bm{m}}^{\epsilon,\rho})$  is only reasonably accurate and computationally efficient when the number of bidders~$I$ is small. In contrast, our proposed mechanisms  $(\hat{\bm{q}},\hat{\bm{m}})$ and $(\bm{q}^\star,\bm{m}^\star)$ are available in closed form for any $I$.

In our experiment, we first set $\rho = 0$ and consider $\epsilon \in \{0, 0.1, \dots, 1\}$. We repeat the analysis for three support configurations, $(\overline{v}^{\operatorname{min}},\overline{v}^{\operatorname{maj}})\in\{(0.8,1),(1,1),(1,0.8)\}$, and compute the expected revenues of the four mechanisms under the discrete approximation of the distribution $\mathbb{F}^{\epsilon, 0}$. 
Note that for a fixed $(\overline{v}^{\operatorname{min}},\overline{v}^{\operatorname{maj}})$ and a fixed $\rho$, the mechanism $({\bm{q}}^{\epsilon, \rho}, {\bm{m}}^{\epsilon, \rho})$ must be recomputed every time $\epsilon$ changes. 
We use the expected revenue of $({\bm{q}}^{\epsilon, \rho}, {\bm{m}}^{\epsilon, \rho})$ to normalize the expected revenues of the remaining three mechanisms. 
By construction, the revenues of $(\hat{\bm{q}},\hat{\bm{m}})$ and $(\bm{q}^\star,\bm{m}^\star)$ cannot exceed that of $({\bm{q}}^{\epsilon, \rho}, {\bm{m}}^{\epsilon, \rho})$, because the discrete approximation relaxes the original constraints of $\mathcal{M}$, and our approximation of the distribution favors $({\bm{q}}^{\epsilon, \rho}, {\bm{m}}^{\epsilon, \rho})$ in terms of expected revenues. Therefore, the normalized revenues of $(\hat{\bm{q}},\hat{\bm{m}})$ and $(\bm{q}^\star,\bm{m}^\star)$ cannot exceed one. In fact, the revenues of $(\bm{q}^\star, \bm{m}^\star)$ and $({\bm{q}}^{\epsilon, \rho}, {\bm{m}}^{\epsilon, \rho})$ do not necessarily match even when $\epsilon = 0$ due to discretization, which favors the latter.
The comparison between $({\bm{q}}^{\epsilon,\rho},{\bm{m}}^{\epsilon,\rho})$ and $(\overline{\bm{q}},\overline{\bm{m}})$ is less straightforward because $(\overline{\bm{q}},\overline{\bm{m}})$ is only required to satisfy the equity constraint in expectation.
For this reason, the normalized revenue of $(\overline{\bm{q}},\overline{\bm{m}})$ can exceed one. However, this potential excess revenue comes at the cost of violating the equity constraint~\eqref{eq:Eq}.  
In fact, the probability that $\sum_{i \in \mathcal I^{\operatorname{min}}} \overline{q}_i(\tilde{\bm v}) < \gamma \sum_{i \in \mathcal I^{\operatorname{maj}}} \overline{q}_i(\tilde{\bm v})$ uniformly exceeds 40\% under $\mathbb{F}$ for the three support configurations, $(\overline{v}^{\operatorname{min}},\overline{v}^{\operatorname{maj}})\in\{(0.8,1),(1,1),(1,0.8)\}$, implying a significant chance that the equity constraint fails to hold ex-post.

\begin{figure}[!htbp]
    \centering
    \includegraphics[scale=0.30]{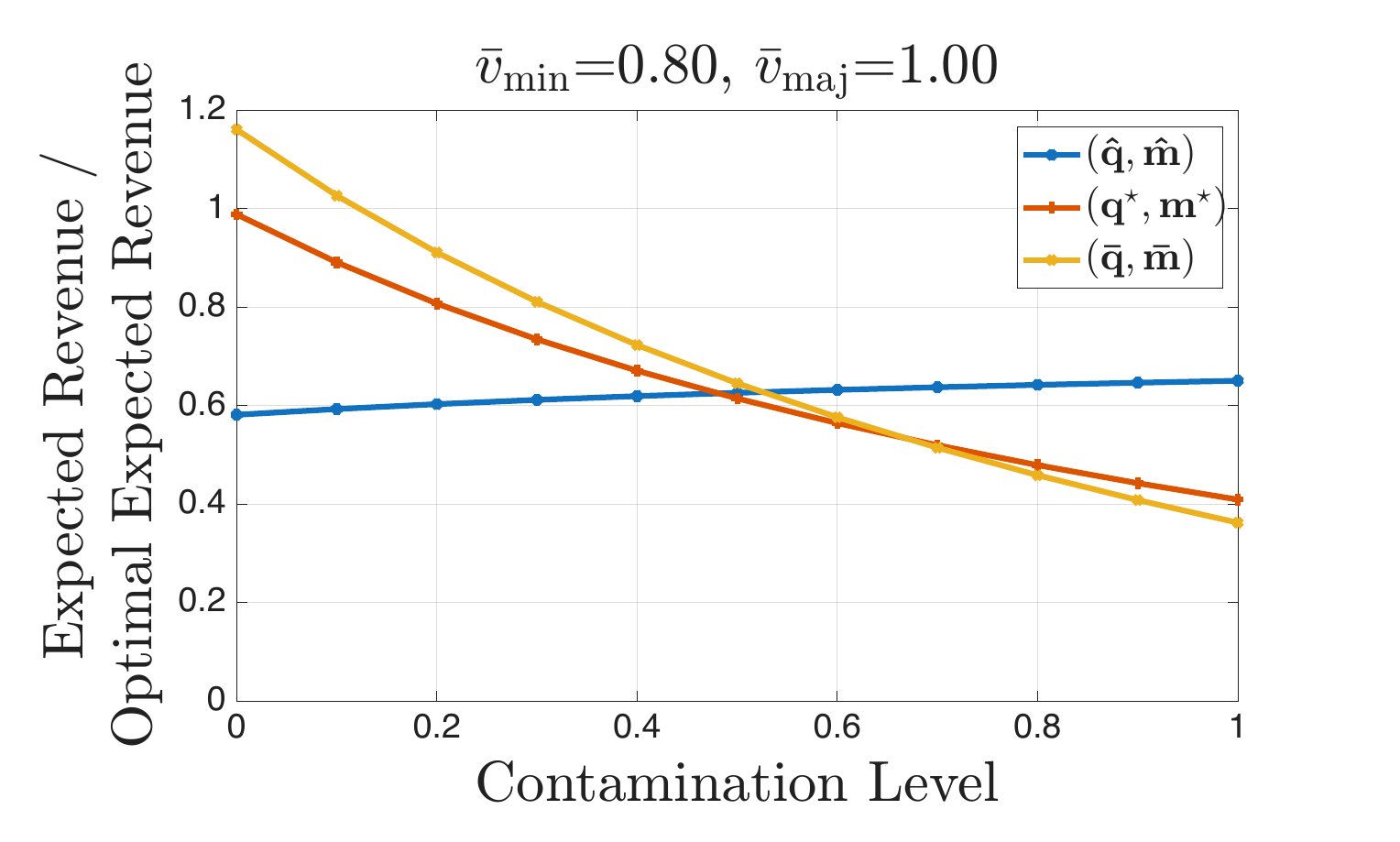}
    \includegraphics[scale=0.30]{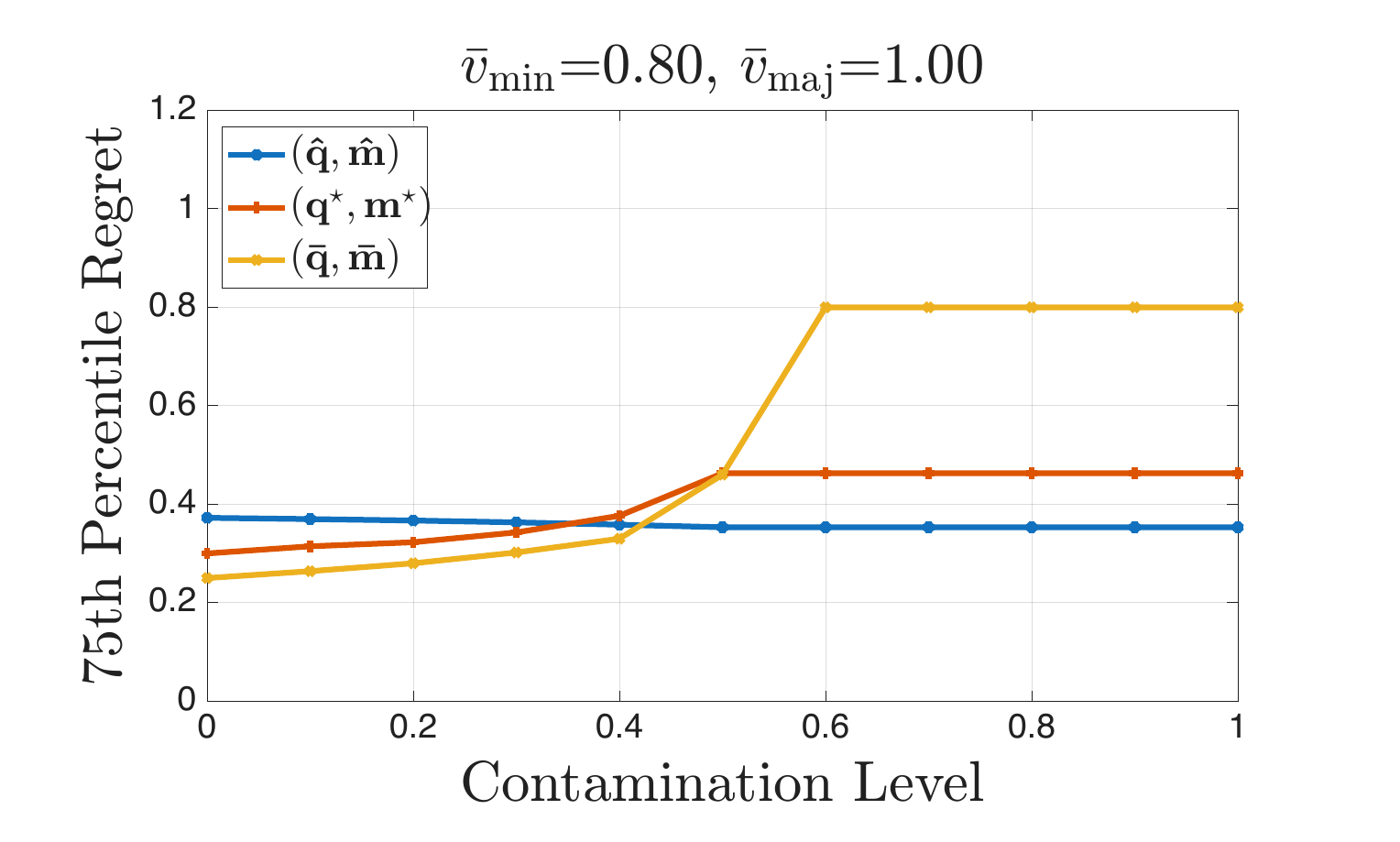}
    
    \includegraphics[scale=0.30]{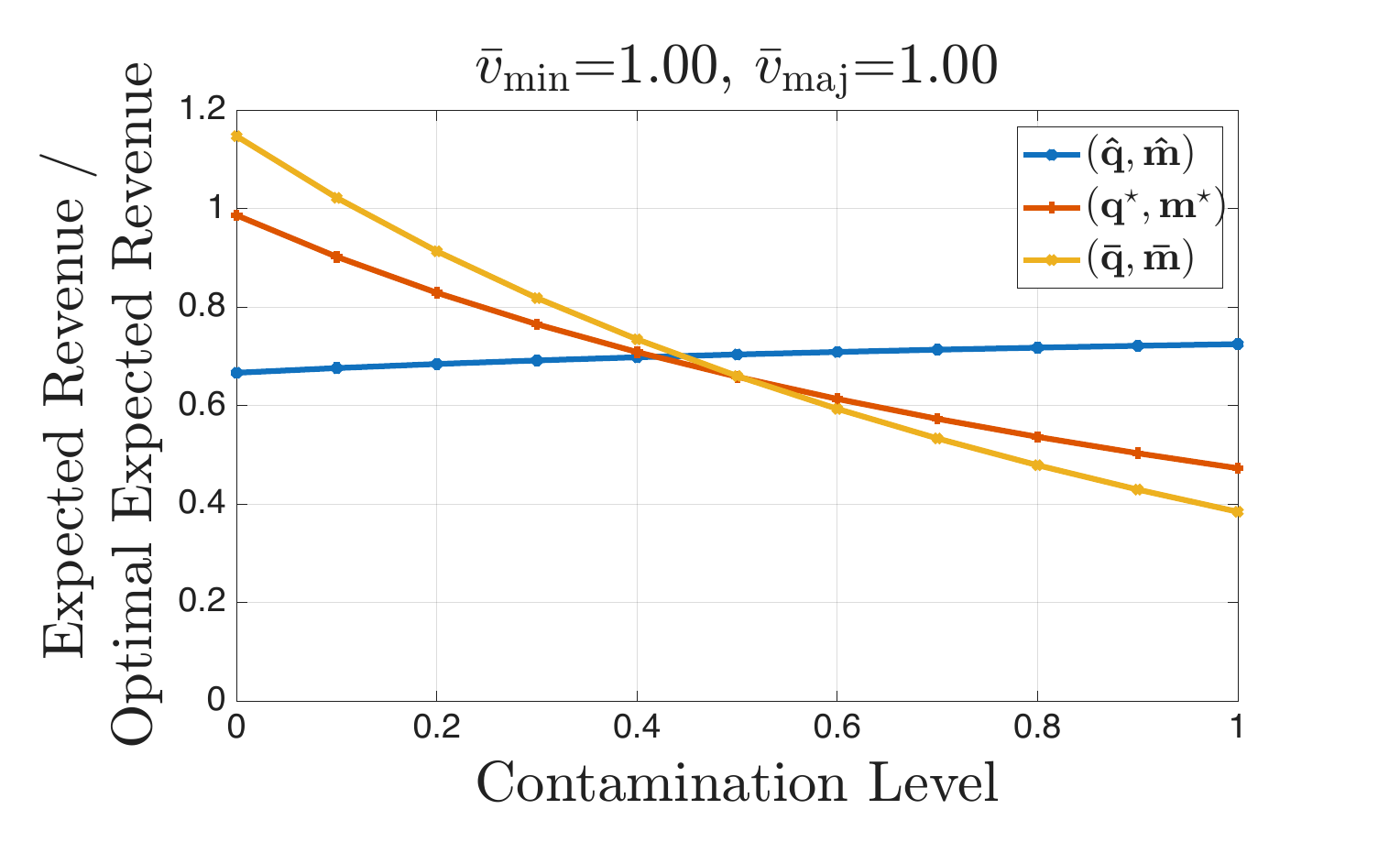}
    \includegraphics[scale=0.30]{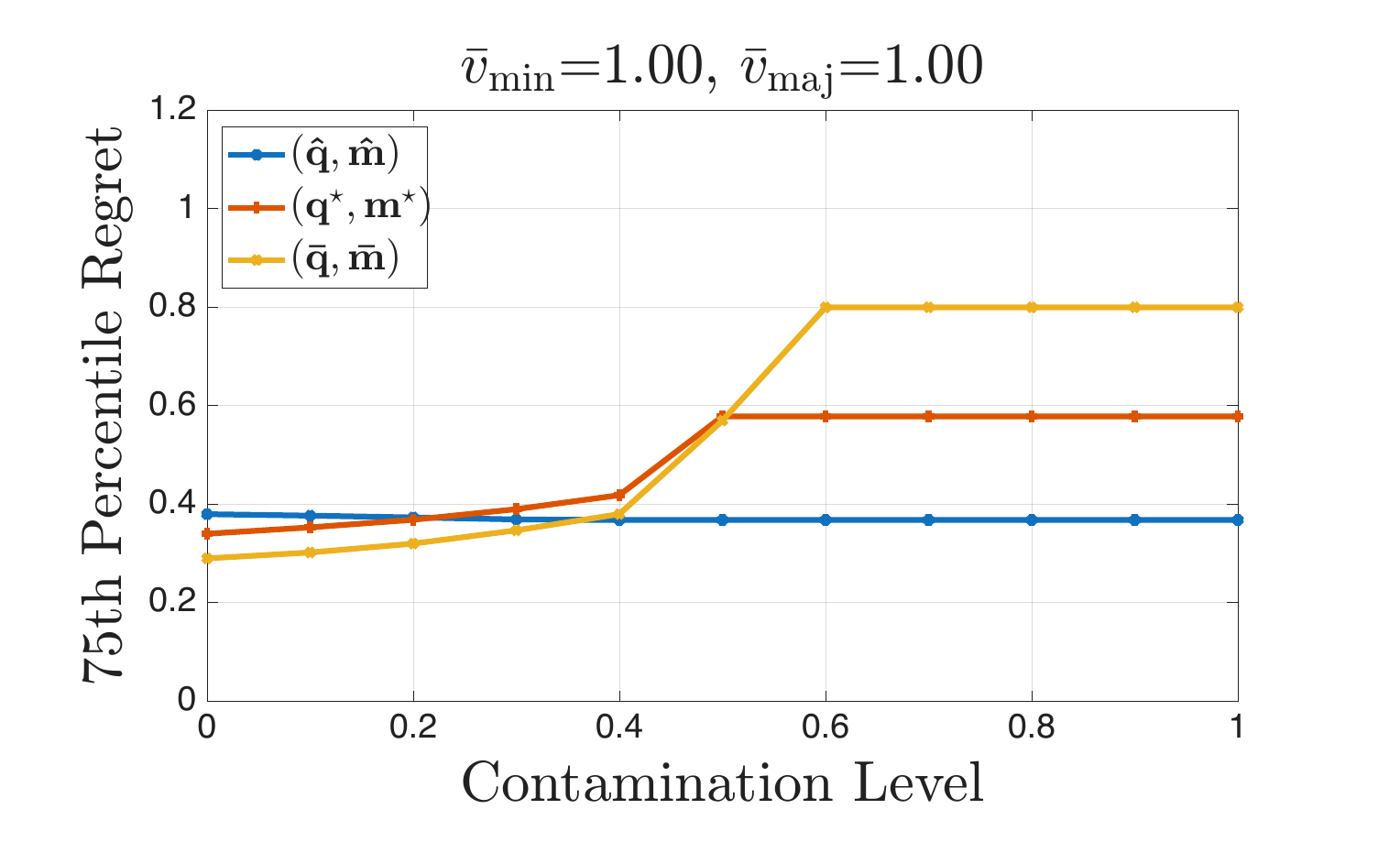}
    
    \includegraphics[scale=0.30]{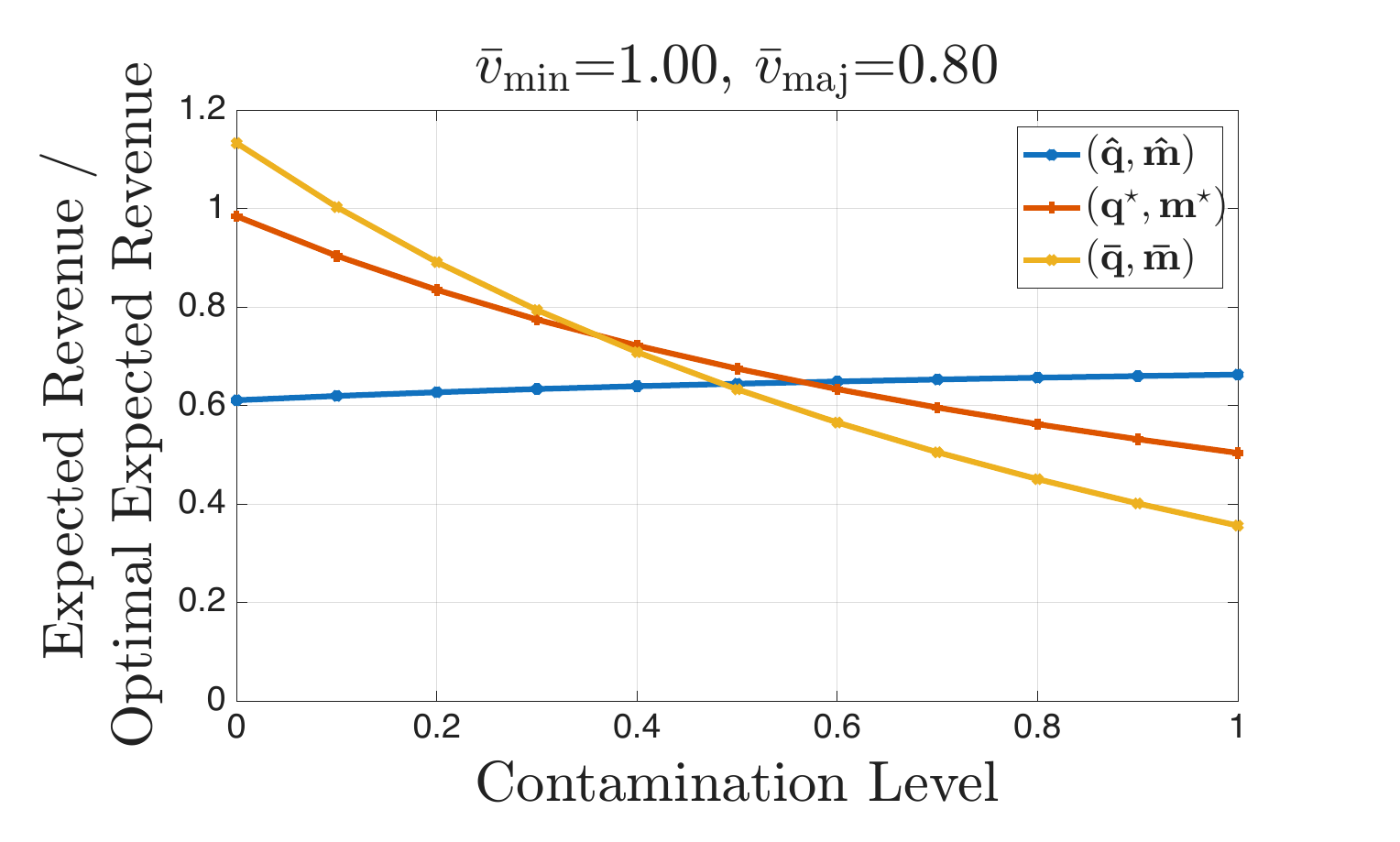}
    \includegraphics[scale=0.30]{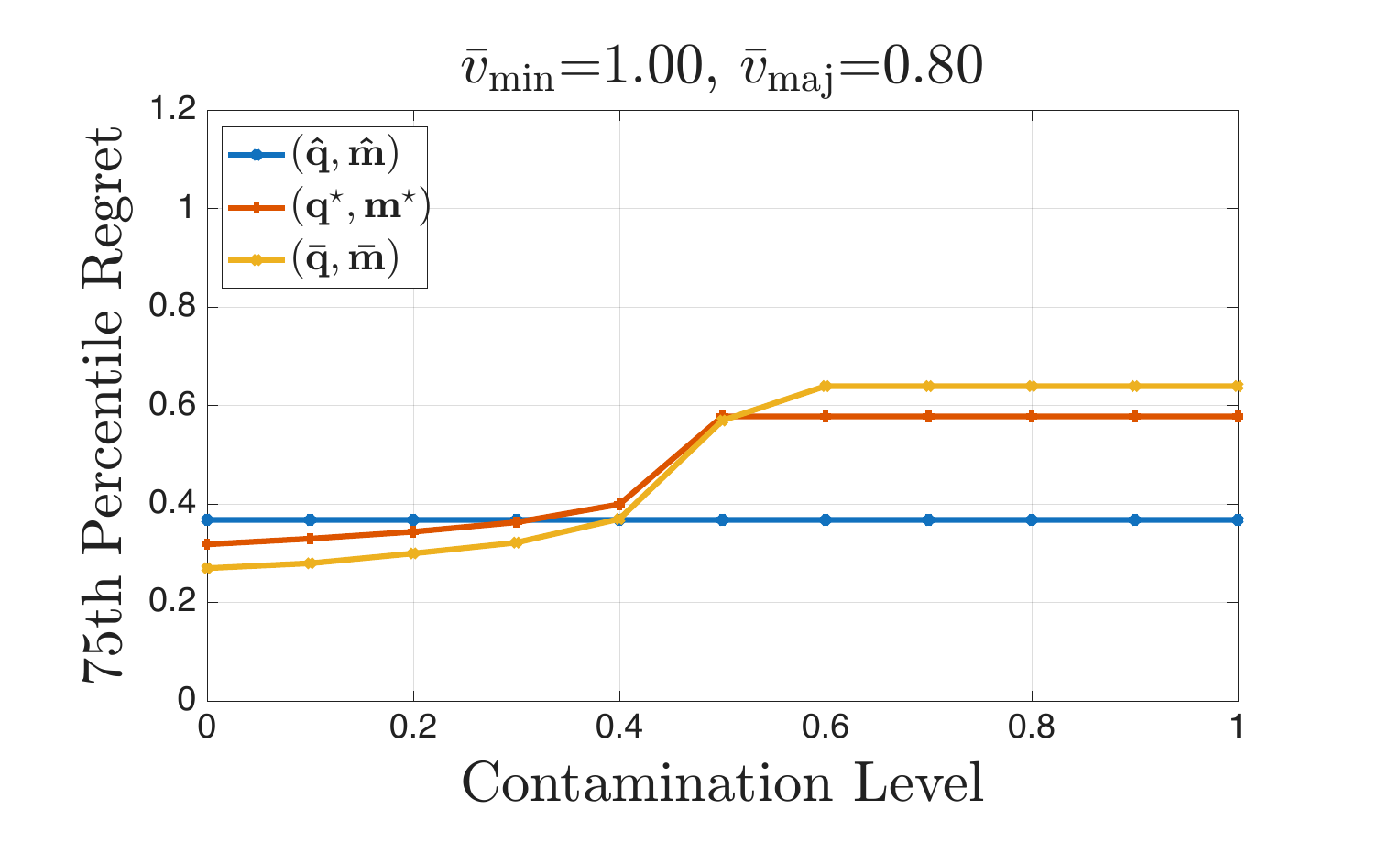}
    \caption{Normalized expected revenues (left) and the upper quartile regrets (right) of the three mechanisms when $\rho = 0$ and $(\overline{v}^{\operatorname{min}},\overline{v}^{\operatorname{maj}}) = (0.8,1),(1,1),(1,0.8)$: the regret-based mechanism $(\hat{\bm{q}},\hat{\bm{m}})$ (blue), the revenue-based mechanism $(\bm{q}^\star,\bm{m}^\star)$ (red) and its variant $(\overline{\bm{q}}, \overline{\bm{m}})$ that enforces equity only in expectation (yellow).}
    \label{fig:zero_rho}
\end{figure}

Figure~\ref{fig:zero_rho} summarizes the normalized expected revenues (left panel) and the 75th percentile regrets (right panel) of $(\hat{\bm{q}}, \hat{\bm{m}}),(\bm{q}^\star, \bm{m}^\star)$ and $(\overline{\bm{q}}, \overline{\bm{m}})$
under different contamination levels and the three support configurations: $(\overline{v}^{\operatorname{min}},\overline{v}^{\operatorname{maj}})\in\{(0.8,1),(1,1),(1,0.8)\}$.
We find (unsurprisingly) that the revenue-based mechanism $(\bm{q}^\star, \bm{m}^\star)$ outperforms the regret-based mechanism $(\hat{\bm{q}}, \hat{\bm{m}})$ when the contamination level $\epsilon$ is small in terms of both metrics. However, we observe that the expected revenue and the 75th percentile regret of $(\hat{\bm{q}}, \hat{\bm{m}})$ remain relatively stable across different contamination levels, highlighting the robustness of the mechanism, and actually outperform those of $(\bm{q}^\star, \bm{m}^\star)$ and $(\overline{\bm{q}}, \overline{\bm{m}})$ as $\epsilon$ increases.
Interestingly, the performance of the mechanism $(\overline{\bm{q}}, \overline{\bm{m}})$ deteriorates significantly (more so in comparison to $(\bm{q}^\star, \bm{m}^\star)$) as $\epsilon$ increases. We attribute this to the mechanism's high reliance on the (incorrect) distribution $\mathbb{F}$ in both the objective and the feasible set. Compared to the other mechanisms, it is more prone to overfitting.

We then repeat the experiment for other $\rho$ values from $\{0.5, -0.5\}$, and the obtained results are presented in Appendix~\ref{appendix: experiment}.
Similar observations can be made for these cases. We observe that, as the contamination level $\epsilon$ grows, the decline in the performance of $(\bm{q}^\star, \bm{m}^\star)$ and $(\overline{\bm{q}}, \overline{\bm{m}})$ is faster when $\rho$ is negative and slower when $\rho$ is positive. The performance of the regret-based mechanism $(\hat{\bm{q}}, \hat{\bm{m}})$ remains relatively stable.

\bibliographystyle{abbrvnat}
\bibliography{bib}

\clearpage

\clearpage
\section*{Appendix}
\renewcommand{\thesection}{\Alph{section}}
\setcounter{section}{0}
\section{Proofs} \label{appendix: proofs}

\begin{proof}{Proof of Theorem \ref{thm:regret_lower}}
The proposed lower bound contains the maximum of two components. In the proof, we show that each of these components, multiplied by the constant $\frac{1}{e}$, constitutes a lower bound on the optimal value $\operatorname{REG}^\star$ of~\eqref{Reg-MDP}. 
For ease of exposition in the proof, we define $\mathcal V^{\rm min} = \times_{i \in \mathcal I^{\operatorname{min}}} \mathcal V_i$ and $\mathcal V^{\rm maj} = \times_{i \in \mathcal I^{\operatorname{maj}}} \mathcal V_i$ as the sets of potential values of the bidders in the minority and majority groups, respectively.

We first prove that $\frac{1}{e} \overline{v}^{{\operatorname{min}}}$ is a lower bound. This relies on the construction of another lower bound on $\operatorname{REG}^\star$ by shrinking the uncertainty set from $\mathcal{V}$ to $\mathcal V' = \{\bm v \in \mathcal V \;|\; v_i = 0 \;\forall i \in \mathcal I^{\operatorname{maj}}\}$. The set $\mathcal{V}'$ can be interpreted as the case where the uncertainty about the majority bidders’ values is removed, and their values are set to zero. We have
\begin{equation}\label{Reg-MDP-lbmin}
\begin{aligned}
    \operatorname{REG}^\star &\geq \min_{(\bm q,\bm m)\in\mathcal{M}} \max_{\bm{v} \in\mathcal V'} \;{\operatorname{REG}}_{(\bm {q}, \bm {m})}(\bm{v}) \\
    & = \min_{(\bm q,\bm m)\in\mathcal{M}} \;\max_{\bm{v} \in\mathcal V'}  \;\;\left(\frac{1}{1+\gamma} \max_{i \in \mathcal I} v_i + \frac{\gamma}{1+\gamma} \max_{i \in \mathcal I^\text{min}} v_i - \sum_{i \in \mathcal I} m_i(\bm v)\right) \\
    &\geq \min_{(\bm q,\bm m)\in\mathcal{M}} \;\max_{\bm{v} \in \mathcal V'}  \;\;\left(\max_{i \in \mathcal I^\text{min}} v_i - \sum_{i \in \mathcal I^{\rm min}} m_i(\bm v)\right),
\end{aligned}
\end{equation}
where the second inequality holds because $m_{i}(\bm {v}) \leq 0$ for all $i \in \mathcal I^{\rm maj}$ as implied by~\eqref{eq:IR} and that the values of these bidders are equal to $0$ across $\mathcal V'$. Note that the last minimax problem above can be interpreted as a variant of~\eqref{Reg-MDP}, involving only quantities related to minority bidders in its objective. 
In particular, the objective function is only a function of the allocations and payments for the minority bidders. In the absence of majority bidders, the equity requirement~\eqref{eq:Eq} also no longer applies, as reflected by the disappearance of the equity parameter $\gamma$. With this observation in mind, to establish a lower-bounding problem for the last minimax problem in~\eqref{Reg-MDP-lbmin}, considering only the allocations and payments of minority bidders as decision variables, we now introduce a new feasible set $\mathcal{M}^{\rm min} \subseteq \mathcal{L}(\mathcal{V}^{\rm min},[0,1]^{I^{\rm min}}) \times \mathcal{L}(\mathcal{V}^{\rm min},\mathbb{R}^{I^{\rm min}})$. This set represents the collection of admissible mechanisms defined solely for minority bidders, i.e., the set of all mechanisms $(\bm{q}^{\operatorname{min}}, \bm{m}^{\operatorname{min}}) \in \mathcal{L}(\mathcal{V}^{\rm min},[0,1]^{I^{\rm min}}) \times \mathcal{L}(\mathcal{V}^{\rm min},\mathbb{R}^{I^{\rm min}})$ that satisfy the following incentive compatibility, individual rationality and allocation feasibility constraints:
\begin{equation*}
\begin{aligned}
        &q^{\rm min}_i(v_i,\bm{v}_{-i}) v_i  - m^{\rm min}_i(v_i, \bm{v}_{-i})  \geq  q^{\rm min}_i(w_i,\bm{v}_{-i}) v_i - m^{\rm min}_i(w_i, \bm{v}_{-i}) \;\, \forall i \in \mathcal{I}^{\rm min}, \forall \bm v \in \mathcal{V}^{\rm min}, \forall w_i \in \mathcal V_i\\
&q^{\rm min}_i(v_i,\bm{v}_{-i}) v_i  - m^{\rm min}_i(v_i, \bm{v}_{-i}) \geq 0 \;\; \forall i \in \mathcal{I}^{\rm min}, \forall \bm v \in \mathcal{V}^{\rm min}\\
&\sum_{i \in \mathcal{I}^{\rm min}} q^{\rm min}_i(\bm v) \leq 1 \;\; \forall \bm v \in \mathcal{V}^{\rm min}.
\end{aligned}
\end{equation*}
We next aim to show that, for any $(\bm{q},\bm{m}) \in \mathcal{M}$, we can construct a $(\bm{q}^{\operatorname{min}}, \bm{m}^{\operatorname{min}}) \in \mathcal{M}^{\rm min}$ such that the allocations and payments for the minority bidders coincide under the two mechanisms across~$\mathcal V'$. This will allow us to replace the minimum over $(\bm{q},\bm{m}) \in \mathcal{M}$ in the last minimax problem in~\eqref{Reg-MDP-lbmin} with a minimum over $(\bm{q}^{\operatorname{min}}, \bm{m}^{\operatorname{min}}) \in \mathcal{M}^{\rm min}$ to establish a lower bound. To this end, we define the vector-valued function $f \in \mathcal{L}(\mathcal{V}^{\rm min}, \mathcal V)$ through  
$$f(\bm v) = (v_1 , \dots, v_{I^{\rm min}}, 0, \dots, 0) \;\;\forall \bm v \in\mathcal{V}^{\rm min}.$$ 
Now, for any $(\bm{q},\bm{m}) \in \mathcal{M}$, define $(\bm{q}^{\operatorname{min}}, \bm{m}^{\operatorname{min}})$ through 
$${q}_i^{\operatorname{min}}(\bm v) = q_i(f(\bm v)), \;\; {m}_i^{\operatorname{min}}(\bm v) = m_i(f(\bm v)) \;\;\forall i \in \mathcal{I}^{\rm min}, \,\forall \bm v \in \mathcal{V}^{\rm min}.$$ 
It can be readily verified that $(\bm{q}^{\operatorname{min}}, \bm{m}^{\operatorname{min}}) \in \mathcal{M}^{\rm min}$ and that the allocations and payments for the minority bidders under $(\bm{q},\bm{m})$ remain the same across $\mathcal V'$. Therefore, the last line in~\eqref{Reg-MDP-lbmin} can be bounded below by 
\begin{equation*}
\begin{aligned}
    &\min_{(\bm{q}^{\operatorname{min}}, \bm{m}^{\operatorname{min}})\in\mathcal{M}^{\rm min}} \max_{\bm{v} \in\mathcal V'} \;\; \left(\max_{i \in \mathcal I^\text{min}} v_i - \sum_{i \in \mathcal I^{\rm min}} m^{\operatorname{min}}_i(v_1, \dots, v_{{I}^{\rm min}})\right)\\
    &= \min_{(\bm{q}^{\operatorname{min}}, \bm{m}^{\operatorname{min}})\in\mathcal{M}^{\rm min}} \max_{\bm{v} \in\mathcal V^{\rm min}} \;\; \left(\max_{i \in \mathcal I^\text{min}} v_i - \sum_{i \in \mathcal I^{\rm min}} m^{\operatorname{min}}_i(\bm v)\right) = \frac{1}{e}\overline{v}^{\rm min},
    \end{aligned}
\end{equation*}
where the first equality holds because $(\bm{q}^{\operatorname{min}}, \bm{m}^{\operatorname{min}})$ does not rely on the majority bidders' values. The second equality follows from Theorem~4 of~\cite{koccyiugit2024regret}, which characterizes the optimal worst-case regret in multi-bidder auctions without fairness considerations, thereby proving the claim $\frac{1}{e} \overline{v}^{{\operatorname{min}}}$ is a lower bound.

To verify the second lower bound $\frac{1}{e}\frac{\overline{v}^{{\operatorname{maj}}}}{1+\gamma}$, we obtain a lower bound on $\operatorname{REG}^\star$ by shrinking the uncertainty set from $\mathcal{V}$ to $\mathcal V'' = \{\bm v \in \mathcal V \;|\; v_i = 0 \;\forall i \in \mathcal I^{\operatorname{min}}\}$. We have
\begin{equation}\label{Reg-MDP-lbmaj}
\begin{aligned}
    \operatorname{REG}^\star &\geq \min_{(\bm q,\bm m)\in\mathcal{M}} \;\max_{\bm{v} \in\mathcal V''} \;{\text{REG}}_{(\bm {q}, \bm {m})}(\bm{v}) \\
    &=\min_{(\bm q,\bm m)\in\mathcal{M}} \;\max_{\bm{v} \in\mathcal V''}\; \left( \frac{1}{1+\gamma} \max_{i \in \mathcal I} v_i + \frac{\gamma}{1+\gamma} \max_{i \in \mathcal I^\text{min}} v_i - \sum_{i \in \mathcal I} m_i(\bm v) \right) \\
    &\geq \min_{(\bm q,\bm m)\in\mathcal{M}} \;\max_{\bm{v} \in\mathcal V''} \;\left( \frac{1}{1+\gamma}\max_{i \in \mathcal I^{\rm maj}} v_i - \sum_{i \in \mathcal I^{\rm maj}} m_i(\bm v) \right),
\end{aligned}
\end{equation}
where the second inequality holds because $m_{i}(\bm {v}) \leq 0$ for all $i \in \mathcal I^{\rm min}$ as implied by~\eqref{eq:IR} and that the values of these bidders are equal to $0$ across $\mathcal V''$. Similarly to the first part, note that the last minimax problem above can be interpreted as a variant of~\eqref{Reg-MDP} that involves only quantities related to majority bidders in its objective. 
The key difference, however, is that the equity parameter $\gamma$ remains present in this case. Next, we similarly introduce a new feasible set $\mathcal{M}^{\rm maj} \subseteq \mathcal{L}(\mathcal{V}^{\rm maj},[0,1]^{I^{\rm maj}}) \times \mathcal{L}(\mathcal{V}^{\rm maj},\mathbb{R}^{I^{\rm maj}})$, which represents the collection of admissible mechanisms defined solely for majority bidders, i.e., the set of all mechanisms $(\bm{q}^{\operatorname{maj}}, \bm{m}^{\operatorname{maj}}) \in \mathcal{L}(\mathcal{V}^{\rm maj},[0,1]^{I^{\rm maj}}) \times \mathcal{L}(\mathcal{V}^{\rm maj},\mathbb{R}^{I^{\rm maj}})$ that satisfy the following incentive compatibility, individual rationality and allocation feasibility constraints:
\begin{equation*}
\begin{aligned}
        &q^{\rm maj}_i(v_i,\bm{v}_{-i}) v_i  - m^{\rm maj}_i(v_i, \bm{v}_{-i})  \geq  q^{\rm maj}_i(w_i,\bm{v}_{-i}) v_i - m^{\rm maj}_i(w_i, \bm{v}_{-i}) \;\, \forall i \in \mathcal{I}^{\rm maj}, \forall \bm v \in \mathcal{V}^{\rm maj}, \forall w_i \in \mathcal V_i\\
&q^{\rm maj}_i(v_i,\bm{v}_{-i}) v_i  - m^{\rm maj}_i(v_i, \bm{v}_{-i}) \geq 0 \;\; \forall i \in \mathcal{I}^{\rm maj}, \forall \bm v \in \mathcal{V}^{\rm maj}\\
&\sum_{i \in \mathcal{I}^{\rm maj}} q^{\rm maj}_i(\bm v) \leq 1 \;\; \forall \bm v \in \mathcal{V}^{\rm maj}.
\end{aligned}
\end{equation*}
We will show that the last minimax problem in~\eqref{Reg-MDP-lbmaj} can be bounded below by
\begin{equation*}
    \frac{1}{1+\gamma}\min_{(\bm q^{\operatorname{maj}},\bm m^{\operatorname{maj}})\in\mathcal{M}^{\rm maj}} \;\max_{\bm{v} \in\mathcal V^{\rm maj}}\; \left( \max_{i \in \mathcal I^{\rm maj}} v_i - \sum_{i \in \mathcal I^{\rm maj}} m^{\operatorname{maj}}_i(\bm v) \right).
\end{equation*}
To this end, consider any $(\bm{q},\bm{m}) \in \mathcal{M}$ and note that the total allocation to the majority bidders is uniformly bounded above by $\frac{1}{1+\gamma}$ due to the equity constraint~\eqref{eq:Eq}. We construct a $(\bm{q}^{\operatorname{maj}}, \bm{m}^{\operatorname{maj}}) \in \mathcal{M}^{\rm maj}$ from $(\bm{q},\bm{m})$ as follows. Let $g \in \mathcal{L}(\mathcal{V}^{\rm maj}, \mathcal V)$ be a vector-valued function defined through  
$$g(\bm v) = (\underbrace{0, \dots, 0}_{I^{\rm min}}, v_1 , \dots, v_{I^{\rm maj}}) \;\;\forall \bm v \in\mathcal{V}^{\rm maj}.$$
Define $(\bm{q}^{\operatorname{maj}}, \bm{m}^{\operatorname{maj}})$ through 
$${q}_i^{\operatorname{maj}}(\bm v) = (1 + \gamma)q_i(g(\bm v)), \;\; {m}_i^{\operatorname{maj}}(\bm v) = (1 + \gamma) m_i(g(\bm v)) \;\;\forall i \in \mathcal{I}^{\rm maj}, \,\forall \bm v \in \mathcal{V}^{\rm maj}.$$ 
Since incentive compatibility and individual rationality constraints are preserved under non-negative scaling of the mechanism, and because $\bm{q}$ is uniformly bounded above by $\frac{1}{1+\gamma}$, it can be readily verified that $(\bm{q}^{\operatorname{maj}}, \bm{m}^{\operatorname{maj}}) \in \mathcal{M}^{\rm maj}$. Additionally, the allocations and payments under $(\bm{q}^{\operatorname{maj}}, \bm{m}^{\operatorname{maj}})$ are equal to $(1 + \gamma)$ times those under $(\bm{q},\bm{m})$ across $\mathcal V''$. Therefore, the last line in~\eqref{Reg-MDP-lbmaj} can be bounded below by 
\begin{equation*}
    \begin{aligned}
        &\min_{(\bm{q}^{\operatorname{maj}}, \bm{m}^{\operatorname{maj}}) \in \mathcal{M}^{\rm maj}} \;\max_{\bm{v} \in\mathcal V''} \;\left( \frac{1}{1+\gamma}\max_{i \in \mathcal I^{\rm maj}} v_i - \frac{1}{1+\gamma}\sum_{i \in \mathcal I^{\rm maj}} m^{\operatorname{maj}}_i(v_{I^{\rm min}+1}, \dots, v_{I}) \right)\\
        &= \frac{1}{1+\gamma} \min_{(\bm{q}^{\operatorname{maj}}, \bm{m}^{\operatorname{maj}}) \in \mathcal{M}^{\rm maj}} \;\max_{\bm{v} \in\mathcal V^{\operatorname{maj}}} \;\left(\max_{i \in \mathcal I^{\rm maj}} v_i - \sum_{i \in \mathcal I^{\rm maj}} m^{\operatorname{maj}}_i(\bm v) \right),
    \end{aligned}
\end{equation*}
where the first inequality follows from~\eqref{Reg-MDP-lbmaj}, and the last equality holds because $(\bm{q}^{\operatorname{maj}}, \bm{m}^{\operatorname{maj}})$ does not rely on minority bidders' values. Finally, analogously to the first part, we observe that
\begin{equation*}
    \begin{aligned}
    \min_{(\bm{q}^{\operatorname{maj}}, \bm{m}^{\operatorname{maj}}) \in \mathcal{M}^{\rm maj}} \;\max_{\bm{v} \in\mathcal V^{\operatorname{maj}}} \;\left(\max_{i \in \mathcal I^{\rm maj}} v_i - \sum_{i \in \mathcal I^{\rm maj}} m^{\operatorname{maj}}_i(\bm v) \right) = \frac{1}{e}\overline{v}^{\rm maj},
    \end{aligned}
\end{equation*}
thanks to Theorem~4 of~\cite{koccyiugit2024regret}. This observation concludes the second part of the proof and, together with the first part, completes the proof. 
\end{proof}

\begin{proof}{Proof of Proposition \ref{prop:feasibility of robust case}}
We will first prove that the constraint \eqref{eq:IC} holds under the mechanism $(\bm{\hat{q}},\bm{\hat{m}})$ for both minority bidders {(\it Case 1)} and majority bidders {(\it Case 2)}.  
Then, we will establish the feasibility of the same mechanism in view of the other constraints, namely~\eqref{eq:IR},~\eqref{eq:Eq} and ~\eqref{eq:AF}, to conclude the proof.

To prove that $(\bm{\hat{q}},\bm{\hat{m}})$ satisfies the \eqref{eq:IC} constraint, it is sufficient to show that conditions (i) and (ii) in Lemma \ref{lem: monotonicity of allocation} hold. By the definition of $\bm{\hat{m}}$, condition (ii) in Lemma~\ref{lem: monotonicity of allocation} immediately holds. 
Hence, it remains to establish condition~(i), that is, to show that the allocation rule ${\hat{q}}_i$ is non-decreasing in $v_i$ for all $i\in\mathcal{I}$.

{\it Case 1} ($i \in \minority$):
By construction, we can re-express the allocation~$\hat{q}_i$ as 
\begin{equation*}
\begin{aligned}
    \hat{q}_i(\bm{v}) &= \begin{cases}
        \dfrac{\gamma}{1+\gamma}\left(  1 + \log \left({{\dfrac{\ell(\bm{v})}{\overline{\ell}}}}\right) \right)^+\mathbb{I}_{\{i={i^{\operatorname{min}}(\bm{v})}\}} & \text{ if }  v_i < v_{{i^{\operatorname{maj}}(\bm{v})}} \\
        \left(1 + \log \left(\dfrac{v_i}{\overline{\ell}}\right) \right)^+\mathbb{I}_{\{i={i^{\operatorname{min}}(\bm{v})}\}} & \text{ otherwise. }
    \end{cases}
\end{aligned}
\end{equation*}
Consider an arbitrary fixed $\bm{v}_{-i}$. We aim to show that $\hat{q}_i$ is non-decreasing in $v_i$ given the fixed~$\bm{v}_{-i}$. Note that the value $v_{i^{\rm maj}}(\bm{v})$ is not affected by the values of the minority bidders. This means that the value $v_{i^{\rm maj}}(\bm{v})$ is a fixed quantity for the given fixed $\bm{v}_{-i}$. To avoid ambiguity, we denote this value by~$s$.
Note that $\ell(v_i,\bm{v}_{-i})$ and the indicator function $\mathbb{I}_{\{i={i^{\operatorname{min}}(v_i,\bm{v}_{-i})}\}}$ are non-decreasing in~$v_i$. Since logarithm function is increasing, our $\hat{q}_i(\cdot,\bm{v}_{-i})$ is non-decreasing in intervals $[0, s)$ and $[s, +\infty)$. 
It thus follows that $\hat{q}_i(\cdot,\bm{v}_{-i})$ is non-decreasing in the interval $\mathcal{V}_i = [0,\overline{v}_i]$ if $\overline{v}_i< s$. For the case that $\overline{v}_i\geq s$, suppose, for the sake of contradiction, that there exist $\alpha \in [0,s)$ and $\beta \in [s,\overline{v}_i]$ such that 
$\hat{q}_i(\alpha,\bm{v}_{-i}) > \hat{q}_i(\beta,\bm{v}_{-i}) \geq 0$. This means that $\mathbb{I}_{\{i={i^{\operatorname{min}}(\alpha,\bm{v}_{-i})}\}} = 1$ because otherwise $\hat{q}_i(\alpha,\bm{v}_{-i}) = 0$, and this further implies that $\mathbb{I}_{\{i={i^{\operatorname{min}}(\beta,\bm{v}_{-i})}\}} = 1$ as $\beta > \alpha$. Therefore, $\hat{q}_i(\alpha,\bm{v}_{-i}) > \hat{q}_i(\beta,\bm{v}_{-i})$ holds true if and only if
\begin{equation*} 
\begin{aligned}
    \frac{\gamma}{1+\gamma} \left( 1 + \log \left(\dfrac{\ell(\alpha,\bm{v}_{-i})}{\overline\ell}\right)  \right)^+ 
    > \left( 1 + \log \left(\dfrac{\ell(\beta,\bm{v}_{-i})}{\overline\ell} \right)\right)^+ .
\end{aligned}
\end{equation*}
However, the above inequality cannot be satisfied because
\begin{equation*}
    \frac{\gamma}{1+\gamma} \left( 1 + \log \left(\dfrac{\ell(\alpha,\bm{v}_{-i})}{\overline\ell}\right)  \right)^+
    \leq \left( 1 + \log \left(\dfrac{\ell(\alpha,\bm{v}_{-i})}{\overline\ell}\right)  \right)^+
    \leq \left( 1 + \log \left(\dfrac{s}{\overline\ell}\right)  \right)^+,
\end{equation*}
and
\begin{equation*}
    \left( 1 + \log \left(\dfrac{\ell(\beta,\bm{v}_{-i})}{\overline\ell} \right) \right)^+ 
    \geq \left( 1 + \log \left(\dfrac{s}{\overline\ell}\right)  \right)^+.
\end{equation*}
This contradiction confirms that $\hat{q}_i(v_i, \bm{v}_{-i})$ is non-decreasing in $v_i$ over $\mathcal V_i = [0,\overline{v}_i]$ for any fixed~$\bm{v}_{-i}$.

{\it Case 2} ($i \in \majority$):
By construction, we can similarly express the allocation $\hat{q}_i$ as
\begin{equation*}
\begin{aligned}
    \hat{q}_i(\bm{v}) &= \begin{cases}
        \dfrac{1}{1+\gamma}\left(  1 + \log \left({{\dfrac{\ell(\bm{v})}{\overline{\ell}}}}\right) \right)^+\mathbb{I}_{\{i={i^{\operatorname{maj}}(\bm{v})}\}} & \text{ if }  v_i > v_{{i^{\operatorname{min}}(\bm{v})}}\\
        0 & \text{ otherwise. } 
    \end{cases} \\
\end{aligned}
\end{equation*}
For an arbitrary fixed $\bm{v}_{-i}$, which uniquely determines $v_{i^{\rm min}}(\bm{v})$, as $\ell(\bm{v})$ and the indicator function $\mathbb{I}_{\{i={i^{\operatorname{maj}}(\bm{v})}\}}$ are non-decreasing in $v_i$, it can be seen that~$\hat{q}_i$ is non-decreasing in $v_i$.

Constraints ~\eqref{eq:IR},~\eqref{eq:Eq} and ~\eqref{eq:AF}  trivially hold due to the construction of the mechanism~$(\bm{\hat{q}},\bm{\hat{m}})$.
Indeed, this mechanism is individually rational because
$$v_i{\hat{q}}_i(\bm{v})-{\hat{m}}_i(\bm{v}) = \int_{0}^{v_i}{\hat{q}}_i(x,\bm{v}_{-i}) {\rm{d}} x \geq 0  \quad \forall i \in \mathcal{I},$$
and it is equitable as
\begin{equation*}\label{eqn: multi-bidder sum qmin} 
    \sum_{i\in\mathcal{I}^{\operatorname{min}}}{\hat{q}}_i(\bm{v}) = \begin{cases}
        \dfrac{\gamma}{1+\gamma}\left(  1 + \log \left({{\dfrac{\ell(\bm{v})}{\overline{\ell}}}}\right) \right)^+ & \text{ if } v_{{i^{\operatorname{maj}}(\bm{v})}} > v_{{i^{\operatorname{min}}(\bm{v})}} \\
        \left(  1 + \log \left(\dfrac{v_{{i^{\operatorname{min}}(\bm{v})}}}{\overline{\ell}} \right)\right)^+ & \text{ otherwise,}
    \end{cases}
\end{equation*}
and
\begin{equation*}\label{eqn: multi-bidder sum qmaj}             
    \sum_{i\in\mathcal{I}^{\operatorname{maj}}}{\hat{q}}_i(\bm{v}) = \begin{cases}
        \dfrac{1}{1+\gamma}\left(  1 + \log \left({{\dfrac{\ell(\bm{v})}{\overline{\ell}}}}\right) \right)^+ & \text{ if } v_{{i^{\operatorname{maj}}(\bm{v})}} > v_{{i^{\operatorname{min}}(\bm{v})}} \\
        0 & \text{ otherwise.}
    \end{cases}
\end{equation*}
From $\sum_{i\in\mathcal{I}^{\operatorname{min}}}{\hat{q}}_i(\bm{v})$ and $\sum_{i\in\mathcal{I}^{\operatorname{maj}}}{\hat{q}}_i(\bm{v})$ computed above, we also obtain
\begin{equation*}
    \sum_{i\in\mathcal{I}}{\hat{q}}_i(\bm{v}) = \sum_{i\in\mathcal{I}^{\operatorname{min}}}{\hat{q}}_i(\bm{v})+ \sum_{i\in\mathcal{I}^{\operatorname{maj}}}{\hat{q}}_i(\bm{v}) = \begin{cases}
        \left(  1 + \log \left({{\dfrac{\ell(\bm{v})}{\overline{\ell}}}}\right) \right)^+ & \text{ if } v_{{i^{\operatorname{maj}}(\bm{v})}} > v_{{i^{\operatorname{min}}(\bm{v})}} \\
        \left(  1 + \log \left(\dfrac{v_{{i^{\operatorname{min}}(\bm{v})}}}{\overline{\ell}}\right) \right)^+ & \text{ otherwise,}
    \end{cases}
\end{equation*}
and hence the total allocation probability does not exceed one as both $\ell(\bm{v})$ and $v_{i^{\operatorname{min}}(\bm{v})}$ do not exceed ${\overline{\ell}}$. As the mechanism $(\bm{\hat{q}},\bm{\hat{m}})$ satisfies all constraints in~$\mathcal M$, the feasibility claim follows.
\end{proof}

\begin{proof}{Proof of \Cref{thm:regret_upper}}
Fixing an arbitrary $\bm v \in \mathcal V$, we will 
prove the upper bound separately for two cases: $v_{{i^{\operatorname{maj}}(\bm{v})}} > v_{{i^{\operatorname{min}}(\bm{v})}}$ ({\it Case 1}) and  $v_{{i^{\operatorname{maj}}(\bm{v})}} \leq v_{{i^{\operatorname{min}}(\bm{v})}}$ ({\it Case 2}).

\textit{Case 1} ($v_{{i^{\operatorname{maj}}(\bm{v})}} > v_{{i^{\operatorname{min}}(\bm{v})}}$): By the definition of~$\ell(\bm v)$, we can write $\text{REG}_{(\bm {\hat{q}}, \bm {\hat{m}})}(\bm{v})$ as
\begin{equation*}
    \text{REG}_{(\bm {\hat{q}}, \bm {\hat{m}})}(\bm{v}) = 
    \ell (\bm{v}) - \sum_{i\in\mathcal{I}^{\operatorname{min}}} \hat{m}_{i}(\bm{v}) - \sum_{i\in\mathcal{I}^{\operatorname{maj}}} \hat{m}_{i}(\bm{v}) = 
    \ell (\bm{v}) - \hat{m}_{i^{\operatorname{min}}(\bm v)}(\bm{v}) - \hat{m}_{i^{\operatorname{maj}}(\bm v)}(\bm{v}),
\end{equation*}
where the second equality follows from the definition of $(\bm {\hat{q}}, \bm {\hat{m}})$, under which only the highest bidders from each group can receive a positive share of the item and may be charged non-zero payments.
By the definition of $\bm {\hat{m}}$, the regret above can be equivalently expressed as 
\begin{equation}
\label{eq:regret_1stcase}
\begin{aligned}
    \text{REG}_{(\bm {\hat{q}}, \bm {\hat{m}})}(\bm{v}) =
    \ell(\bm{v}) &- v_{{i^{\operatorname{min}}(\bm{v})}}{\hat{q}}_{{i^{\operatorname{min}}(\bm{v})}}(\bm{v})+\int_{0}^{v_{{i^{\operatorname{min}}(\bm{v})}}}{\hat{q}}_{{i^{\operatorname{min}}(\bm{v})}}(x,\bm{v}_{-{{i^{\operatorname{min}}(\bm{v})}}}) {\rm{d}} x \\
    &- v_{{i^{\operatorname{maj}}(\bm{v})}}{\hat{q}}_{{i^{\operatorname{maj}}(\bm{v})}}(\bm{v})+\int_{0}^{v_{{i^{\operatorname{maj}}(\bm{v})}}}{\hat{q}}_{{i^{\operatorname{maj}}(\bm{v})}}(x,\bm{v}_{-{{i^{\operatorname{maj}}(\bm{v})}}}) {\rm{d}} x.
\end{aligned}
\end{equation}
To bound $\text{REG}_{(\bm {\hat{q}}, \bm {\hat{m}})}(\bm{v})$, we first study the two integrals on the right-hand side of \eqref{eq:regret_1stcase} individually. For the first integral, we have 
\begin{equation}
\allowdisplaybreaks
\label{eq:regret_1stcase_1stint}
\begin{aligned}
    &\int_{0}^{v_{{i^{\operatorname{min}}(\bm{v})}}}{\hat{q}}_{{i^{\operatorname{min}}(\bm{v})}}(x,\bm{v}_{-{{i^{\operatorname{min}}(\bm{v})}}}) {\rm{d}} x \\
    &= \int_{0}^{v_{{i^{\operatorname{min}}(\bm{v})}}}\dfrac{\gamma}{1+\gamma}\left(  1 + \log \left(\dfrac{\frac{\gamma}{1+\gamma}x+\frac{1}{1+\gamma}v_{{i^{\operatorname{maj}}(\bm{v})}}}{\overline{\ell}}\right) \right)^+ \mathds{I}_{\{x\geq \max_{j \in \mathcal{I}^{\operatorname{min}} \setminus \{ i^{\operatorname{min}}(\bm{v}) \} } v_j \}}{\rm{d}} x \\
    &\leq \int_{0}^{v_{{i^{\operatorname{min}}(\bm{v})}}}\dfrac{\gamma}{1+\gamma}\left(  1 + \log \left(\dfrac{\frac{\gamma}{1+\gamma}x+\frac{1}{1+\gamma}v_{{i^{\operatorname{maj}}(\bm{v})}}}{\overline{\ell}}\right) \right)^+ {\rm{d}} x \\
    &= \int_{\frac{1}{1+\gamma}v_{{i^{\operatorname{maj}}(\bm{v})}}}^{\frac{\gamma}{1+\gamma}v_{{i^{\operatorname{min}}(\bm{v})}}+\frac{1}{1+\gamma}v_{{i^{\operatorname{maj}}(\bm{v})}}}\left(  1 + \log \left(\dfrac{y}{\overline{\ell}}\right) \right)^+ {\rm{d}} y \\
    &= \int_{\frac{v_{{i^{\operatorname{maj}}(\bm{v})}}}{(1+\gamma)\overline{\ell}}}^{\frac{\frac{\gamma}{1+\gamma}v_{{i^{\operatorname{min}}(\bm{v})}}+\frac{1}{1+\gamma}v_{{i^{\operatorname{maj}}(\bm{v})}}}{\overline{\ell}}}{\overline{\ell}}\left(  1 + \log z \right)^+ {\rm{d}} z 
    = \int_{\frac{v_{{i^{\operatorname{maj}}(\bm{v})}}}{(1+\gamma)\overline{\ell}}}^{{\frac{\ell(\bm{v})}{\overline{\ell}}}}{\overline{\ell}}\left(  1 + \log z \right)^+ {\rm{d}} z,
\end{aligned}
\end{equation}
where the first equality follows from the definition of $\hat{q}_{i^{\operatorname{min}}(\bm v)}$, and the next two equalities follow from affine transformations
$y = \frac{\gamma}{1+\gamma}x + \frac{1}{1+\gamma}v_{{i^{\text{maj}(\bm{v})}}}$ and $z = \frac{y}{\overline{\ell}}$, respectively. Finally, the last equality follows from the definition of $\ell(\bm v)$. 

For the second integral on the right-hand side of \eqref{eq:regret_1stcase}, we have
\begin{equation}
\allowdisplaybreaks
\label{eq:regret_1stcase_2ndint}
\begin{aligned}
    &\int_{0}^{v_{{i^{\operatorname{maj}}(\bm{v})}}}{\hat{q}}_{{i^{\operatorname{maj}}(\bm{v})}}(x,\bm{v}_{-{{i^{\operatorname{maj}}(\bm{v})}}}) {\rm{d}} x \\
    &= \int_{v_{{i^{\operatorname{min}}(\bm{v})}}}^{v_{{i^{\operatorname{maj}}(\bm{v})}}}\dfrac{1}{1+\gamma}\left(  1 + \log \left(\dfrac{\frac{\gamma}{1+\gamma}v_{{i^{\operatorname{min}}(\bm{v})}}+\frac{1}{1+\gamma}x}{\overline{\ell}}\right) \right)^+ \mathds{I}_{\{x\geq \max_{j \in \mathcal{I}^{\operatorname{maj}} \setminus \{ i^{\operatorname{maj}}(\bm{v}) \} } v_j \}}{\rm{d}} x \\
    &\leq \int_{v_{{i^{\operatorname{min}}(\bm{v})}}}^{v_{{i^{\operatorname{maj}}(\bm{v})}}}\dfrac{1}{1+\gamma}\left(  1 + \log \left(\dfrac{\frac{\gamma}{1+\gamma}v_{{i^{\operatorname{min}}(\bm{v})}}+\frac{1}{1+\gamma}x}{\overline{\ell}}\right) \right)^+ {\rm{d}} x \\
    &= \int_{v_{{i^{\operatorname{min}}(\bm{v})}}}^{\frac{1}{1+\gamma}v_{{i^{\operatorname{maj}}(\bm{v})}}+\frac{\gamma}{1+\gamma}v_{{i^{\operatorname{min}}(\bm{v})}}}\left(  1 + \log \left(\dfrac{y}{\overline{\ell}}\right) \right)^+ {\rm{d}} y\\ 
    &= \int_{\frac{v_{{i^{\operatorname{min}}(\bm{v})}}}{\overline{\ell}}}^{\frac{\frac{1}{1+\gamma}v_{{i^{\operatorname{maj}}(\bm{v})}}+\frac{\gamma}{1+\gamma}v_{{i^{\operatorname{min}}(\bm{v})}}}{\overline{\ell}}}{\overline{\ell}}\left(  1 + \log z \right)^+ {\rm{d}} z
    = \int_{\frac{v_{{i^{\operatorname{min}}(\bm{v})}}}{\overline{\ell}}}^{\frac{\ell(\bm{v})}{\overline{\ell}}}{\overline{\ell}}\left(  1 + \log z \right)^+ {\rm{d}} z,
\end{aligned}    
\end{equation}
where, similarly to before, the first equality follows from the definition of $\hat{q}_{i^{\operatorname{maj}}(\bm v)}$, and the next two equalities follow from affine transformations
$y = \frac{\gamma}{1+\gamma}v_{{i^{\operatorname{min}}(\bm{v})}} + \frac{1}{1+\gamma}x$ and $z = \frac{y}{\overline{\ell}}$, respectively. In addition, the last equality follows from the definition of $\ell(\bm v)$. We thus bound the two integrals from above.

Next, we consider the term $v_{{i^{\operatorname{min}}(\bm{v})}}{\hat{q}}_{{i^{\operatorname{min}}(\bm{v})}}(\bm{v}) + v_{{i^{\operatorname{maj}}(\bm{v})}}{\hat{q}}_{{i^{\operatorname{maj}}(\bm{v})}}(\bm{v})$, which also appears on the right-hand side of~\eqref{eq:regret_1stcase}. This term can be expressed equivalently as
\begin{equation}
\label{eq:regret_1stcase_3rd}
\begin{aligned}
    &v_{{i^{\operatorname{min}}(\bm{v})}}{\hat{q}}_{{i^{\operatorname{min}}(\bm{v})}}(\bm{v}) + v_{{i^{\operatorname{maj}}(\bm{v})}}{\hat{q}}_{{i^{\operatorname{maj}}(\bm{v})}}(\bm{v})\\
    &=   \dfrac{\gamma}{1+\gamma}v_{{i^{\operatorname{min}}(\bm{v})}}\left(  1 + \log \left({{\dfrac{\ell(\bm{v})}{\overline{\ell}}}}\right) \right)^+ +   \dfrac{1}{1+\gamma}v_{{i^{\operatorname{maj}}(\bm{v})}}\left(  1 + \log \left({{\dfrac{\ell(\bm{v})}{\overline{\ell}}}}\right)\right)^+ \\
    &= \ell(\bm{v})\left(  1 + \log \left({{\dfrac{\ell(\bm{v})}{\overline{\ell}}}}\right)\right)^+,
    \end{aligned}
\end{equation}
where the first equality follows from the definitions of ${\hat{q}}_{{i^{\operatorname{min}}(\bm{v})}}$ and ${\hat{q}}_{{i^{\operatorname{maj}}(\bm{v})}}$, and the second equality follows from the definition of $\ell(\bm{v})$.

Leveraging the algebraic insights from~\eqref{eq:regret_1stcase_1stint},~\eqref{eq:regret_1stcase_2ndint} and~\eqref{eq:regret_1stcase_3rd}, we can now bound the regret in~\eqref{eq:regret_1stcase} from above by $\overline{\text{REG}}_{(\bm {\hat{q}}, \bm {\hat{m}})}(\bm{v})$ that is defined as
\begin{equation}\label{eq:regret_1stcase_123}
    \ell(\bm{v}) - \ell(\bm{v})\left(  1 + \log \left({{\dfrac{\ell(\bm{v})}{\overline{\ell}}}}\right) \right)^+ 
    + \int_{\frac{v_{{i^{\operatorname{maj}}(\bm{v})}}}{(1+\gamma)\overline{\ell}}}^{\frac{\ell(\bm{v})}{\overline{\ell}}}{\overline{\ell}}\left(  1 + \log z \right)^+ {\rm{d}} z
    + \int_{\frac{v_{{i^{\operatorname{min}}(\bm{v})}}}{\overline{\ell}}}^{\frac{\ell(\bm{v})}{\overline{\ell}}}{\overline{\ell}}\left(  1 + \log z \right)^+ {\rm{d}} z.
\end{equation}
If $\ell(\bm{v}) < \frac{{\overline{\ell}}}{e}$, then we have $\overline{\text{REG}}_{(\bm {\hat{q}}, \bm {\hat{m}})}(\bm{v}) = \ell(\bm{v}) < \frac{{\overline{\ell}}}{e}\leq {\overline{\ell}}\max\left\{\frac{1}{e}, \theta^\star - \beta^\star \log \beta^\star \right\}$ as the three positive parts in \eqref{eq:regret_1stcase_123} simultaneously evaluate to zero.

For the case $\ell(\bm{v}) \geq \frac{{\overline{\ell}}}{e}$, we have
\begin{equation*}
    \dfrac{\overline{\text{REG}}_{(\bm {\hat{q}}, \bm {\hat{m}})}(\bm{v})}{\overline{\ell}} \\
    =  - {{\dfrac{\ell(\bm{v})}{\overline{\ell}}}}\log \left({{\dfrac{\ell(\bm{v})}{\overline{\ell}}}}\right)  
    + \int_{\frac{v_{{i^{\operatorname{maj}}(\bm{v})}}}{(1+\gamma)\overline{\ell}}}^{\frac{\ell(\bm{v})}{\overline{\ell}}}\left(  1 + \log z \right)^+ {\rm{d}} z
    + \int_{\frac{v_{{i^{\operatorname{min}}(\bm{v})}}}{\overline{\ell}}}^{\frac{\ell(\bm{v})}{\overline{\ell}}}\left(  1 + \log z \right)^+ {\rm{d}} z.
\end{equation*}
We will first show that $\overline{\text{REG}}_{(\bm {\hat{q}}, \bm {\hat{m}})}(\bm{v})$ (weakly) is non-decreasing in~$v_{i^{\operatorname{maj}}(\bm{v})}$ and then use this observation to argue that $\overline{\text{REG}}_{(\bm {\hat{q}}, \bm {\hat{m}})}(\bm{v})$ is bounded above by replacing $\bm{v}$ in the right-hand side of \eqref{eq:regret_1stcase_123} with another value profile $\bm{v}'$, where~$\bm{v}'$ is the value profile obtained by identifying the majority bidder with the largest maximum value and increasing this bidder’s value to its maximum, while keeping all other bidders' values unchanged. We note that the affected bidder need not be $i^{\operatorname{maj}}(\bm{v})$, whose private value (rather than its upper bound) is the highest.
This ensures that $v'_{i^{\operatorname{maj}}(\bm{v}')}$ is equal to $\overline{v}^{\operatorname{maj}}$ and that $v'_{i^{\operatorname{min}}(\bm{v}')}$ remains the same as $v_{i^{\operatorname{min}}(\bm{v})}$. 

As we suppose $\ell(\bm{v}) \geq \frac{{\overline{\ell}}}{e}$, we have
\begin{equation}
\label{eq:regret_1stcase_4th}
\begin{aligned}
    &\dfrac{\overline{\text{REG}}_{(\bm {\hat{q}}, \bm {\hat{m}})}(\bm{v})}{\overline{\ell}} =  - {{\dfrac{\ell(\bm{v})}{\overline{\ell}}}}\log \left({{\dfrac{\ell(\bm{v})}{\overline{\ell}}}}\right)
    + \int_{\frac{v_{{i^{\operatorname{maj}}(\bm{v})}}}{(1+\gamma)\overline{\ell}}}^{\frac{\ell(\bm{v})}{\overline{\ell}}}\left(  1 + \log z \right)^+ {\rm{d}} z
    + \int_{\frac{v_{{i^{\operatorname{min}}(\bm{v})}}}{\overline{\ell}}}^{\frac{\ell(\bm{v})}{\overline{\ell}}}\left(  1 + \log z \right)^+ {\rm{d}} z\\
    &=  - {{\dfrac{\ell(\bm{v})}{\overline{\ell}}}}\log \left({{\dfrac{\ell(\bm{v})}{\overline{\ell}}}} \right)
    + \int_{\max\left\{{\frac{v_{{i^{\operatorname{maj}}(\bm{v})}}}{(1+\gamma)\overline{\ell}}},\frac{1}{e}\right\}}^{\frac{\ell(\bm{v})}{\overline{\ell}}}\left(  1 + \log z \right) {\rm{d}} z
    + \int_{\max\left\{\frac{v_{{i^{\operatorname{min}}(\bm{v})}}}{\overline{\ell}},\frac{1}{e}\right\}}^{\frac{\ell(\bm{v})}{\overline{\ell}}}\left(  1 + \log z \right) {\rm{d}} z \\[2mm]
    &= {{\dfrac{\ell(\bm{v})}{\overline{\ell}}}}\log \left({{\dfrac{\ell(\bm{v})}{\overline{\ell}}}}\right) - \alpha(\bm{v}) \log \alpha (\bm{v}) - \beta (\bm{v}) \log \beta (\bm{v}),  
\end{aligned}
\end{equation}
where mappings $\alpha$ and $\beta$ are defined as
\begin{equation*}
    \alpha(\bm{v}) = \max \left\{ \frac{v_{{i^{\operatorname{min}}(\bm{v})}}}{\overline{\ell}},\frac{1}{e} \right\}
    \quad \text{and} \quad 
    \beta(\bm{v}) = \max \left\{{\frac{v_{{i^{\operatorname{maj}}(\bm{v})}}}{(1+\gamma)\overline{\ell}}},\frac{1}{e}\right\}.
\end{equation*}
Consider the case $\beta(\bm{v}) = \frac{1}{e}$, where the third term $\beta (\bm{v}) \log \beta (\bm{v})$ in the final line of \eqref{eq:regret_1stcase_4th} is a constant. 
The first term ${{\frac{\ell(\bm{v})}{\overline{\ell}}}}\log \left({{\frac{\ell(\bm{v})}{\overline{\ell}}}}\right)$ is non-decreasing in $\ell(\bm{v})$ and particularly in $v_{i^{\operatorname{maj}}(\bm{v})}$. As $\alpha(\bm v)$ and hence $\alpha (\bm{v}) \log \alpha (\bm{v})$ do not depend on $v_{i^{\operatorname{maj}}(\bm{v})}$, $\overline{\text{REG}}_{(\bm {\hat{q}}, \bm {\hat{m}})}(\bm{v})$ is also non-decreasing in $v_{i^{\operatorname{maj}}(\bm{v})}$. 
Otherwise, if $\beta(\bm{v}) = {\frac{v_{{i^{\operatorname{maj}}(\bm{v})}}}{(1+\gamma)\overline{\ell}}} > \frac{1}{e}$, we have
\begin{equation*}
\begin{aligned}
    &{{\dfrac{\ell(\bm{v})}{\overline{\ell}}}}\log \left({{\dfrac{\ell(\bm{v})}{\overline{\ell}}}}\right) - \beta(\bm{v}) \log \beta(\bm{v})\\
    &= \left( \beta(\bm{v}) + \frac{\gamma}{1+\gamma}  \frac{v_{i^{\operatorname{min}}(\bm{v})}}{\overline{\ell}} \right) \log \left( \beta(\bm{v}) + \frac{\gamma}{1+\gamma} \frac{v_{i^{\operatorname{min}}(\bm{v})}}{\overline{\ell}}  \right) - \beta(\bm{v}) \log \beta(\bm{v}),
\end{aligned}
\end{equation*}
where the equality follows from the definition of $\ell(\bm v)$.
The above term is non-decreasing in $\beta(\bm{v})$ (and consequently in $v_{i^{\operatorname{maj}}(\bm{v})}$) because \begin{equation*}
\frac{{\rm d}}{{\rm d}\eta} \left(
 (\eta+\delta)\log(\eta+\delta)-\eta \log\eta\right) = \log\left(1+\frac{\delta}{\eta}\right) > 0 \quad \forall \eta > 0, \delta \geq 0.   
\end{equation*} 
Hence, regardless of the value of $\beta(\bm{v})$, we find that $\overline{\text{REG}}_{(\bm {\hat{q}}, \bm {\hat{m}})}(\bm{v})$ is non-decreasing in $v_{i^{\operatorname{maj}}(\bm{v})}$ for any $v_{i^{\operatorname{maj}}(\bm{v})}\geq 0$. We can thus conclude that 
\begin{equation}
\begin{aligned}   \label{eq:regret_1stcase_uniopt1} 
&\dfrac{\overline{\text{REG}}_{(\bm {\hat{q}}, \bm {\hat{m}})}(\bm{v})}{\overline{\ell}}\\ 
&= \left( \frac{v_{i^{\operatorname{maj}}(\bm{v})}+\gamma v_{i^{\operatorname{min}}(\bm{v})}}{(1+\gamma){\overline{\ell}}} \right) \log \left( \frac{v_{i^{\operatorname{maj}}(\bm{v})}+\gamma v_{i^{\operatorname{min}}(\bm{v})}}{(1+\gamma){\overline{\ell}}} \right) - \alpha(\bm{v}) \log \alpha(\bm{v}) - \beta(\bm{v}) \log \beta(\bm{v}) \\
&\leq \left( \frac{{\overline{v}^{\operatorname{maj}}}+\gamma v_{i^{\operatorname{min}}(\bm{v})}}{(1+\gamma){\overline{\ell}}} \right) \log \left( \frac{{\overline{v}^{\operatorname{maj}}}+\gamma v_{i^{\operatorname{min}}(\bm{v})}}{(1+\gamma){\overline{\ell}}} \right) - \alpha(\bm{v}) \log \alpha(\bm{v}) - \beta^\star \log \beta^\star \\
&\leq \max_{\substack{u \in [0,{\overline{v}^{\operatorname{min}}}] \\ \frac{{\overline{v}^{\operatorname{maj}}}+\gamma u}{1+\gamma} \geq \frac{{\overline{\ell}}}{e}}}
\left\{ \left( \frac{{\overline{v}^{\operatorname{maj}}}+\gamma u}{(1+\gamma){\overline{\ell}}} \right) \log \left( \frac{{\overline{v}^{\operatorname{maj}}}+\gamma u}{(1+\gamma){\overline{\ell}}} \right) - \max \left\{ \dfrac{u}{\overline{\ell}}, \frac{1}{e} \right\} \log \max \left\{ \dfrac{u}{\overline{\ell}}, \frac{1}{e} \right\}  \right\} - \beta^\star \log \beta^\star 
\end{aligned}
\end{equation}
where $\beta^\star$ is a constant defined as $\max \left\{ \frac{\overline{v}^{\operatorname{maj}}}{(1+\gamma)\overline{\ell}},\frac{1}{e} \right\}$. 
The first equality follows from \eqref{eq:regret_1stcase_4th}, and
the first inequality in \eqref{eq:regret_1stcase_uniopt1} is obtained by replacing $\bm{v}$ with $\bm{v}'$. Note that with this replacement $v'_{i^{\operatorname{min}}(\bm{v}')}$ (consequently, $\alpha(\bm v')$) remains the same as $v_{i^{\operatorname{min}}(\bm{v})}$ (consequently, $\alpha(\bm v)$). 
The second inequality in \eqref{eq:regret_1stcase_uniopt1} holds by treating $u = v_{i^{\operatorname{min}}(\bm{v})}\in [0,{\overline{v}^{\operatorname{min}}}]$ as a free variable to be chosen independently. 
Note that we subject $u$ to another constraint $\frac{{\overline{v}^{\operatorname{maj}}}+\gamma u}{1+\gamma} \geq \frac{{\overline{\ell}}}{e}$ 
which follows from the assumption that $\ell (\bm{v})\geq \frac{{\overline{\ell}}}{e}$
and from the fact that $\frac{{\overline{v}^{\operatorname{maj}}}+\gamma u}{1+\gamma} = \ell (\bm{v}') \geq \ell (\bm{v})$. 
This constraint $\frac{{\overline{v}^{\operatorname{maj}}}+\gamma u}{1+\gamma} \geq \frac{{\overline{\ell}}}{e}$ essentially imposes a constant lower bound on $u$, namely $ u\geq \frac{{\overline{\ell}(1+\gamma)}}{e\gamma}-\frac{{\overline{v}^{\operatorname{maj}}}}{\gamma}$. 
The maximum in the last line of~\eqref{eq:regret_1stcase_uniopt1} is, in fact, $\theta^\star$. 
It then follows that ${\text{REG}}_{(\bm {\hat{q}}, \bm {\hat{m}})}(\bm{v}) \leq
\overline{\text{REG}}_{(\bm {\hat{q}}, \bm {\hat{m}})}(\bm{v})\leq {\overline{\ell}}( \theta^\star - \beta^\star \log \beta^\star )
\leq {\overline{\ell}}\max\left\{\frac{1}{e}, \theta^\star - \beta^\star \log \beta^\star \right\}$.

\textit{Case 2} ($v_{{i^{\operatorname{maj}}(\bm{v})}} \leq v_{{i^{\operatorname{min}}(\bm{v})}}$): In this case, as the benchmark (i.e., the maximum revenue in hindsight) \eqref{eq: hindsight revenue} simplifies to $v_{{i^{\operatorname{min}}(\bm{v})}}$, we can write $\text{REG}_{(\bm {\hat{q}}, \bm {\hat{m}})}(\bm{v})$ as
\begin{equation*}
    \text{REG}_{(\bm {\hat{q}}, \bm {\hat{m}})}(\bm{v}) = 
    v_{{i^{\operatorname{min}}(\bm{v})}} - \sum_{i\in\mathcal{I}^{\operatorname{min}}} \hat{m}_{i}(\bm{v}) - \sum_{i\in\mathcal{I}^{\operatorname{maj}}} \hat{m}_{i}(\bm{v}) = 
    v_{{i^{\operatorname{min}}(\bm{v})}} - \hat{m}_{i^{\operatorname{min}}(\bm v)}(\bm{v}),
\end{equation*}
where the second equality follows from the definition of $\bm {\hat{m}}$ under which only the highest bidder ${i^{\operatorname{min}}(\bm{v})}$ can be charged a non-zero payment when $v_{{i^{\operatorname{maj}}(\bm{v})}} \leq v_{{i^{\operatorname{min}}(\bm{v})}}$. As $\hat{m}_{i^{\operatorname{min}}(\bm v)}(\bm{v})$ is non-negative, if $v_{{i^{\operatorname{min}}(\bm{v})}} \leq \frac{{\overline{\ell}}}{e}$, it immediately holds that $\text{REG}_{(\bm {\hat{q}}, \bm {\hat{m}})}(\bm{v}) \leq \frac{{\overline{\ell}}}{e}$. Assume now that $v_{{i^{\operatorname{min}}(\bm{v})}} > \frac{{\overline{\ell}}}{e}$.
By the definition of $\bm {\hat{m}}$, the regret above can be equivalently expressed as 
\begin{equation*}
\begin{aligned}
    \text{REG}_{(\bm {\hat{q}}, \bm {\hat{m}})}(\bm{v}) 
    &= v_{{i^{\operatorname{min}}(\bm{v})}}  - v_{{i^{\operatorname{min}}(\bm{v})}}{\hat{q}}_{{i^{\operatorname{min}}(\bm{v})}}(\bm{v})+\int_{0}^{v_{{i^{\operatorname{min}}(\bm{v})}}}{\hat{q}}_{{i^{\operatorname{min}}(\bm{v})}}(x,\bm{v}_{-{i^{\operatorname{min}}(\bm{v})}}) {\rm{d}}  x \\
    &= - v_{{i^{\operatorname{min}}(\bm{v})}} \log \left({\frac{v_{{i^{\operatorname{min}}(\bm{v})}}}{\overline{\ell}}}\right) +\int_{0}^{v_{{i^{\operatorname{min}}(\bm{v})}}}{\hat{q}}_{{i^{\operatorname{min}}(\bm{v})}}(x,\bm{v}_{-{i^{\operatorname{min}}(\bm{v})}}) {\rm{d}}  x,
\end{aligned}
\end{equation*}
where the second equality follows from the definition of ${\hat{q}}_{{i^{\operatorname{min}}(\bm{v})}}$. Analogously to {\it Case 1}, we bound the integral above as follows:
\begin{equation*}
\allowdisplaybreaks
\begin{aligned}
    &\int_{0}^{v_{{i^{\operatorname{min}}(\bm{v})}}}{\hat{q}}_{{i^{\operatorname{min}}(\bm{v})}}(x,\bm{v}_{-{i^{\operatorname{min}}(\bm{v})}}) {\rm{d}} x \\
    &= \int_{0}^{v_{{i^{\operatorname{maj}}(\bm{v})}}}\dfrac{\gamma}{1+\gamma}\left(  1 + \log \left(\dfrac{\frac{\gamma}{1+\gamma}x+\frac{1}{1+\gamma}v_{{i^{\operatorname{maj}}(\bm{v})}}}{\overline{\ell}}\right) \right)^+\mathds{I}_{\{x\geq \max_{j \in \mathcal{I}^{\operatorname{min}} \setminus \{ i^{\operatorname{min}}(\bm{v}) \} } v_j \}} {\rm{d}} x \\ 
    &\quad + \int_{v_{{i^{\operatorname{maj}}(\bm{v})}}}^{v_{{i^{\operatorname{min}}(\bm{v})}}}\left(  1 + \log  \left(\dfrac{x}{\overline{\ell}}\right) \right)^+\mathds{I}_{\{x\geq \max_{j \in \mathcal{I}^{\operatorname{min}} \setminus \{ i^{\operatorname{min}}(\bm{v}) \} } v_j \}} {\rm{d}} x \\
    &\leq \int_{0}^{v_{{i^{\operatorname{maj}}(\bm{v})}}}\dfrac{\gamma}{1+\gamma}\left(  1 + \log \left(\dfrac{\frac{\gamma}{1+\gamma}x+\frac{1}{1+\gamma}v_{{i^{\operatorname{maj}}(\bm{v})}}}{\overline{\ell}}\right) \right)^+ {\rm{d}} x + \int_{v_{{i^{\operatorname{maj}}(\bm{v})}}}^{v_{{i^{\operatorname{min}}(\bm{v})}}}\left(  1 + \log  \left(\dfrac{x}{\overline{\ell}}\right) \right)^+ {\rm{d}} x \\
    &= \int_{\frac{1}{1+\gamma}v_{{i^{\operatorname{maj}}(\bm{v})}}}^{v_{{i^{\operatorname{maj}}(\bm{v})}}}\left(  1 + \log \left(\dfrac{y}{\overline{\ell}}\right) \right)^+ {\rm{d}} y + \int_{v_{{i^{\operatorname{maj}}(\bm{v})}}}^{v_{{i^{\operatorname{min}}(\bm{v})}}}\left(  1 + \log  \left(\dfrac{x}{\overline{\ell}}\right) \right)^+ {\rm{d}} x \\
    &= \int_{\frac{1}{1+\gamma}v_{{i^{\operatorname{maj}}(\bm{v})}}}^{v_{{i^{\operatorname{min}}(\bm{v})}}}\left(  1 + \log \left(\dfrac{y}{\overline{\ell}}\right) \right)^+ {\rm{d}} y
    = \int_{\frac{v_{{i^{\operatorname{maj}}(\bm{v})}}}{(1+\gamma)\overline{\ell}}}^{\frac{v_{{i^{\operatorname{min}}(\bm{v})}}}{\overline{\ell}}}{\overline{\ell}}\left(  1 + \log  z \right)^+ {\rm{d}} z 
    \end{aligned}    
\end{equation*}
where the first equality follows from the definition of $\hat{q}_{i^{\operatorname{min}}(\bm{v})}$, and the second and the fourth equalities follow from affine transformations $y = \frac{\gamma}{1+\gamma}x + \frac{1}{1+\gamma}v_{{i^{\operatorname{maj}}(\bm{v})}}$ and $z = \frac{y}{\overline{\ell}}$, respectively.
Note that
\begin{equation*}
\allowdisplaybreaks
\begin{aligned}
    &\int_{\frac{v_{{i^{\operatorname{maj}}(\bm{v})}}}{(1+\gamma)\overline{\ell}}}^{\frac{v_{{i^{\operatorname{min}}(\bm{v})}}}{\overline{\ell}}}{\overline{\ell}}\left(  1 + \log  z \right)^+ {\rm{d}} z 
    \leq \int_{\frac{1}{e}}^{\frac{v_{{i^{\operatorname{min}}(\bm{v})}}}{\overline{\ell}}}{\overline{\ell}}\left(  1 + \log  z \right) {\rm{d}} z\\
    &= {\overline{\ell}}\left(\dfrac{v_{{i^{\operatorname{min}}(\bm{v})}}}{\overline{\ell}} \log \left(\dfrac{v_{{i^{\operatorname{min}}(\bm{v})}}}{\overline{\ell}}\right)+ \frac{1}{e}\right)
    = v_{{i^{\operatorname{min}}(\bm{v})}} \log \left({\frac{v_{{i^{\operatorname{min}}(\bm{v})}}}{\overline{\ell}}}\right)+\dfrac{\overline{\ell}}{e},
\end{aligned}    
\end{equation*}
It therefore holds that $\text{REG}_{(\bm {\hat{q}}, \bm {\hat{m}})}(\bm{v}) \leq \frac{{\overline{\ell}}}{e}$.

In conclusion, our analysis of Case 2 reveals that 
${\text{REG}}_{(\bm {\hat{q}}, \bm {\hat{m}})}(\bm{v}) \leq \frac{{\overline{\ell}}}{e} \leq {\overline{\ell}}\max\left\{\frac{1}{e}, \theta^\star - \beta^\star \log \beta^\star \right\}$ for any $\bm{v} \in \mathcal{V}$ such that $ v_{{i^{\operatorname{maj}}(\bm{v})}} \leq v_{{i^{\operatorname{min}}(\bm{v})}}$, and the claim follows.
\end{proof}

\begin{proof}{Proof of Proposition \ref{prop:extreme_cases}}

To streamline the exposition, we introduce the shorthand notation
\begin{equation*}
    a(u,\gamma) = \frac{\overline{v}^{\operatorname{maj}}+\gamma u}{(1+\gamma)\overline{\ell}},
    \qquad
    b(u,\gamma) = \max\left\{ \frac{u}{\overline{\ell}}, \frac{1}{e} \right\},
    \qquad
    g(x) = x\log x.
    % \begin{cases} 
    %     x\log x  & \text{ if } x > 0, \\
    %     \lim_{t \downarrow 0} g(t) = 0 & \text{ if } x = 0.
    % \end{cases}
\end{equation*}
so that $\theta^\star = \max_{u\in\mathcal U}  F(u,\gamma) $ with $F(u,\gamma) = g(a(u,\gamma)) - g(b(u,\gamma))$.\footnote{In other parts of the paper, we work with a fixed $\gamma$ and thus suppress its dependence in the notation. Here, $\gamma$ is allowed to vary, but for notational consistency we continue to write $\overline{\ell}$, $\theta^\star$, $\beta^\star$, and $\lambda$ without explicitly indicating their dependence on $\gamma$, even though they are functions of $\gamma$.} 
%For all $u \in \mathcal{U}$, it is noted that $u \leq \overline{v}^{\operatorname{min}}$ and hence
By the definition of $\mathcal{U}$, we have $\frac{1}{e} \leq a(u,\gamma),b(u,\gamma) \leq 1$.
% \begin{equation*} 
%     \frac{1}{e} \leq a(u,\gamma),b(u,\gamma) \leq 1. %\quad %\text{and} \quad 
%     %\frac{1}{e} \leq b(u,\gamma) \leq 1.
% \end{equation*}
The remainder of this proof is divided into two cases, one for each asymptotic regime.

    \textit{Case 1 ($\gamma \downarrow 0$)}: 
    In this asymptotic regime, if $\overline{v}^{\operatorname{min}} \geq  e\overline{v}^{\operatorname{maj}}$, then we have $\overline{\ell} = \overline{v}^{\operatorname{min}}$ and $\overline{v}^{\operatorname{maj}}\leq \frac{\overline{v}^{\operatorname{min}}}{e} = \frac{\overline{\ell}}{e}$, yielding $\beta^\star = \frac{1}{e}$.
    We then consider $\theta^\star$. Since the feasible set $\mathcal{U}$ encodes an explicit lower bound on $u$, we have 
\begin{equation*}\label{eqn: u in prop 2}
\begin{aligned}
    u\geq& \left(\frac{{\overline{\ell}(1+\gamma)}}{e\gamma}-\frac{{\overline{v}^{\operatorname{maj}}}}{\gamma}\right)^+
    \geq \frac{{\overline{\ell}(1+\gamma)}}{e\gamma}-\frac{{\overline{v}^{\operatorname{maj}}}}{\gamma}
    = \frac{1}{\gamma}\left( \frac{\overline{\ell}}{e} - \overline{v}^{\operatorname{maj}} \right) + \frac{\overline{\ell}}{e}\geq\frac{\overline{\ell}}{e}.
\end{aligned}
\end{equation*}     
% From \eqref{eqn: u in prop 2}, we can directly obtain two inequalities: $\frac{{\overline{v}^{\operatorname{maj}}}+\gamma u}{(1+\gamma){\overline{\ell}}}\geq \frac{1}{e}$ and $u\geq \frac{\overline{\ell}}{e}$, respectively. That is, $a(u,\gamma)\geq \frac{1}{e}$ and 
It follows that $b(u,\gamma) = \frac{u}{\overline{\ell}}$. 
%Combining $u\geq \frac{\overline{\ell}}{e}$ and $\overline{v}^{\operatorname{maj}}\leq \frac{\overline{\ell}}{e}$ yields $u\geq \overline{v}^{\operatorname{maj}}$, which implies that 
Furthermore, since $u \geq \frac{\overline{\ell}}{e} \geq \overline{v}^{\operatorname{maj}}$, it also holds that
\begin{equation*}
a(u,\gamma) = \frac{{\overline{v}^{\operatorname{maj}}}+\gamma u}{(1+\gamma){\overline{\ell}}} \leq \frac{u}{\overline{\ell}} = b(u,\gamma).
\end{equation*}
%which constitutes a linear combination of $\frac{u}{\overline{\ell}}$ and $\frac{\overline{v}^{\operatorname{maj}}}{\overline{\ell}}$, is no greater than $\frac{u}{\overline{\ell}}$, \emph{i.e.}, $a(u,\gamma)\leq b(u,\gamma)$.
Since $g(x)$ is a non-decreasing function when $x\geq \frac{1}{e}$, we have $g(a(u,\gamma)) - g(b(u,\gamma)) \leq 0$, yielding $\theta^\star\leq0$.
Thus, $e\Bigl(\theta^\star - \beta^\star \log \beta^\star\Bigr) \leq -e \beta^\star \log \beta^\star \leq 1$ and 
\begin{equation}\label{eqn:lambda is 1-1}
\begin{aligned}
    \lim_{\gamma\downarrow0}\lambda =~\ &  \lim_{\gamma\downarrow0} \left({\overline{\ell}}  /  \max\Bigl\{\overline{v}^{{\operatorname{min}}},\frac{\overline{v}^{{\operatorname{maj}}}}{1+\gamma}\Bigr\}\right) = 1.
\end{aligned}
\end{equation}
Next, if $\overline{v}^{\operatorname{min}} < e\overline{v}^{\operatorname{maj}}$, then
for any sufficiently small $\gamma\in(0,\frac{e\overline{v}^{\operatorname{maj}}-\max\{\overline{v}^{\operatorname{min}}, \overline{v}^{\operatorname{maj}}\}}{\overline{v}^{\operatorname{min}}})$,
we have
\begin{equation}\label{eqn:prop 2 a>=beta}
    \frac{1}{e} < \frac{\overline{v}^{\operatorname{maj}}}{(1+\gamma)\overline{\ell}} = \beta^\star = a(u,\gamma) - \frac{\gamma u}{(1+\gamma)\overline{\ell}} \leq a(u,\gamma)\leq \frac{\overline{v}^{\operatorname{maj}}+\gamma \overline{v}^{\operatorname{min}}}{(1+\gamma)\overline{\ell}}\leq 1 \quad \forall u \in \mathcal{U}.
\end{equation}
The first inequality is because $\gamma <\frac{e\overline{v}^{\operatorname{maj}}-\max\{\overline{v}^{\operatorname{min}}, \overline{v}^{\operatorname{maj}}\}}{\overline{v}^{\operatorname{min}}}$. The first and second equalities are by the definition of $\beta^\star$ and $a(u,\gamma)$. The third inequality is because $\overline{v}^{\operatorname{min}}$ is the right endpoint of the feasible set $\mathcal U$.
We then find
\begin{equation*}
\begin{aligned} 
    \theta^\star - \beta^\star \log \beta^\star
    &=\max_{u\in \mathcal U} \left\{
        g(a(u,\gamma)) - g(b(u,\gamma)) \right\}- \beta^\star \log \beta^\star\\
    &=\max_{u\in \mathcal U} \left\{
        g(a(u,\gamma)) - g(\beta^\star)  - g(b(u,\gamma)) \right\}\\
    &\leq  \max_{u\in \mathcal U} \left\{
        g(a(u,\gamma)) - g(\beta^\star)  \right\} +\max_{u\in \mathcal U} \left\{ -g(b(u,\gamma)) \right\}\\
    &\leq \max_{u\in \mathcal U} \left\{
        \left(a(u,\gamma)) - \beta^\star\right)  g'(a(u,\gamma))  \right\} +\max_{u\in \mathcal U} \left\{ -g(b(u,\gamma)) \right\} \\
    &= \max_{u\in \mathcal U} \left\{
        \frac{\gamma u g'(a(u,\gamma))}{(1+\gamma)\overline{\ell}}     \right\} +\max_{u\in \mathcal U} \left\{ -g(b(u,\gamma)) \right\} \\
    &\leq \max_{u\in \mathcal U} \left\{
        \frac{\gamma u }{(1+\gamma)\overline{\ell}}     \right\} +\frac{1}{e}  \\
    & = \frac{\gamma \overline{v}^{\operatorname{min}}}{(1+\gamma)\overline{\ell}}  + \frac{1}{e} \leq \frac{\gamma }{(1+\gamma)} +\frac{1}{e},
\end{aligned}
\end{equation*} 
where the second inequality is due to convexity of $g$. The third equality holds by \eqref{eqn:prop 2 a>=beta}.
The third inequality holds because $g'(a(u,\gamma)) = 1 + \log a(u,\gamma) \leq 1$ when $a(u,\gamma) \leq 1$ as well as $-\frac{1}{e}$ is a global minimum of $g(x)$ over $x > 0$. Hence, $\lim_{\gamma\downarrow 0} \left(\theta^\star - \beta^\star \log \beta^\star\right) \leq \frac{1}{e}$, and \eqref{eqn:lambda is 1-1} holds again.

\textit{Case 2 ($\gamma \uparrow \infty$)}: 
If $\overline{v}^{\operatorname{min}} \geq  e\overline{v}^{\operatorname{maj}}$, then similar to the first part of the first case, we have $e\Bigl(\theta^\star - \beta^\star \log \beta^\star\Bigr)\leq 1$ and
\begin{equation}\label{eqn:lambda is 1-2}
\begin{aligned}
    \lim_{\gamma\uparrow\infty}\lambda =~\ &  \lim_{\gamma\uparrow\infty} \left({\overline{\ell}}  /  \max\Bigl\{\overline{v}^{{\operatorname{min}}},\frac{\overline{v}^{{\operatorname{maj}}}}{1+\gamma}\Bigr\}\right) = 1.
\end{aligned}
\end{equation} 
Assume now that $\overline{v}^{\operatorname{min}} < e\overline{v}^{\operatorname{maj}}$. Our roadmap is to show that $\lim_{\gamma\uparrow\infty}\theta^\star = 0$. Together with $\lim_{\gamma \uparrow \infty} \beta^\star \log \beta^\star = -1/e$, this would imply
\begin{equation*}
    \lim_{\gamma \uparrow \infty} e\left(\theta^\star - \beta^\star \log \beta^\star\right)
    \leq
    -e \lim_{\gamma \uparrow \infty} \beta^\star \log \beta^\star
    = 1,
\end{equation*}
and hence \eqref{eqn:lambda is 1-2} holds again. To this end, we first simplify the optimization problem defining $\theta^\star$ in \eqref{eq: definition of theta star}. In particular, we show that its feasible set $\mathcal{U}$ can be restricted, without loss of optimality, to
$\mathcal{U} \cap [\overline{v}^{\operatorname{min}}/e,\overline{v}^{\operatorname{min}}]$.

Indeed, any $u < \overline{v}^{\operatorname{min}}/e$ is strictly suboptimal in \eqref{eq: definition of theta star}. For such a $u$, the objective value is
\begin{equation*}
\begin{aligned}
    F(u,\gamma)
    = g(a(u,\gamma)) - g(b(u,\gamma))
    = g(a(u,\gamma)) - g(1/e),
\end{aligned}
\end{equation*}
where the second equality follows from
\[
    \frac{u}{\overline{\ell}}
    \leq
    \frac{u}{\overline{v}^{\operatorname{min}}}
    <
    \frac{1}{e}.
\]
Since $a(u,\gamma)$ is non-decreasing in $u$ and 
%since $u\geq \left(\frac{\overline{\ell}(1+\gamma)}{e\gamma}-\frac{\overline{v}^{\operatorname{maj}}}{\gamma}\right)^+\geq \frac{\overline{\ell}(1+\gamma)}{e\gamma}-\frac{\overline{v}^{\operatorname{maj}}}{\gamma}$, we have $a(u,\gamma)\geq \frac{1}{e}$.
% Since 
$g(x)$ is increasing on $x\in[1/e,\infty)$, 
%$a(u,\gamma)$ is increasing in $u$, and $\color{purple}a(u,\gamma) \geq 1/e~[why?]$ under the condition
% $\overline{v}^{\operatorname{min}} < e\overline{v}^{\operatorname{maj}}$,
it follows that $F(u,\gamma)$ is non-decreasing in $u$ whenever
$u < \overline{v}^{\operatorname{min}}/e$. Hence, any such $u$ generates an objective value that is smaller than or equal to that of $u=\overline{v}^{\operatorname{min}}/e$, and may be excluded from the feasible set without loss of generality.

Focusing on $u \in \mathcal{U} \cap [\overline{v}^{\operatorname{min}}/e,\overline{v}^{\operatorname{min}}]$, 
% we obtain
% \begin{equation*}\color{red}
% \begin{aligned}
% \lim_{\gamma \uparrow \infty} F(u,\gamma)
% = g \left( \lim_{\gamma \uparrow \infty} a(u,\gamma) \right)
% - g \left( \lim_{\gamma \uparrow \infty} b(u,\gamma) \right) 
% = g\left(\frac{u}{\overline{v}^{\operatorname{min}}}\right)
% - g\left(\max\left\{\frac{u}{\overline{v}^{\operatorname{min}}},\frac{1}{e}\right\}\right) 
% = 0,
% \end{aligned}
% \end{equation*}
% where the last equality follows from $u/\overline{v}^{\operatorname{min}} \geq 1/e$. In fact, a stronger uniform convergence~statement holds: 
we now show that $F(u,\gamma)$ uniformly converges to $0$ when $\gamma\uparrow \infty$, in the sense that
there exists a $u$-independent function $\varepsilon(\gamma)$ with $\varepsilon(\gamma) \downarrow 0$ as $\gamma \uparrow \infty$ such that, 
\begin{equation*}
\left|F(u,\gamma)\right|
\leq \varepsilon(\gamma)
\quad
\forall u \in \mathcal{U} \cap[\overline{v}^{\operatorname{min}}/e,\overline{v}^{\operatorname{min}}],\gamma > 0.
\end{equation*}
%To establish this uniform convergence, note first that under $\overline{v}^{\operatorname{min}} < e\overline{v}^{\operatorname{maj}}$, both $a(u,\gamma)$ and $b(u,\gamma)$ belong to $[1/e,1]$ for every $u\in \overline{v}^{\operatorname{min}}[1/e,1]$. 
To establish this uniform convergence, note that $g(x)=x\log x$ is
$1$-Lipschitz on $[1/e,1]$, which contains all admissible values of $a(u,\gamma)$ and $b(u,\gamma)$. Therefore, we have
\begin{equation*}
\begin{aligned}
    \vert F(u,\gamma) \vert
    &= \vert g(a(u,\gamma)) - g(b(u,\gamma)) \vert \\
    &\leq  \vert a(u,\gamma) - b(u,\gamma) \vert \\
    &\leq  \left\vert a(u,\gamma) - \frac{u}{\overline{v}^{\operatorname{min}}} \right\vert
    + \left\vert b(u,\gamma) - \frac{u}{\overline{v}^{\operatorname{min}}} \right\vert \\
    &\leq  \left\vert a(u,\gamma) - \frac{u}{\overline{v}^{\operatorname{min}}} \right\vert
    +  \left\vert \max\left\{ \frac{u}{\overline{\ell}}, \frac{1}{e} \right\}
    -  \max\left\{\frac{u}{\overline{v}^{\operatorname{min}}}, \frac{1}{e} \right\} \right\vert \\
    &\leq  \left\vert a(u,\gamma) - \frac{u}{\overline{v}^{\operatorname{min}}} \right\vert
    +  \left\vert  \frac{u}{\overline{\ell}} - \frac{u}{\overline{v}^{\operatorname{min}}} \right\vert \\
    &=
    \frac{
    \left\vert
    \overline{v}^{\operatorname{min}}\overline{v}^{\operatorname{maj}}
    -
    u\max \{ \overline{v}^{\operatorname{min}},\overline{v}^{\operatorname{maj}} \}
    \right\vert
    }{
    \overline{v}^{\operatorname{min}}
    \left(
    \max \{ \overline{v}^{\operatorname{min}},\overline{v}^{\operatorname{maj}} \}
    +
    \gamma \overline{v}^{\operatorname{min}}
    \right)
    }
    +
    \frac{
    u \left(
    \max \{ \overline{v}^{\operatorname{min}},\overline{v}^{\operatorname{maj}} \}
    -
    \overline{v}^{\operatorname{min}}
    \right)
    }{
    \overline{v}^{\operatorname{min}}
    \left(
    \max \{ \overline{v}^{\operatorname{min}},\overline{v}^{\operatorname{maj}} \}
    +
    \gamma \overline{v}^{\operatorname{min}}
    \right)
    }\\
    & \leq   \frac{
    \overline{v}^{\operatorname{maj}}
    +
    \max \{ \overline{v}^{\operatorname{min}},\overline{v}^{\operatorname{maj}} \}
    }{
    \max \{ \overline{v}^{\operatorname{min}},\overline{v}^{\operatorname{maj}} \}
    +
    \gamma \overline{v}^{\operatorname{min}}   
    }
    +
    \frac{
    \max \{ \overline{v}^{\operatorname{min}},\overline{v}^{\operatorname{maj}} \}
    -
    \overline{v}^{\operatorname{min}}
    }{
    \max \{ \overline{v}^{\operatorname{min}},\overline{v}^{\operatorname{maj}} \}
    +
    \gamma \overline{v}^{\operatorname{min}}
    }. 
\end{aligned}
\end{equation*}
The fourth inequality holds because the mapping $x \mapsto \max\{x,\frac{1}{e}\}$ is 1-Lipshitz. 
The right-hand side is an $u$-independent upper bound $\varepsilon(\gamma)$ that converges to zero
as $\gamma\uparrow\infty$.

Therefore,
\begin{equation*}
    \lim_{\gamma \uparrow \infty} \theta^\star \leq \lim_{\gamma \uparrow \infty} \max_{ u \in  \overline{v}^{\operatorname{min}}[1/e,1] }  F(u,\gamma)  = 0,
\end{equation*}
which in turn implies that $\lim_{\gamma \uparrow \infty} \lambda = 1$.

Combining the results from the two cases, we have shown that in both asymptotic regimes, the constant approximation ratio converges to one. The proof is therefore complete. 
\end{proof}

\begin{proof}{Proof of Proposition \ref{prop:fairness of payments}}

If the highest bidder belongs to the minority group, then by the definition of $(\bm q^{\operatorname{sa}}, \bm m^{\operatorname{sa}})$, the allocations and payments for majority bidders are zero, and the inequality holds immediately.
        
If the highest bidder is from the majority group, then by the definition of $(\bm q^{\operatorname{sa}}, \bm m^{\operatorname{sa}})$, there are two bidders whose allocations (and payments) might be non-zero: the highest minority bidder with $(q^{\operatorname{sa}}_{i^{\operatorname{min}}(\bm{v})}(\bm{v}),m^{\operatorname{sa}}_{i^{\operatorname{min}}(\bm{v})}(\bm{v}))$ 
and the highest majority bidder with $(q^{\operatorname{sa}}_{i^{\operatorname{maj}}(\bm{v})}(\bm{v}),m^{\operatorname{sa}}_{i^{\operatorname{maj}}(\bm{v})}(\bm{v}))$. 
By \eqref{eq:IR}, the payment of the highest minority bidder satisfies 
\begin{equation}\label{eqn:fairness of payments 1}
     v_{i^{\operatorname{min}}(\bm{v})}q^{\operatorname{sa}}_{i^{\operatorname{min}}(\bm{v})}(\bm{v})\geq m^{\operatorname{sa}}_{i^{\operatorname{min}}(\bm{v})}(\bm{v}).
\end{equation}
If $v_{i^{\operatorname{min}}(\bm{v})} = 0$, then $m^{\operatorname{sa}}_{i^{\operatorname{min}}(\bm{v})}(\bm{v}) =0$ and the inequality in this proposition holds immediately. 
Therefore, it suffices to consider the case $v_{i^{\operatorname{min}}(\bm{v})} > 0$.
We then analyze the payment of the highest majority bidder. 

\begin{equation}\label{eqn:fairness of payments 2}
\begin{aligned}
    &m^{\operatorname{sa}}_{i^{\operatorname{maj}}(\bm{v})}(\bm{v})\\
    &= v_{i^{\operatorname{maj}}(\bm{v})}q^{\operatorname{sa}}_{i^{\operatorname{maj}}(\bm{v})}(\bm{v})-\int_{0}^{v_{i^{\operatorname{maj}}(\bm{v})}}q^{\operatorname{sa}}_{i^{\operatorname{maj}}(\bm{v})}(x,\bm{v}_{-{i^{\operatorname{maj}}(\bm{v})}}) {\rm{d}} x\\
    &= v_{i^{\operatorname{maj}}(\bm{v})}q^{\operatorname{sa}}_{i^{\operatorname{maj}}(\bm{v})}(\bm{v})-\int_{v_{i^{\operatorname{min}}(\bm{v})}}^{v_{i^{\operatorname{maj}}(\bm{v})}}q^{\operatorname{sa}}_{i^{\operatorname{maj}}(\bm{v})}(x,\bm{v}_{-{i^{\operatorname{maj}}(\bm{v})}}) {\rm{d}} x\\
    &= v_{i^{\operatorname{min}}(\bm{v})}q^{\operatorname{sa}}_{i^{\operatorname{maj}}(\bm{v})}(\bm{v}) + (v_{i^{\operatorname{maj}}(\bm{v})} - v_{i^{\operatorname{min}}(\bm{v})}) q^{\operatorname{sa}}_{i^{\operatorname{maj}}(\bm{v})}(\bm{v})-\int_{v_{i^{\operatorname{min}}(\bm{v})}}^{v_{i^{\operatorname{maj}}(\bm{v})}}q^{\operatorname{sa}}_{i^{\operatorname{maj}}(\bm{v})}(x,\bm{v}_{-{i^{\operatorname{maj}}(\bm{v})}}) {\rm{d}} x\\
     &= v_{i^{\operatorname{min}}(\bm{v})}q^{\operatorname{sa}}_{i^{\operatorname{maj}}(\bm{v})}(\bm{v}) + \int_{v_{i^{\operatorname{min}}(\bm{v})}}^{v_{i^{\operatorname{maj}}(\bm{v})}}(q^{\operatorname{sa}}_{i^{\operatorname{maj}}(\bm{v})}(\bm{v})-q^{\operatorname{sa}}_{i^{\operatorname{maj}}(\bm{v})}(x,\bm{v}_{-{i^{\operatorname{maj}}(\bm{v})}})) {\rm{d}} x\\
    &\geq v_{i^{\operatorname{min}}(\bm{v})}q^{\operatorname{sa}}_{i^{\operatorname{maj}}(\bm{v})}(\bm{v}). %= v_{i^{\operatorname{min}}(\bm{v})} \sum_{i\in \majority}q^{\operatorname{sa}}_i(\bm{v})
\end{aligned}
\end{equation}
The first equality above is obtained by the definition of $\bm {m}^{\operatorname{sa}}$. The second equality holds by the definition of $\bm {q}^{\operatorname{sa}}$ and that
the value of $q^{\operatorname{sa}}_{i^{\operatorname{maj}}(\bm{v})}(x,\bm{v}_{-{i^{\operatorname{maj}}(\bm{v})}})$ is zero when $x<v_{i^{\operatorname{min}}(\bm{v})}$. 
The inequality holds because $q^{\operatorname{sa}}_{i^{\operatorname{maj}}(\bm{v})}(x,\bm{v}_{-{i^{\operatorname{maj}}(\bm{v})}})$ is an non-decreasing function of $x$.
Since $v_{i^{\operatorname{min}}(\bm{v})} > 0$ and $\bm {q}^{\operatorname{sa}}$ is non-negative, combining \eqref{eqn:fairness of payments 1} and \eqref{eqn:fairness of payments 2} yields the desired inequality in the proposition.
\end{proof}

\begin{proof}{Proof of Proposition \ref{prop:feasibility of stochastic case}}
We will first prove that the constraint \eqref{eq:IC} holds under the mechanism $(\bm q^\star, \bm m^\star)$ for both minority bidders {(\it Case 1)} and majority bidders {(\it Case 2)}.  
Then, we will prove constraints ~\eqref{eq:IR},~\eqref{eq:AF} and ~\eqref{eq:Eq} and conclude the proof.

    To prove that $(\bm q^\star, \bm m^\star)$ satisfies the \eqref{eq:IC} constraint, it is sufficient to show that conditions (i) and (ii) in Lemma \ref{lem: monotonicity of allocation} hold. By the construction of $\bm m^\star$, condition (ii) in Lemma~\ref{lem: monotonicity of allocation} immediately holds. We now consider (i) and show that the allocation rule $\bm q^\star$ is non-decreasing, that is, for each $i \in \mathcal I$, it holds that $q^\star_i(v_i,\bm{v}_{-i}) \geq q^\star_i(w_i,\bm{v}_{-i})$ for all $\bm{v} \in \mathcal{V}$ and $w_i \in \mathcal{V}_i: v_i \geq w_i$.
    Suppose for the sake of finding a contradiction that $q^\star_i(v_i,\bm{v}_{-i}) < q^\star_i(w_i,\bm{v}_{-i})$ for some $\bm v \in \mathcal V$, $w_i \in \mathcal V_i$ such that $v_i \geq w_i$, and consider two cases.

{\it Case 1} ($i \in \minority$): 
If $0 \leq q^\star_i(v_i,\bm{v}_{-i}) < q^\star_i(w_i,\bm{v}_{-i})$, by the definition of $\bm q^\star$, we must have $q_i(w_i,\bm{v}_{-i}) = 1$ or $q_i(w_i,\bm{v}_{-i}) = \frac{\gamma}{1+\gamma}$.
Suppose first that $q^\star_i(w_i,\bm{v}_{-i}) = 1$. By definition of $\bm q^\star$, in this case, $i$ must be the smallest index of the bidders with the highest and non-negative virtual value. 
As $v_i > w_i$ and as $\psi_i$ is non-decreasing by Assumption \ref{assum:regularity}, $\psi_i(v_i) \geq 0$ and $\psi_i(v_i) = \max_{j \in \mathcal{I}} \psi_j(v_j)$. Hence, $q^\star_i(v_i,\bm{v}_{-i}) = 1$ which coincides with $q^\star_i(w_i,\bm{v}_{-i})$, leading to a contradiction. 
Suppose next that $q^\star_i(w_i,\bm{v}_{-i}) = \frac{\gamma}{1+\gamma}$. In this case, $i$ must represent the bidder with the highest virtual value from the minority group, and it must hold that 
\begin{equation*}
    \max_{j \in \majority} \psi_j(v_j) > \psi_i(w_i) \quad \text{and} \quad \max_{j \in \majority} \psi_j(v_j) + \gamma \psi_i(w_i) \geq 0.
\end{equation*}
The second inequality continues to hold when $\psi_i(w_i)$ is replaced by $\psi_i(v_i)$. If the first inequality also remains satisfied after a similar substitution $\psi_i(v_i)$ for $\psi_i(w_i)$, then $q^\star_i(v_i,\bm{v}_{-i}) = \frac{\gamma}{1+\gamma} = q^\star_i(w_i,\bm{v}_{-i})$ resulting in a contradiction. On the other hand,~if $\psi_i(v_i) \geq \max_{j \in \majority} \psi_j(v_j)$, then $i$ must be the bidder with the highest virtual value among all bidders and 
\begin{equation*}
    (1+\gamma) \psi_i(v_i) \geq \max_{j \in \majority} \psi_j(v_j) + \gamma \psi_i(v_i) \geq \max_{j \in \majority} \psi_j(v_j) + \gamma \psi_i(w_i) \geq 0.
\end{equation*}
As a result, $q^\star_i(v_i,\bm{v}_{-i}) = 1 > q^\star_i(w_i,\bm{v}_{-i})$ resulting in another contradiction.

{\it Case 2} ($i \in \majority$): If $0 \leq q^\star_i(v_i,\bm{v}_{-i}) < q^\star_i(w_i,\bm{v}_{-i})$, by the definition of $\bm q^\star$, we must have $q_i(w_i,\bm{v}_{-i}) = \frac{1}{1+\gamma}$. 
It hence follows that
\begin{equation*}
    \psi_i(w_i) \geq \max_{j \in \mathcal{I}{\backslash\{i\}}} \psi_j(v_j) \quad \text{and} \quad 
    \psi_i(w_i) + \gamma \max_{j \in \minority} \psi_j(v_j) \geq 0.
\end{equation*}
Since $v_i > w_i$ and $\psi_i$ is non-decreasing by Assumption \ref{assum:regularity}, $i$ remains the highest majority bidder in scenario $(v_i,\bm{v}_{-i})$, and the above two inequalities continue to hold even if $\psi_i(w_i)$ is replaced by $\psi_i(v_i)$. Hence, it must also hold that $q^\star_i(v_i,\bm{v}_{-i}) = \frac{1}{1+\gamma}$ and a similar contradiction is found.

Constraint \eqref{eq:IR} holds immediately because $q_i(\bm v)v_i - m_i(\bm v) = \int_{0}^{{v}_i} q_i(x, \bm v_{-i})\text{d}x \geq 0 $
for all $i \in \mathcal I$ and $\bm v \in \mathcal V$. 

Constraint \eqref{eq:AF} is also trivially satisfied. Indeed, if a bidder with the highest virtual value is from the minority group, then the item can be allocated to neither the majority group nor any minority bidder with dominated virtual values. On the other hand, if this bidder comes from the majority group, then $\sum_{i\in\mathcal{I}} q_i^\star(\bm{v})$ cannot exceed $\frac{\gamma}{1+\gamma} + \frac{1}{1+\gamma}=1$.

Finally, we prove \eqref{eq:Eq}. Suppose for the sake of contradiction that \eqref{eq:Eq} does not hold for some~$\bm v \in \mathcal V$. Then, it must hold that $\sum_{i \in \mathcal{I}^{\operatorname{maj}}} q^\star_i(\bm{v}) = \frac{1}{1+\gamma}$ and $\sum_{i \in \mathcal{I}^{\min}} q^\star_i(\bm{v}) < \frac{\gamma}{1+\gamma}$, and hence, $q^\star_i(\bm{v}) = 0,\ \forall i \in \mathcal{I}^{\min}$. 
Since the total allocation to the majority group is non-zero, $\min \arg \max_{j \in  \mathcal I} \psi_j(v_{j})\in \mathcal I^\text{maj}$ and $ \psi_i(v_i) + \gamma \max_{j \in \mathcal I^{\operatorname{min}}} \psi_j(v_j) \geq 0$.
Thus, the allocation to the bidder with the largest virtual value in the minority group is $\frac{\gamma}{1+\gamma}$, leading to a contradiction.
The claim thus follows.
\end{proof}

The next result reformulates problem \eqref{Rev-MDP} by leveraging \Cref{lem: monotonicity of allocation} to express the payment rules $m_i$ in terms of the allocation rules $q_i$, and by exploiting Assumption~\ref{assum:regularity} to express the total revenue in terms of the virtual values and allocation outcomes, taking inspiration from \citep[Chapter~5]{krishna2009auction}. This lemma will be useful to prove \Cref{theorem:optimal mechanism REV}.

\begin{lemma}[See, e.g., \cite{krishna2009auction}]\label{lem: reformulation in terms of virtual values}
Problem \eqref{Rev-MDP} is equivalent to
   \begin{equation}\label{eq: reformulation REV-MDP}
    \begin{aligned}
    &\max &&\mathbb E_{\mathbb F}\left[\sum_{i \in \mathcal I} \psi_i(\tilde v_i) q_i(\tilde{\bm v})\right]\\
        &\;\;\operatorname{s.t.} &&\bm q \in \mathcal L(\mathcal V, [0,1]^I), \; \bm m \in \mathcal L(\mathcal V, \mathbb R^I)\\
&&& q_i(v_i, \bm{v}_{-i}) \geq q_i(w_i, \bm{v}_{-i}) \;\;\;\forall i \in \mathcal I, \forall \bm v \in \mathcal V, \forall w_i \in \mathcal V_i : w_i \leq v_i\\
&&& m_i(\bm v) = q_i(\bm v)v_i - \int_{0}^{{v}_i} q_i(x, \bm v_{-i})\text{d}x \;\;\;\forall i \in \mathcal I, \forall \bm v \in \mathcal V\\
&&&\eqref{eq:AF},\,\eqref{eq:Eq}.
    \end{aligned}
\end{equation}
\end{lemma}

\begin{proof}{Proof of Lemma \ref{lem: reformulation in terms of virtual values}}
   By utilizing Lemma \ref{lem: monotonicity of allocation} and noting that $m_i(0, \bm v_{-i}) = 0$ holds for all $\bm v \in  \mathcal{V}$ and $i \in \mathcal I$ at optimality, we can reformulate problem \eqref{Rev-MDP} equivalently as follows. 
\begin{equation}\label{eq: reformulation step 1 REV-MDP}
    \begin{aligned}
    &\max &&\mathbb E_{\mathbb F}\left[\sum_{i \in \mathcal I} \left(q_i(\tilde{\bm v})\tilde{v}_i - \int_{0}^{\tilde{v}_i} q_i(x, \tilde{\bm v}_{-i})\text{d}x \right) \right]\\
        &\;\;\text{s.t.} &&\bm q \in \mathcal L(\mathcal V, [0,1]^I), \; \bm m \in \mathcal L(\mathcal V, \mathbb R^I)\\
&&& q_i(v_i, \bm{v}_{-i}) \geq q_i(w_i, \bm{v}_{-i}) \;\;\;\forall i \in \mathcal I, \forall \bm v \in \mathcal V, \forall w_i \in \mathcal V_i : w_i \leq v_i\\
&&& m_i(\bm v) = q_i(\bm v)v_i - \int_{0}^{{v}_i} q_i(x, \bm v_{-i})\text{d}x \;\;\;\forall i \in \mathcal I, \forall \bm v \in \mathcal V\\
&&&\eqref{eq:AF},\,\eqref{eq:Eq}
    \end{aligned}
\end{equation}
Let $f$ denote the joint density function of $\tilde{\bm{v}}$. The objective function of~\eqref{eq: reformulation step 1 REV-MDP} can then be expressed in terms of \( f \) and simplified as follows:
\begin{equation}\label{eq: reformulation step 1}
\allowdisplaybreaks
    \begin{aligned}
    &\mathbb E_{\mathbb F}\left[\sum_{i \in \mathcal I} \left(q_i(\tilde{\bm v})\tilde{v}_i - \int_{0}^{\tilde{v}_i} q_i(x, \tilde{\bm v}_{-i})\text{d}x \right) \right]\\
    &= \sum_{i \in \mathcal I} \left(\int_{\mathcal{V}_{-i}} \int_{\mathcal V_i} q_i({\bm v}){v}_i f(\bm v) \text{d}v_i \text{d}\bm v_{-i} -  \int_{\mathcal{V}_{-i}} \int_{\mathcal V_i} \int_{0}^{{v}_i} q_i(x, {\bm v}_{-i})\text{d}x f(\bm v) \text{d}v_i \text{d}\bm v_{-i} \right)\\
    &= \sum_{i \in \mathcal I} \left(\int_{\mathcal{V}_{-i}} \int_{\mathcal V_i} q_i({\bm v}){v}_i f(\bm v) \text{d}v_i \text{d}\bm v_{-i} -  \int_{\mathcal{V}_{-i}} \int_{\mathcal V_i} \int_{0}^{x} q_i(v_i, {\bm v}_{-i})\text{d}v_i f(x,\bm v_{-i}) \text{d}x \text{d}\bm v_{-i} \right)\\
    &= \sum_{i \in \mathcal I} \left(\int_{\mathcal{V}_{-i}} \int_{\mathcal V_i} q_i({\bm v}){v}_i f(\bm v) \text{d}v_i \text{d}\bm v_{-i} -  \int_{\mathcal{V}_{-i}} \int_{\mathcal V_i}  q_i(v_i, {\bm v}_{-i}) \left(\int_{v_i}^{\overline{v}_i} f(x, \bm v_{-i}) \text{d}x\right) \text{d}v_i \text{d}\bm v_{-i} \right)\\
    &= \sum_{i \in \mathcal I} \int_{\mathcal{V}_{-i}} \int_{\mathcal V_i} \left(v_i - \frac{\int_{v_i}^{\overline{v}_i} f(x, \bm v_{-i}) \text{d}x}{f(\bm v)} \right) q_i({\bm v}) f(\bm v) \text{d}v_i \text{d}\bm v_{-i},
    \end{aligned}
\end{equation}
where the second equality follows from relabeling the variable \(x\) as \(v_i\) and vice versa, and the third equality results from swapping the order of integration over \(x\) and \(v_{i}\), which is justified by Fubini's theorem.

Next, we recall the definition of the virtual values
\begin{equation*}
    \begin{aligned}
    \psi_i(v_i) = v_i - \frac{1 - \int_{0}^{{v}_i} f_i(x) \text{d}x}{f_i(v_i)} 
    = v_i - \frac{\int_{v_i}^{\overline{v}_i} f_i(x) \text{d}x}{f_i(v_i)} 
    = v_i - \frac{\int_{v_i}^{\overline{v}_i} f(x, \bm v_{-i}) \text{d}x}{f(\bm v)} \;\;\;\forall \bm v \in \mathcal V,
    \end{aligned}
\end{equation*}
where the transition from marginal density function $f_i$ to the joint density function $f$ in the rightmost equation is valid when the bidder's values are independent.
Using this observation and the equation \eqref{eq: reformulation step 1}, we can finally re-express the objective of~\eqref{eq: reformulation step 1 REV-MDP} as
\begin{equation*}
    \begin{aligned}
    \sum_{i \in \mathcal I} \int_{\mathcal{V}_{-i}} \int_{\mathcal V_i} \psi_i(v_i) q_i({\bm v}) f(\bm v) \text{d}v_i \text{d}\bm v_{-i} = \mathbb E_{\mathbb F}\left[\sum_{i \in \mathcal I} \psi_i(\tilde{v}_i) q_i(\tilde{\bm v})\right],
    \end{aligned}
\end{equation*}
which completes the proof. 
\end{proof}

\begin{proof}{Proof of Theorem \ref{theorem:optimal mechanism REV}}
Since~\eqref{Rev-MDP} and~\eqref{eq: reformulation REV-MDP} are equivalent by Lemma~\ref{lem: reformulation in terms of virtual values}, it suffices to prove the optimality of $(\bm q^\star, \bm m^\star)$ in~\eqref{eq: reformulation REV-MDP}. 
To this end, first, we relax the monotonicity constraints of $q_i$ in $v_i$ in~\eqref{eq: reformulation REV-MDP} and consider the relaxed problem
\begin{equation} \label{Rev-MDP-relaxed}
\begin{aligned}
    \max_{\bm{q} \in \mathcal L(\mathcal V, [0,1]^I) } \left\{ \mathbb E_{\mathbb F}\left[\sum_{i \in \mathcal I} \psi_i(\tilde v_i) q_i(\tilde{\bm v})\right]
    : \eqref{eq:AF},\,\eqref{eq:Eq} \right\}.
\end{aligned}
\end{equation}
Note that the constraints related to the payment rule $\bm m$ in~\eqref{eq: reformulation REV-MDP} can be omitted, allowing the problem to be solved only for the allocation rule $\bm q$ since $\bm m$ does not appear in any other constraints or the objective function. The optimal payment rule can be constructed after solving the problem by using the optimal allocation rule and setting $m_i(0, \bm v_{-i}) = 0$ for all $\bm v \in  \mathcal{V}$ and $i \in \mathcal I$ (see Lemma~\ref{lem: monotonicity of allocation}).
Our strategy is to show that $\bm{q}^\star$ is optimal in~\eqref{Rev-MDP-relaxed}. Given that $(\bm{q}^\star,\bm{m}^\star)$ is incentive compatible (see Proposition~\ref{prop:feasibility of stochastic case}), it would then automatically hold that $\bm{q}^\star$ satisfies the omitted monotonicity constraints corresponding to \eqref{eq:IC} (see Lemma~\ref{lem: monotonicity of allocation}), and hence $\bm{q}^\star$ would then be optimal not only in~\eqref{Rev-MDP-relaxed} but also in~\eqref{eq: reformulation REV-MDP} and in~\eqref{Rev-MDP}. 

To achieve this, we first observe that the objective and the constraints of~\eqref{Rev-MDP-relaxed} are decoupled over~$\bm{v} \in \mathcal{V}$. Solving~\eqref{Rev-MDP-relaxed} is thus tantamount to solving the optimization problem 
\begin{equation}\label{lp-I}
    \max_{q(\bm{v}) \in [0,1]^I} \left\{ \sum_{i \in \mathcal I} \psi_i(v_i)q_i(\bm{v}): \sum_{i \in \mathcal{I}} q_i(\bm{v}) \leq 1, \sum_{i \in \mathcal I^{\operatorname{min}}} q_i(\bm v) \geq \gamma \sum_{i \in \mathcal I^{\operatorname{maj}}} q_i(\bm v) \right\},
\end{equation}
which is parametrized in $\bm{v} \in \mathcal{V}$ and is a linear program consisting of $I$ decision variables. 
Introducing aggregate allocation variables for each group of bidders: $\tau^{\operatorname{min}}(\bm{v}) = \sum_{i \in \mathcal{I}^{\min}} q_i(\bm{v})$ and $\tau^{\operatorname{maj}}(\bm{v}) = \sum_{i \in \mathcal{I}^{\operatorname{maj}}} q_i(\bm{v})$, we can relax~\eqref{lp-I} one step further as another parametrized linear program in $\bm{v} \in \mathcal{V}$:
\begin{equation}\label{lp-2}
    \max_{\tau^{\operatorname{min}}(\bm{v}),\tau^{\operatorname{maj}}(\bm{v}) \geq 0} \left\{ \tau^{\operatorname{min}}(\bm{v}) \max_{i \in \mathcal{I}^{\operatorname{min}}} \psi_i(v_i) + \tau^{\operatorname{maj}}(\bm{v}) \max_{i \in \mathcal{I}^{\operatorname{maj}}} \psi_i(v_i): 
    \begin{array}{l} 
        \tau^{\operatorname{min}}(\bm{v}) + \tau^{\operatorname{maj}}(\bm{v}) \leq 1 \\
        \tau^{\operatorname{min}}(\bm{v}) \geq \gamma \tau^{\operatorname{maj}}(\bm{v})
    \end{array}
    \right\}
\end{equation}
Problem~\eqref{lp-2} consists of only two variables and admits a vertex solution, that is, the optimal solution belongs to the set $\left\{ (0,0), (1,0), (\frac{\gamma}{1+\gamma},\frac{1}{1+\gamma}) \right\}$ of vertices. When the highest minority virtual value $\max_{i \in \mathcal{I}^{\operatorname{min}}} \psi_i(v_i)$ surpasses the highest majority virtual value $\max_{i \in \mathcal{I}^{\operatorname{min}}} \psi_i(v_i)$ and is non-negative, $(1,0)$ is optimal. On the other hand, if the highest majority virtual value dominates and if $\max_{i \in \mathcal{I}^{\operatorname{maj}}} \psi_i(v_i) + \gamma \max_{i \in \mathcal{I}^{\operatorname{min}}} \psi_i(v_i) \geq 0$, then $(\frac{\gamma}{1+\gamma},\frac{1}{1+\gamma})$ is optimal. Otherwise, the optimal solution of~\eqref{lp-2} is $(0,0)$. It turns out that~\eqref{lp-2} is in fact a tight relaxation because its optimal objective value can be attained by $\bm{q}^\star(\bm{v})$ in~\eqref{lp-I}. As a result, for each $\bm{v} \in \mathcal{V}$, $\bm{q}^\star(\bm{v})$ is optimal in~\eqref{lp-I}, and therefore $\bm{q}^\star$ itself is optimal in~\eqref{Rev-MDP-relaxed},~\eqref{eq: reformulation REV-MDP}, and~\eqref{Rev-MDP}. 
\end{proof}

\section{Supplementary Numerical Results}\label{appendix: experiment}
In this section, we repeat the experiment from Section~\ref{sec: Numerical Experiment} for $\rho$ values from $\{0.5, -0.5\}$ and show the results in Figures~\ref{fig:pos_rho} and~\ref{fig:neg_rho}, respectively. We observe that the results are similar to the case where $\rho = 0$ (see Figure~\ref{fig:zero_rho}), where the regret-based mechanism $(\hat{\bm{q}}, \hat{\bm{m}})$ exhibits the most stable performance across different contamination levels and support configurations.
Notably, numerical results indicate that the revenue of the regret-based mechanism increases with the contamination level for all $\rho = -0.5,0,0.5$.

\begin{figure}[!htbp]
    \centering
    \includegraphics[scale=0.30]{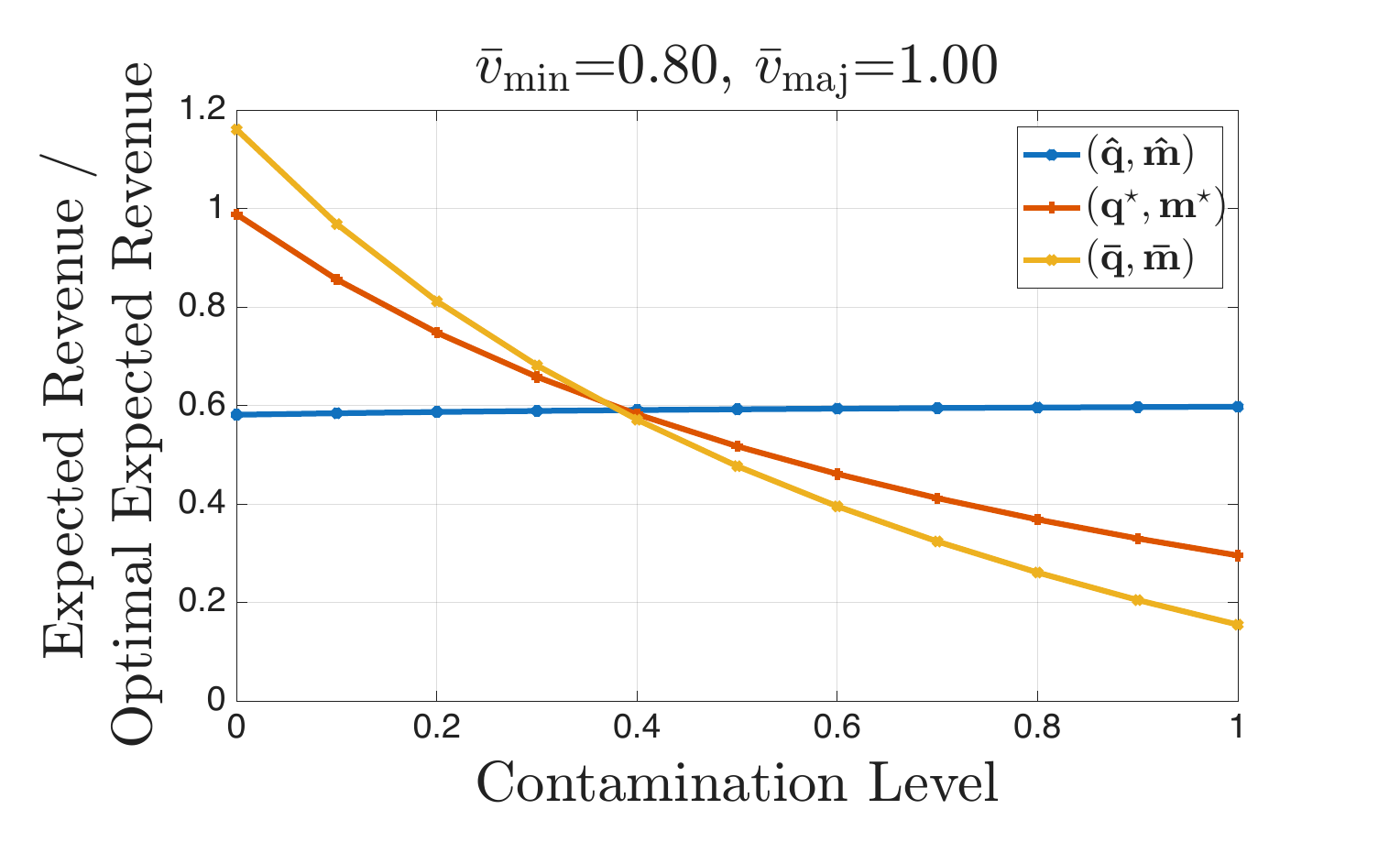}
    \includegraphics[scale=0.30]{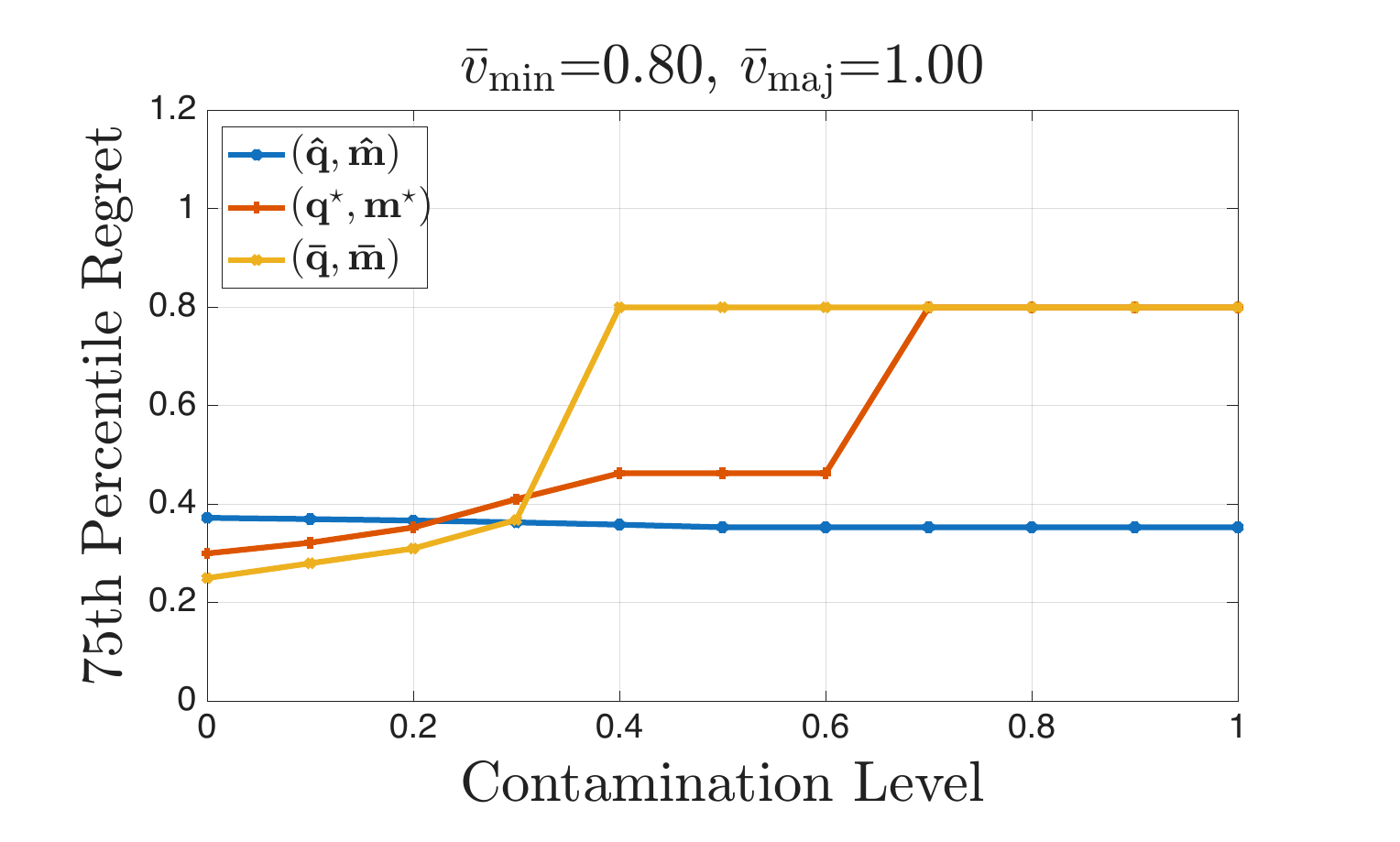}
    
    \includegraphics[scale=0.30]{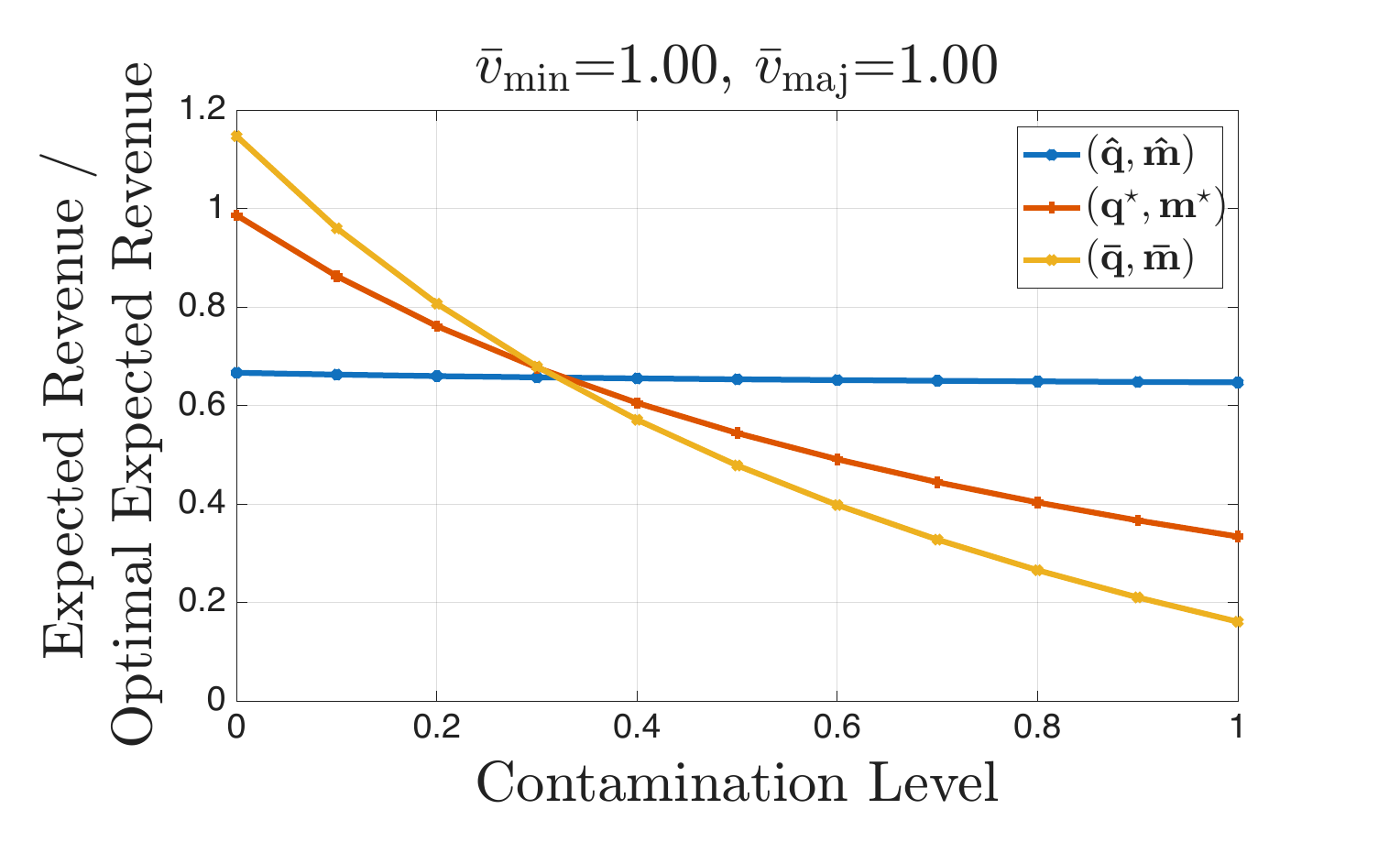}
    \includegraphics[scale=0.30]{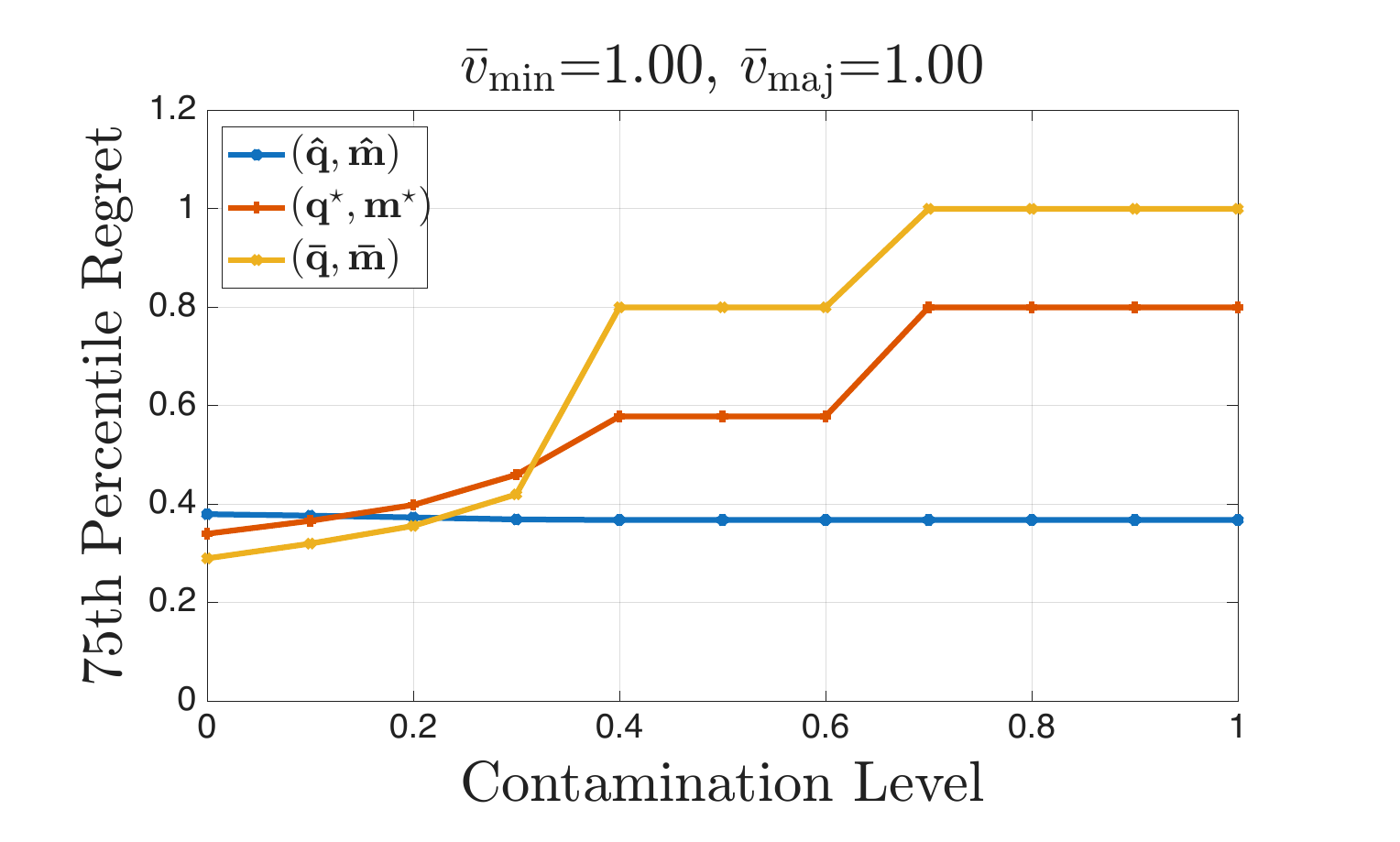}
    
    \includegraphics[scale=0.30]{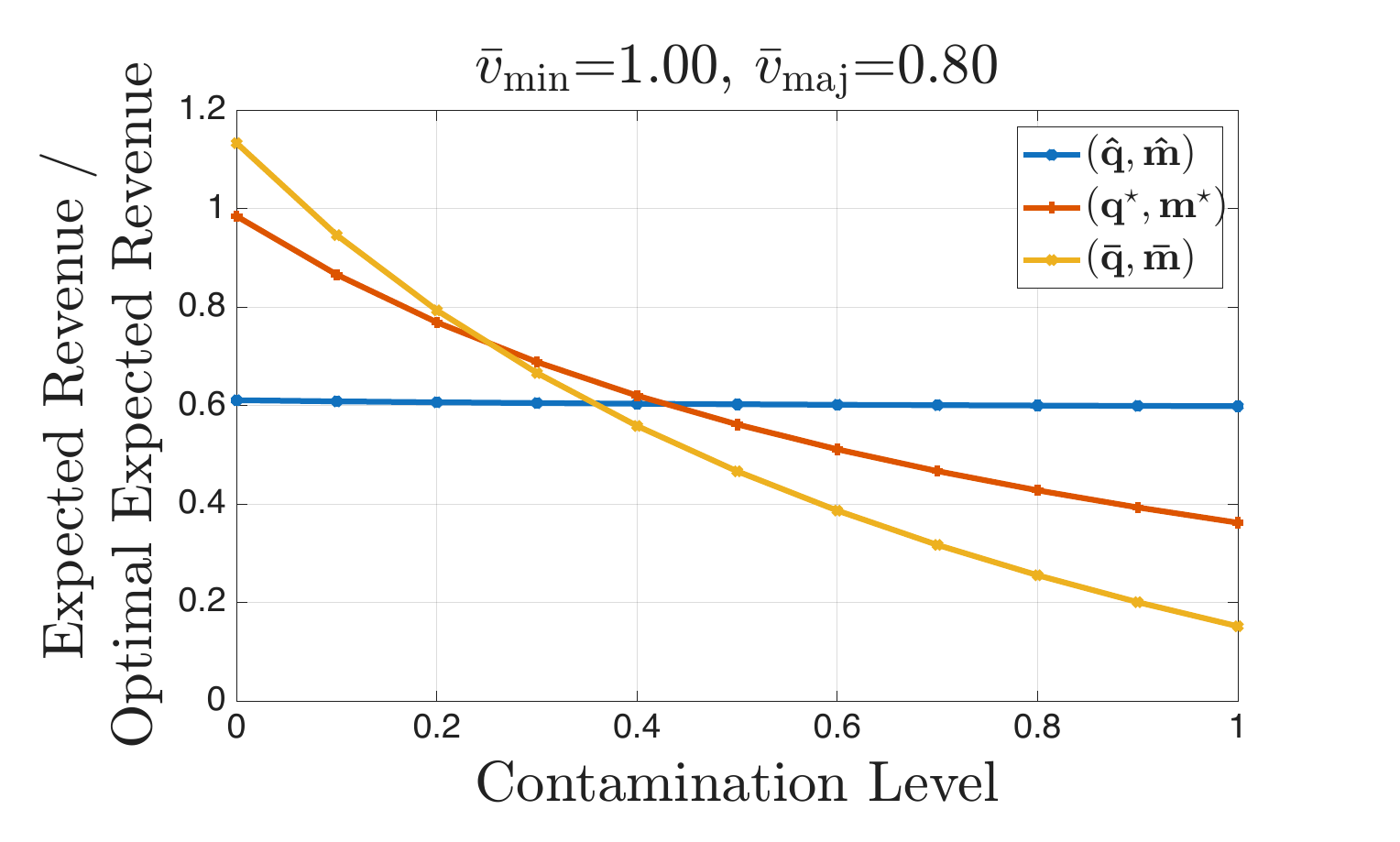}
    \includegraphics[scale=0.30]{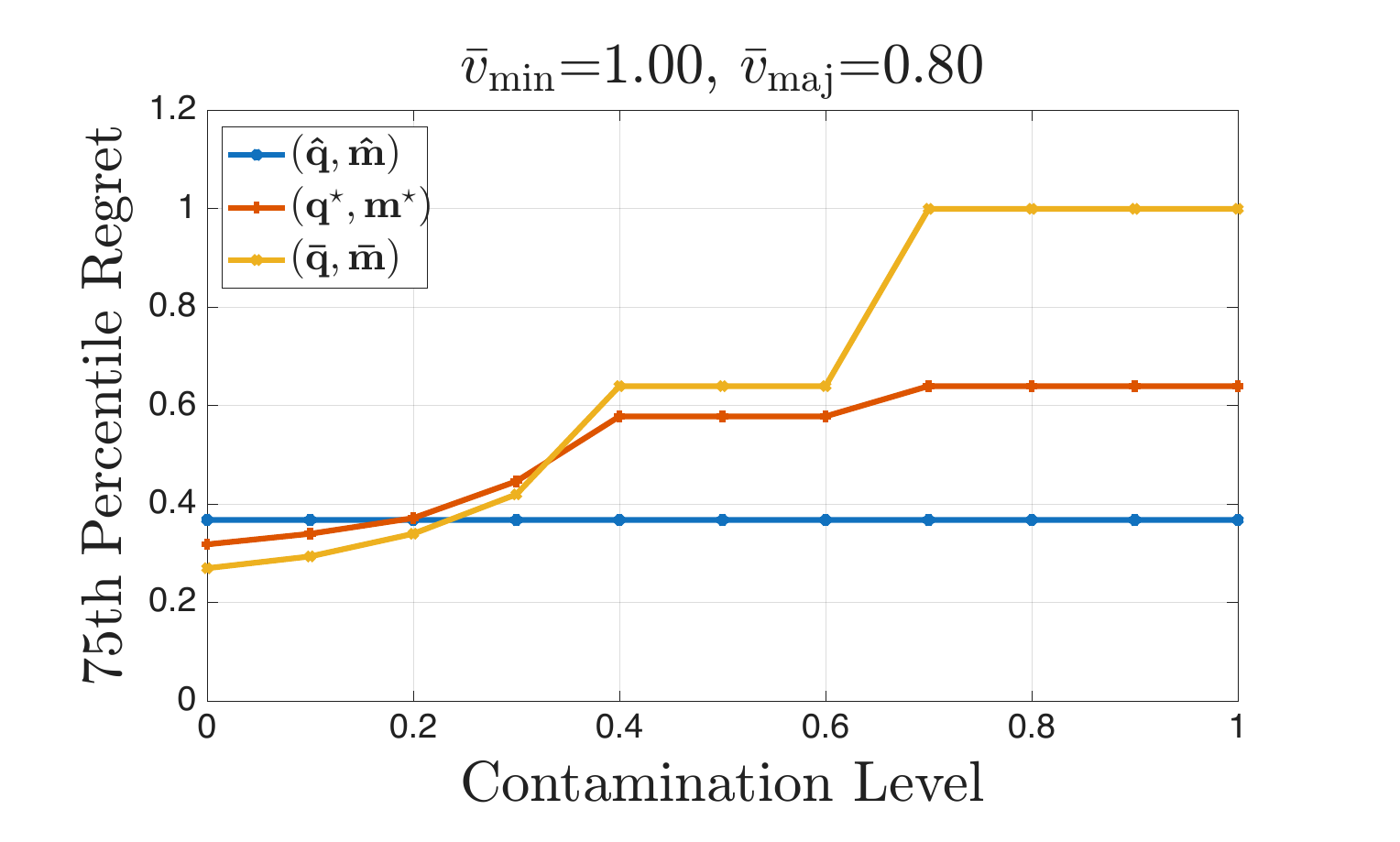}
    \caption{Normalized expected revenues (left) and the upper quartile regrets (right) of the three mechanisms when $\rho = -0.5$ and $(\overline{v}^{\operatorname{min}},\overline{v}^{\operatorname{maj}}) = (0.8,1),(1,1),(1,0.8)$: the regret-based mechanism $(\hat{\bm{q}},\hat{\bm{m}})$ (blue), the revenue-based mechanism $(\bm{q}^\star,\bm{m}^\star)$ (red) and its variant $(\overline{\bm{q}}, \overline{\bm{m}})$ that enforces equity only in expectation (yellow).}
    \label{fig:pos_rho}
\end{figure}

\begin{figure}[!htbp]
    \centering
    \includegraphics[scale=0.30]{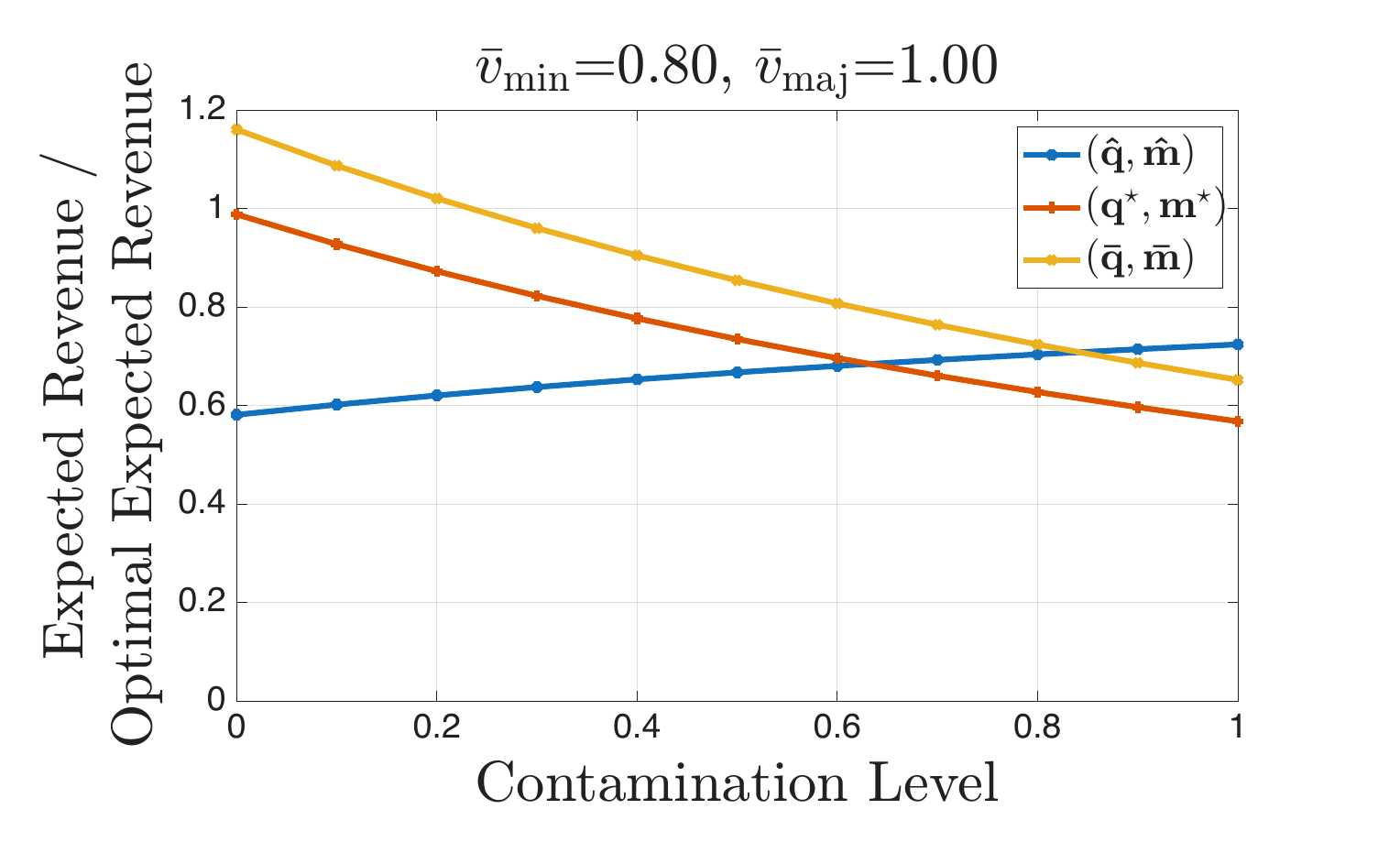}
    \includegraphics[scale=0.30]{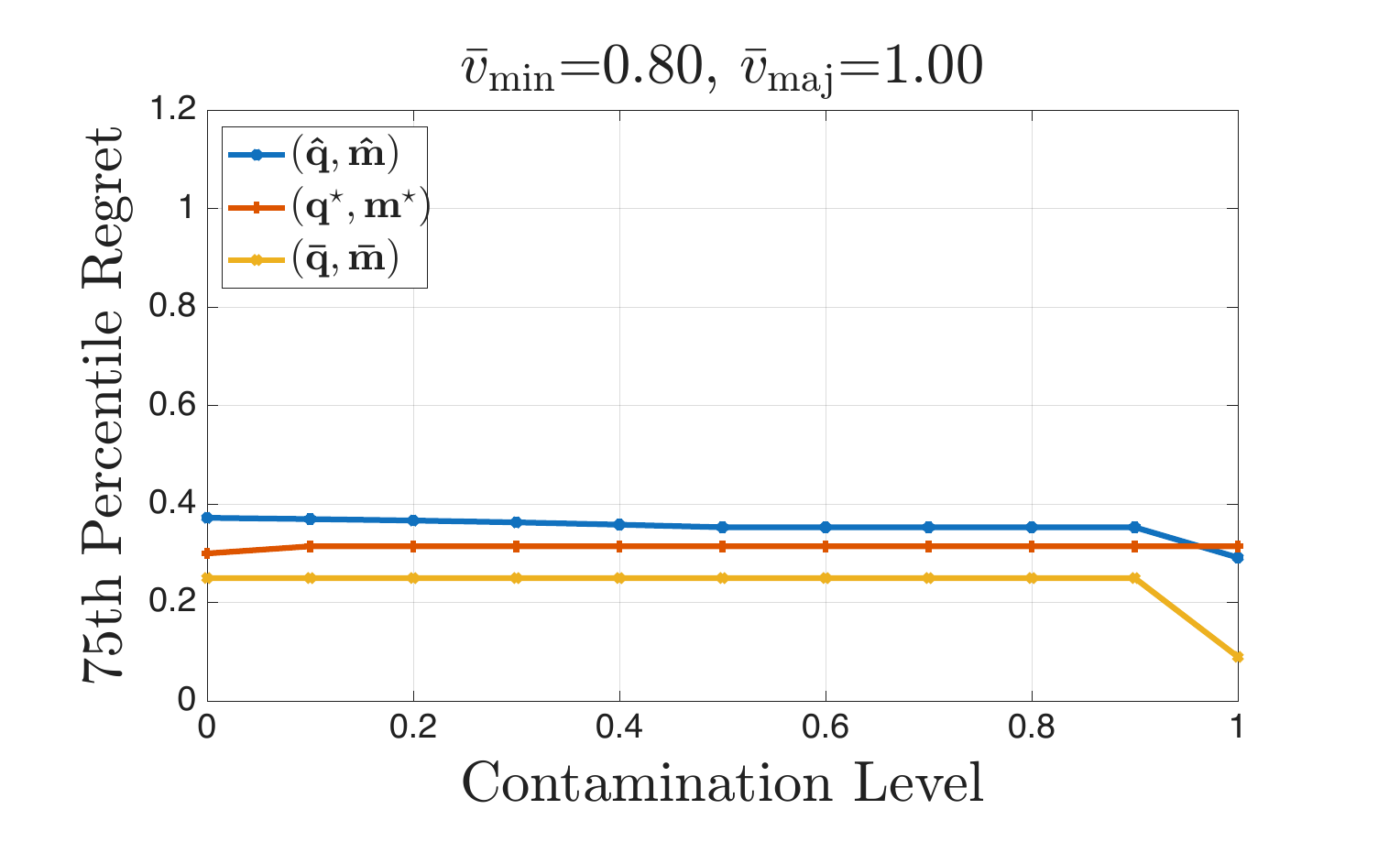}
    
    \includegraphics[scale=0.30]{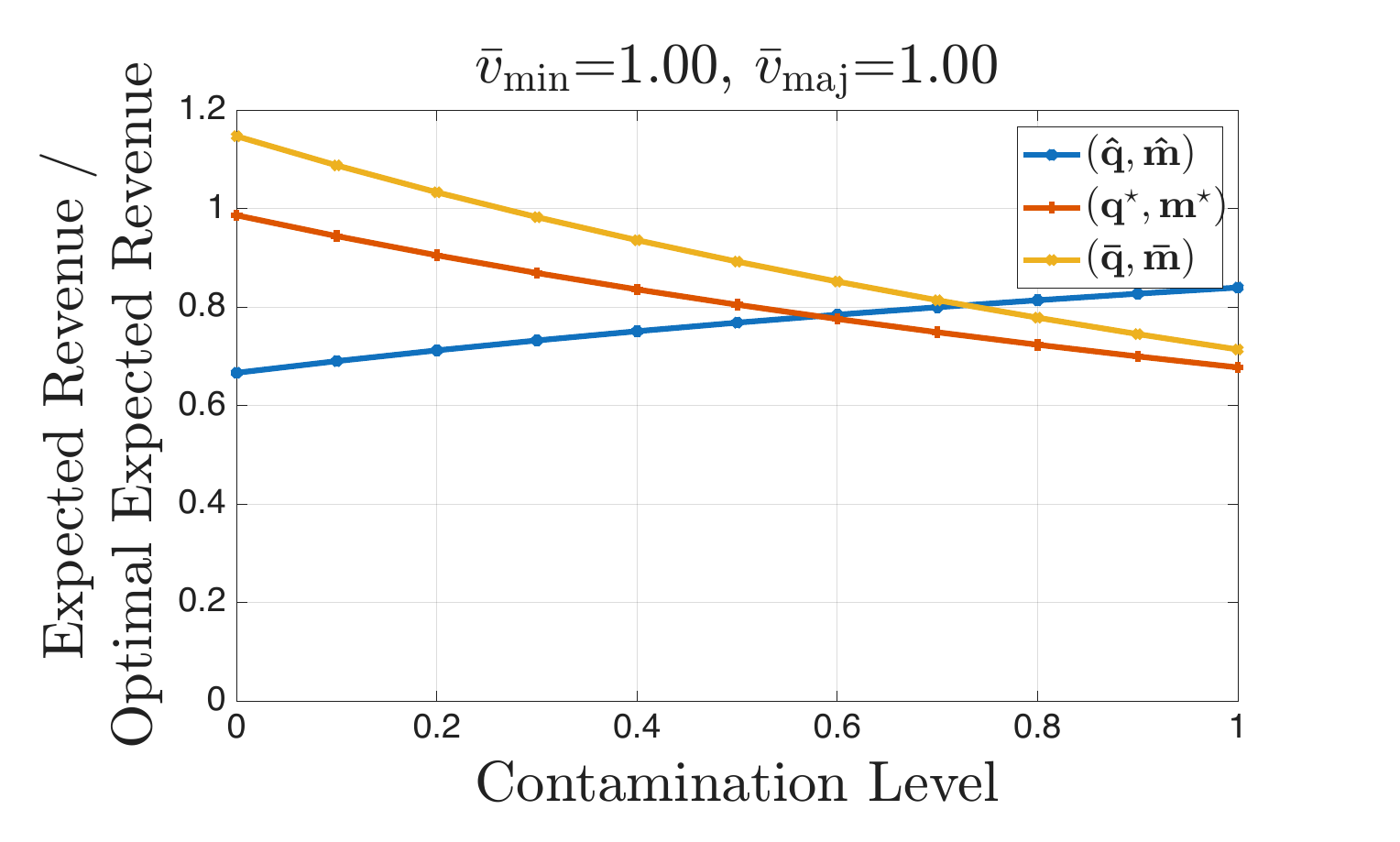}
    \includegraphics[scale=0.30]{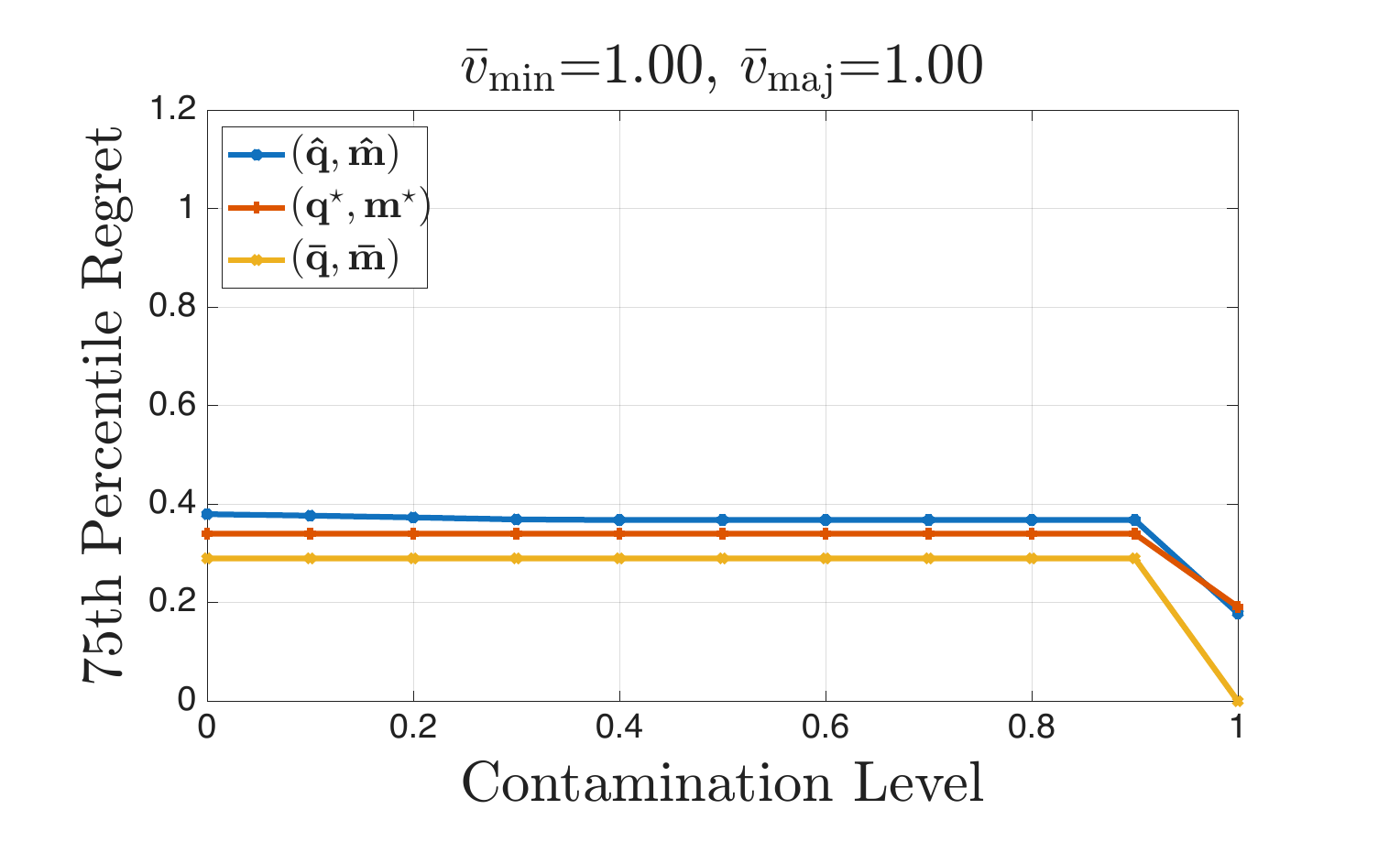}
    
    \includegraphics[scale=0.30]{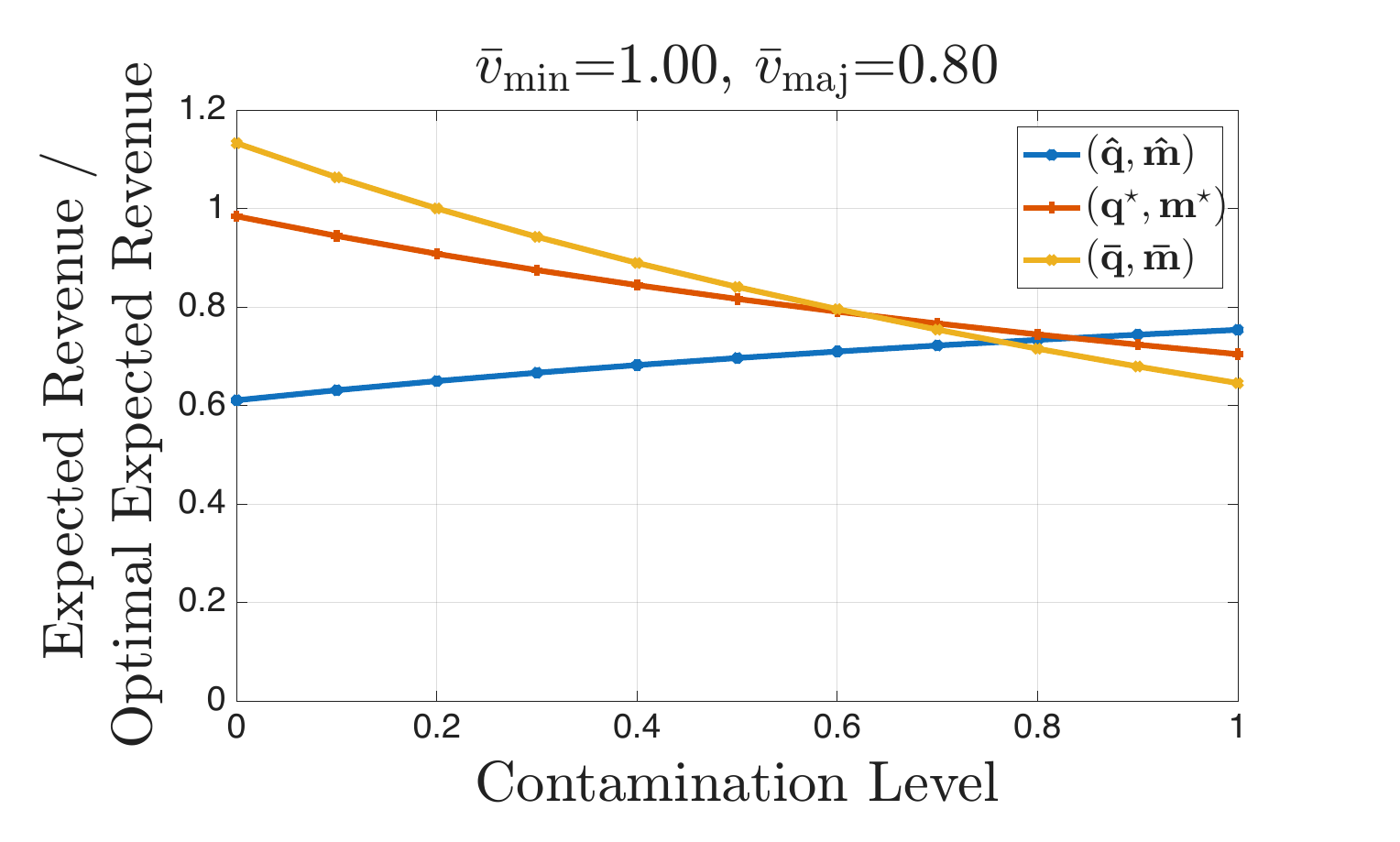}
    \includegraphics[scale=0.30]{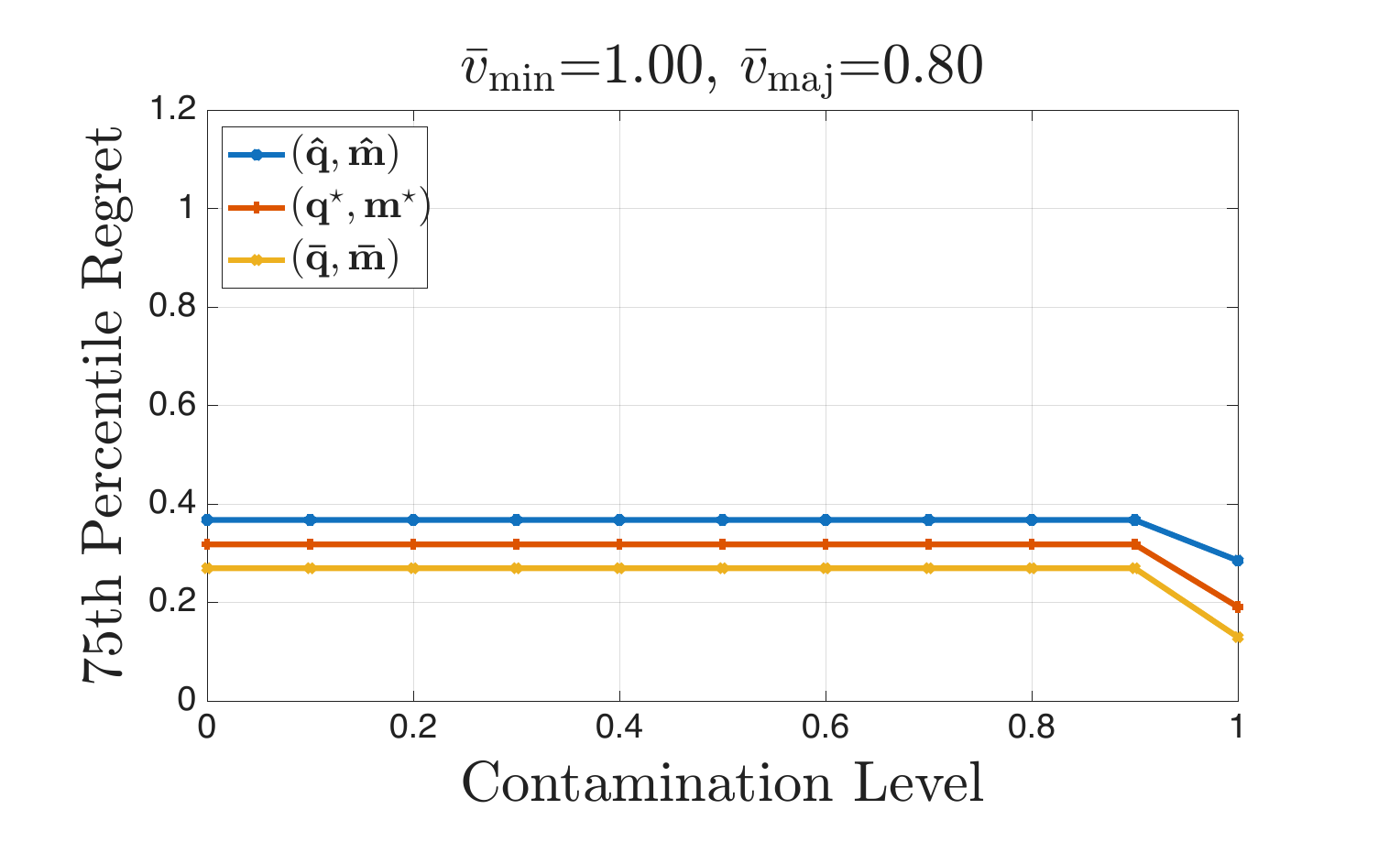}
    \caption{Normalized expected revenues (left) and the upper quartile regrets (right) of the three mechanisms when $\rho = 0.5$ and $(\overline{v}^{\operatorname{min}},\overline{v}^{\operatorname{maj}}) = (0.8,1),(1,1),(1,0.8)$: the regret-based mechanism $(\hat{\bm{q}},\hat{\bm{m}})$ (blue), the revenue-based mechanism $(\bm{q}^\star,\bm{m}^\star)$ (red) and its variant $(\overline{\bm{q}}, \overline{\bm{m}})$ that enforces equity only in expectation (yellow).}
    \label{fig:neg_rho}
\end{figure}

\end{document}